\documentclass[final,1p]{elsarticle}
 \usepackage{etoolbox}
\makeatletter
\patchcmd{\ps@pprintTitle}{\footnotesize\itshape
       Preprint submitted to \ifx\@journal\@empty Elsevier
      \else\@journal\fi\hfill\today}{\relax}{}{}
\makeatother
\makeatletter
\def\ps@pprintTitle{%
 \let\@oddhead\@empty
 \let\@evenhead\@empty
 \def\@oddfoot{\centerline{\thepage}}%
 \let\@evenfoot\@oddfoot}
\makeatother

\usepackage{amssymb}
\usepackage{amsmath}
\usepackage{color}
\newcommand\scalemath[2]{\scalebox{#1}{\mbox{\ensuremath{\displaystyle #2}}}}

 \usepackage{amsthm}
 \newtheorem{thm}{Theorem}[section]
 \newtheorem*{thmm}{Theorem}
 \newtheorem{lem}{Lemma}[section]
 \newtheorem{cor}{Corollary}[section]
 \newdefinition{rmk}{Remark}[section]
  \newdefinition{defe}{Definition}[section]
 \newproof{pf}{Proof}
 \newproof{pot}{Proof of Theorem \ref{thm2}}

 \usepackage{caption}
\usepackage{subcaption}

\journal{Elsevier}

\begin{document}

\begin{frontmatter}



\title{Splines and Wavelets on Circulant Graphs}


\author[MSK]{M. S. Kotzagiannidis\corref{cor1}\fnref{label2}}
\ead{madeleine.kotzagiannidis@ed.ac.uk }

\address[MSK]{Institute for Digital Communications, The University of Edinburgh, King's Buildings, Thomas Bayes Road, Edinburgh EH9 3FG, UK}
\address[PLD]{Department of Electrical and Electronic Engineering, Imperial College London, London SW7 2AZ, UK}
\cortext[cor1]{Corresponding author}
\fntext[label2]{The work in this paper was carried out while the first author was a PhD student at the second author's institution.}
\author[PLD]{P. L. Dragotti}

\begin{abstract}
 We present novel families of wavelets and associated filterbanks for the analysis and representation of functions defined on circulant graphs. In this work, we leverage the inherent vanishing moment property of the circulant graph Laplacian operator, and by extension, the e-graph Laplacian, which is established as a parameterization of the former with respect to the degree per node, for the design of vertex-localized and critically-sampled higher-order graph (e-)spline wavelet filterbanks, which can reproduce and annihilate classes of (exponential) polynomial signals on circulant graphs. In addition, we discuss similarities and analogies of the detected properties and resulting constructions with splines and spline wavelets in the Euclidean domain. Ultimately, we consider generalizations to arbitrary graphs in the form of graph approximations, with focus on graph product decompositions. In particular, we proceed to show how the use of graph products facilitates a multi-dimensional extension of the proposed constructions and properties.
\end{abstract}


\end{frontmatter}


\section{Introduction}
\label{}
There exists a certain fascination with the idea of transferring fundamental signal processing insights to the higher-dimensional domain of graphs and implications thereof. Concurrently, the breadth of emerging applications, in light of the availability of large complex data, arising from i.a. social or biological information structures, has created a need for advanced representation and processing techniques. \\
Notably, the appeal of operating with respect to data encapsulated within the higher-dimensional dependency structures of a graph lies not only in the potential for superior data processing for real-world applications, but also becomes apparent in the development of a corresponding mathematical framework, which seeks to extend conventional signal processing properties to the graph domain, thus naturally challenging the structural confinement of existing frameworks and posing intriguing new questions. \\
The breadth of contributions towards the aforementioned problem statement encompass the novel field of Graph Signal Processing (GSP), having predominantly evolved from two different model assumptions: the collective of works originating from spectral graph theory (\cite{shu}, \cite{chung}), with the graph Laplacian matrix as the central operator on the one hand, and the more generalized setting with focus on the graph adjacency matrix, expanding on concepts from i.a. algebraic signal processing \cite{moura}, on the other hand.\\
The notion of wavelets on graphs, in particular, presents a promising avenue to facilitate sophisticated processing of complex data, which is captured in the graph signal and underlying graph, beyond classical wavelet theory, due to the potential to operate with respect to the inherent geometry of the data in a more localised manner.
A range of designs have been proposed, including the diffusion wavelet \cite{Coifman}, the biorthogonal and perfect reconstruction filterbank on bipartite graphs (\cite{ortega2}, \cite{ortega3}), and the spectral graph wavelet \cite{spectral}, 
tailored to satisfy a set (or subset) of properties, which evolved from the traditional domain, such as localization in the vertex or spectral graph domain, critical sampling and invertibility, along with notions of graph-specific downsampling and graph-coarsening for a multiscale representation, as well as to facilitate generalizations to arbitrary graphs, for applications such as image processing \cite{shu}. In particular, sparsity on graphs via wavelet analysis would appear as a natural extension to its foundation in the discrete-time domain, and some works have opened its discussion through topics such as the wavelet coefficient decay at small scales of graph-regular signals \cite{ricaud} via the spectral graph wavelet transform, the tight wavelet frame transform on graphs \cite{tight}, as well as the overcomplete Laplacian pyramid transform with a spline-like interpolation step \cite{mult}. These approaches so far lacked a concrete graph wavelet design methodology which targets the annihilation of graph signals, or, alternatively, the identification of classes of graph signals which can be annihilated by existing constructions. \\
\\
With the overall agenda to explore sparsity and sparse representations on circulant graphs, we seek to build a bridge from the Euclidean to the graph domain traversing the topics of graph spline wavelet theory and multi-dimensional signal processing on graphs, while providing an intuition behind the emerging spline-like graph functions as an underlying theme. The appeal of circulant graphs pertains to their convenient set of properties such as Linear Shift Invariance (LSI)\cite{ekambaram1}, which facilitate intuitive downsampling and shifting operations, as well as to their link with the traditional domain of signal processing, where the Graph Fourier Transform (GFT) of a circulant graph is a simple permutation of the DFT. Further, circulant matrices (and hence graphs) can be efficiently stored as they are entirely characterized by only one row. Our main focus lies on undirected circulant graphs, as the associated symmetry yields real-valued filters, however, as will be revealed in Sect. $3$, most properties can be generalized to directed circulant graphs, reaffirming the established similarities with classical splines.\\
\\
In this work, we present novel families of higher-order graph (e-)spline wavelets and associated filterbanks on circulant graphs, inspired by the critically sampled `spline-like' graph wavelet filterbank by Ekambaram et al. \cite{ekambaram2} and the classical (e-)spline wavelets (\cite{spline},\cite{espline}), which can reproduce and annihilate certain classes of graph signals. By leveraging the vanishing moment property of the e-graph Laplacian matrix, as a generalization of the classical graph Laplacian, we discover (e-)spline-like functions, which bear similar properties to
the traditional cases (\cite{splines}, \cite{espline}), and give rise to associated wavelet and filterbank constructions, a subset of which we previously introduced in (\cite{spie}, \cite{icassp}). In particular, we identify the classes of smooth graph signals, which can be annihilated on the vertices of circulant graphs up to a graph-dependent border effect, as (exponential) polynomials.\\
We eventually provide generalizations to the developed theory by proposing to approximate arbitrary graphs as either communities of partitioned circulant sub-graphs or graph (Kronecker) products of circulant graphs, the latter of which serve as the building blocks of a multi-dimensional extension. Here, we resort to employing graph products to both generalize wavelet analysis to arbitrary graphs and facilitate lower-dimensional processing, while inducing and/or preserving sparsity where applicable, and proceed to introduce multi-dimensional separable and non-separable graph wavelet transforms. \\
\\
{\bf Contributions.} We list the main original contributions in this paper as follows:
\begin{enumerate}
\item Novel families of higher-order circulant \textit{graph spline wavelets}, which extend traditional vanishing moment properties to the graph domain, and associated filterbanks
\item Novel families of higher-order circulant \textit{graph e-spline wavelets}, which extend vanishing exponential moment properties to the graph domain, and associated filterbanks
\item \textit{Multi-dimensional separable} and \textit{non-separable} graph spline wavelet transforms, which generalize properties and operations to arbitrary graphs through graph products
\end{enumerate}
{\bf Related Work.} To the best of our knowledge, there do not exist comparable graph wavelet constructions on circulant graphs with the aforementioned reproduction and annihilation properties. While established graph wavelet constructions, such as the spectral or tight graph wavelet (\cite{spectral}, \cite{tight}) may achieve sufficiently sparse representations in the graph wavelet domain, i.a. for appropriate design choices of the associated wavelet kernel, there is no concrete (or intuitive) theory on what types of graph signals can be annihilated, beyond the class of piecewise-constant signals, in particular, using the properties and connectivity of the graph at hand. Sparsity on graphs has been more specifically addressed in dictionary learning on graphs \cite{dict}, which is concerned with the problem of identifying an (overcomplete) basis ${\bf D}$ under which a given graph signal ${\bf y}$ can be sparsely represented as ${\bf y}={\bf D}{\bf x}$. Furthermore, multi-dimensional wavelet analysis has been considered for bipartite graphs as the operation with respect to separate edge sets on the same vertex set within a bipartite subgraph decomposition \cite{ortega2}. In our proposed framework for graph approximation and wavelet analysis via graph product decomposition, the notion of multiple dimensions arises from the graph product operation itself, with each factor constituting a separate dimension. The present work on graph spline wavelet theory further provides the foundation for the sparse graph signal sampling and graph coarsening framework, developed in a complementary manuscript \cite{acha2}. \\  
\\
This paper is organised as follows: we introduce the notation and background theory in Section $2$, and proceed to discuss the collective of derived graph wavelet families as well as set them into the context of classical spline theory in Section $3$. Subsequently, in Section $4$, we expand the framework through the introduction of graph products. In Section $5$, we present concrete examples illustrating the application of derived graph wavelet transforms in both an artificial and data-driven setting, before concluding with an outlook on future work in Section $6$. The appendix contains all proofs not included in the main text.

\section{Preliminaries}
\subsection{Notation}
\noindent We denote vectors with boldfaced lower case letters ${\bf x}$ and matrices with boldfaced uppercase letters ${\bf A}$. Let ${\bf 1}_N$ and ${\bf 0}_N$ define the constant column vectors of length $N$ with entries of 1's and 0's respectively, and $\begin{bmatrix}{\bf x}^B\\{\bf y}^{B^{\complement}}\end{bmatrix}$ be the vector with samples of ${\bf x}$ at positions in index set $B$ and ${\bf y}$ at positions in complement $B^{\complement}$, with partitions ${\bf x}^B$ and ${\bf y}^{B^{\complement}}$ possibly interlacing. The vector and matrix norms, we will most frequently make use of, are the $l_0$-pseudo-norm, denoted with $||{\bf x}||_0=\#\{i:x_i\neq0\}$, the $l_2$-norm, given by $||{\bf x}||_2=\left(\sum_{i=1}^N |x_i|^2\right)^{1/2}$ and the Frobenius-norm of a  matrix ${\bf A}$, given by $||{\bf A}||_F=\sqrt{\sum_{i=1}^m \sum_{j=1}^n |A_{i,j}|^2}$. In addition, we define the Frobenius inner product between matrices ${\bf A}$ and ${\bf B}$ as $\langle {\bf A},{\bf B}\rangle_F=tr({\bf A}^T{\bf B})$. Given a matrix ${\bf L}$ and the sets of indices $A$ and $B$, the notation ${\bf L}(A,B)$ indicates that the corresponding rows and columns in ${\bf L}$ are chosen. At last, we note that, contrary to classical wavelet theory, we reverse the notation of dual pairs, by denoting the synthesis filters with $(\tilde{{\bf G}},\tilde{{\bf H}})$ and the analysis filters with $({\bf G},{\bf H})$.

\subsection{GSP Theory: Circulant Graphs}
For the ensuing discussion, we consider graphs, which are undirected, connected, (un-)weighted, and do not contain any self-loops. Let a graph $G=(V,E)$ be defined by a set $V=\{0,...,N-1\}$ of vertices, with cardinality $|V|=N$, and a set $E$ of edges. The connectivity of $G$ is given via its adjacency matrix ${\bf A}$, with $A_{i,j}>0$ if there is an edge between nodes $i$ and $j$, and $A_{i,j}=0$ otherwise, and its degree matrix ${\bf D}$, which is diagonal with entries $D_{i,i}=\sum_{j} A_{i,j}$. The non-normalized graph Laplacian matrix ${\bf L}={\bf D}-{\bf A}$ of $G$ has a complete set of orthonormal eigenvectors $\{{\bf u}_{l}\}_{l=0}^{N-1}$, with corresponding nonnegative eigenvalues $0= \lambda_0<\lambda_1\leq \dots \leq \lambda_{N-1}$. \\
Circulant graphs represent a special class of graphs as they reveal a set of convenient properties, which can be used for the preservation of traditional signal processing concepts and operations (see Fig. $1$ for examples). In particular, a circulant graph $G$ is defined via a generating set $S=\{s_1,\dots,s_M\}$, with $0<s_k\leq N/2$, whose elements indicate the existence of an edge between node pairs $(i,(i\pm s_k)_N),\forall s_k\in S$, where $()_N$ is the$\mod N$ operation; more intuitively, a graph is circulant if its associated graph Laplacian is a circulant matrix under a particular node labelling \cite{ekambaram1}. Furthermore, we can define the symmetric, circulant graph Laplacian matrix ${\bf L}$, with first row $\lbrack l_0\quad ... \quad l_{N-1}\rbrack$, via its representer polynomial $l(z)=\sum_{i=0}^{N-1} l_i z^i$ with $z^{-j}=z^{N-j}$. In particular, $l(z)$ gives rise to the eigenvalues of ${\bf L}$, as ordered by the DFT-matrix, at frequency locations $\frac{2\pi i k}{N}$ via $l(e^{\frac{2\pi i k}{N}})=\lambda_k,\enskip k=0,...,N-1$ \cite{circul}. 
Let the $M$-connected ring graph $G$, as part of a sub-class of circulant graphs, be defined via the generating set $S=\{1,...,M\}$, such that there exists an edge between nodes $i$ and $j$, if $(i-j)_{N}\leq M$ is satisfied. The associated circulant graph Laplacian matrix is banded of bandwidth $M$. Another relevant class is that of bipartite graphs, which are characterized by a vertex set $V=X\cup Y$ consisting of two disjoint sets $X$ and $Y$, such that no two vertices within the same set are adjacent.\\
In GSP theory, a graph signal ${\bf x}$ is traditionally a real-valued scalar function defined on the vertices of a graph $G$, with sample value $x(i)$ at node $i$, and can be represented as the vector ${\bf x}\in\mathbb{R}^N$ \cite{shu}; in this work, we extend this definition to include complex-valued graph signals ${\bf x}\in\mathbb{C}^N$, for illustration purposes, while maintaining real weights between connections on $G$. The Graph Fourier Transform (GFT) $\hat{{\bf x}}$ of ${\bf x}$ defined on $G$, is the representation in terms of the graph Laplacian eigenbasis ${\bf U}=\lbrack {\bf u}_0| \cdots |{\bf u}_{N-1}\rbrack$ such that $\hat{{\bf x}}={\bf U}^H {\bf x}$, where $H$ denotes the Hermitian transpose, extending the concept of the Fourier transform to the graph domain \cite{shu}. We note that in case of a circulant graph, the GFT can be represented by the traditional DFT, up to a permutation, as circulant matrices are diagonalisable by the DFT-matrix.\\
Due to their regularity, circulant graphs lend themselves for defining meaningful downsampling operations. As established by Ekambaram et al. in \cite{ekambaram1}, one can downsample a given graph signal by $2$ on the vertices of $G$ with respect to any element $s_k\in S$. For simplicity, we resort to the simple downsampling operation with respect to the outmost cycle ($s_1=1$) of a given circulant $G$, i.e. skipping every other labelled node, assuming that the graph at hand is connected such that $s_1\in S$, and $N=2^n$ for $n\in \mathbb{N}$.
\begin{figure}[tbp]
	\centering
	
	{\includegraphics[width=1.3in]{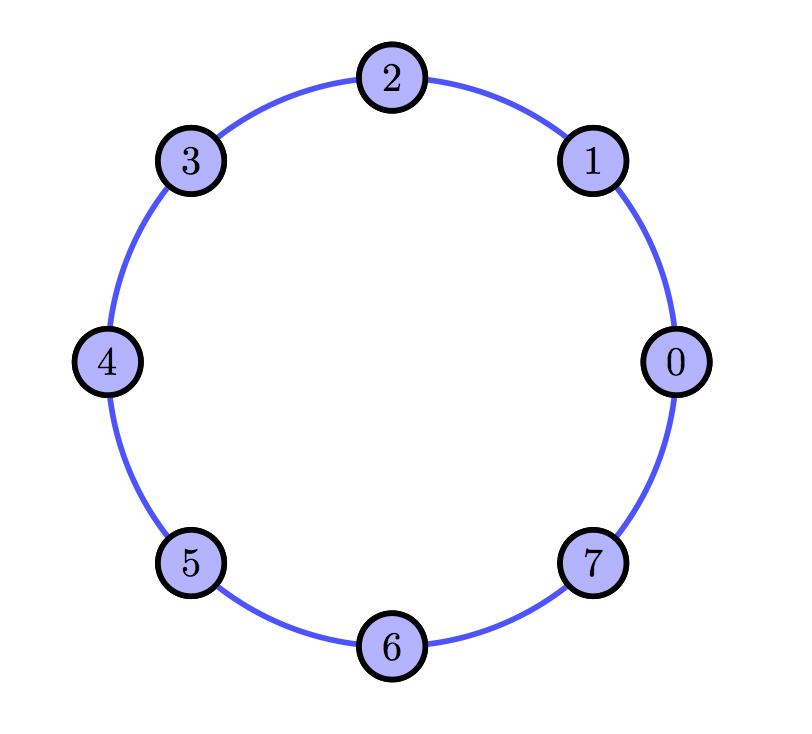}}
	{\includegraphics[width=1.4in]{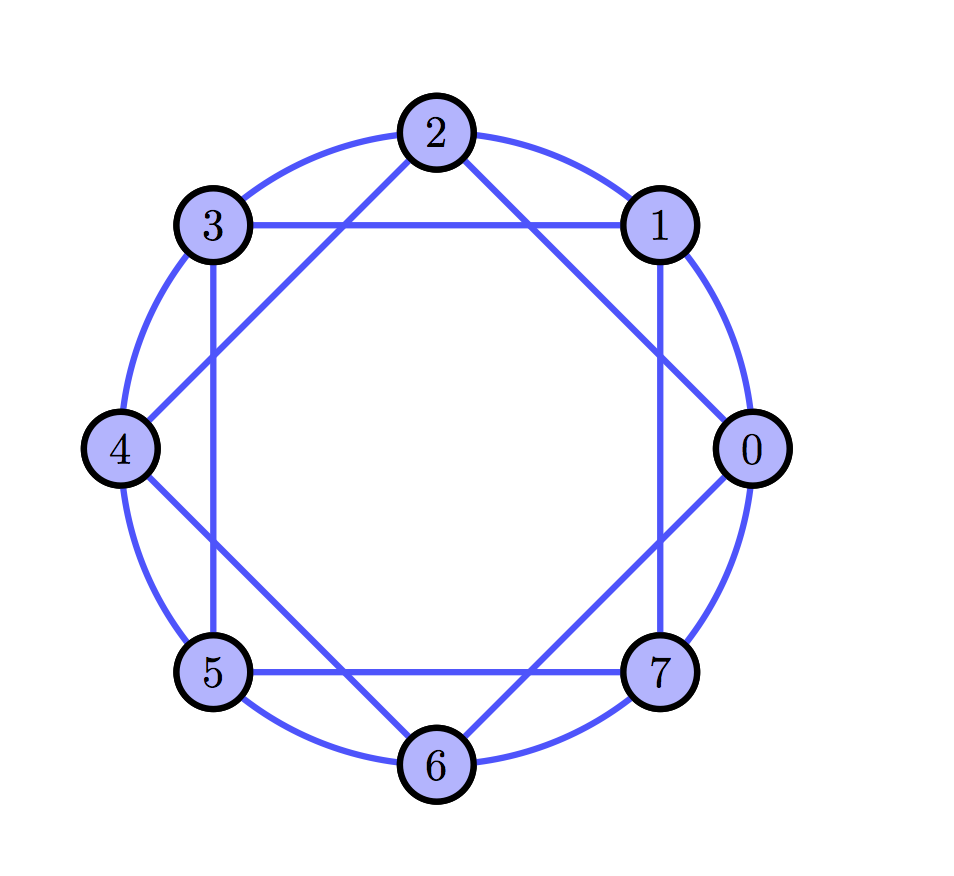}}
	{\includegraphics[width=1.3in]{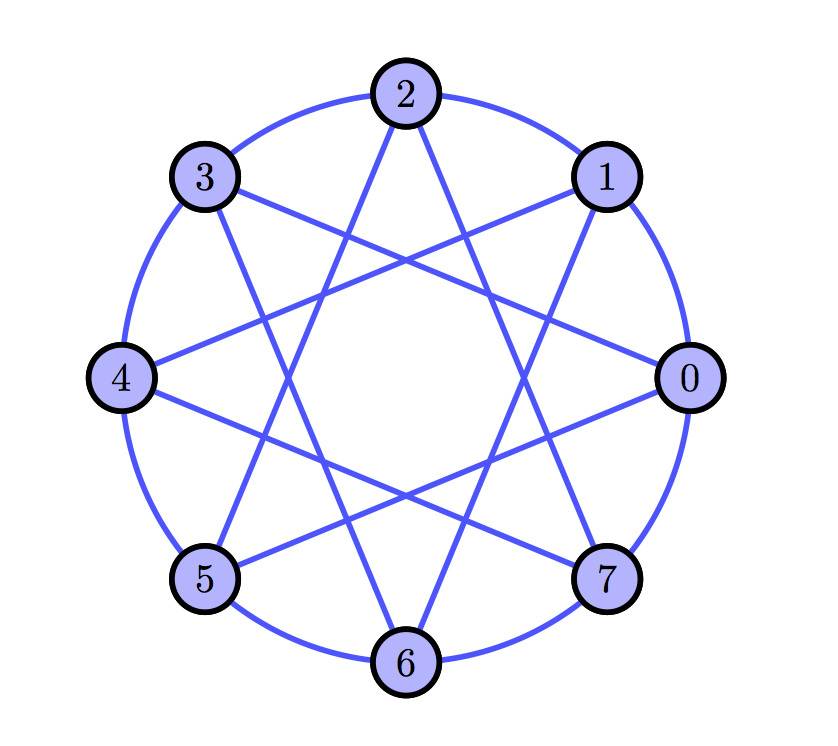}}	
	\caption{Circulant Graphs with generating sets $S=\{1\}$, $S=\{1,2\}$, and $S=\{1,3\}$ (f. left).}
\end{figure}
In addition, the same authors introduced a set of vertex-domain localized filters constituting the `spline-like' graph wavelet filterbank on circulant graphs (\cite{ekambaram2},\cite{Ekambaram3}), which satisfies critical sampling and perfect reconstruction properties:
\begin{thmm} [\cite{ekambaram2}]The set of low-and high-pass filters, defined on an undirected connected circulant graph with adjacency matrix ${\bf A}$ and degree $d$ per node, take (weighted) averages and differences with respect to neighboring nodes at $1$-hop distances of a given graph signal, and can be expressed as:
\begin{equation}\label{eq:eq1}{\bf H}_{LP}=\frac{1}2{}\left({\bf I}_N+\frac{{\bf A}}{d}\right)\end{equation}
\begin{equation}\label{eq:eq2}{\bf H}_{HP}=\frac{1}2{}\left({\bf I}_N-\frac{{\bf A}}{d}\right).\end{equation}
The filterbank is critically sampled and invertible as long as at least one node retains the low-pass component, while the complementary set of nodes retains the high-pass components.
\end{thmm}
\noindent Multiscale analysis can be conducted by iterating the result on the respective downsampled low-pass branches, where corresponding coarsened graphs are obtained through suitable reconnection strategies \cite{Ekambaram3}.
\subsubsection{Downsampling and Reconnection on Circulant Graphs}
Succeeding the definition of a wavelet transform on a circulant graph, we examine the problem of identifying suitable coarsened graph(s) on the vertices of which the downsampled low-and high-pass-representations of the original graph signal can be defined, so as to facilitate a multiresolution decomposition in the graph domain. 
While a downsampling pattern can be easily identified, it is not straight-forward to determine if or how to reconnect the reduced set of vertices. In general, the set of desired properties of a coarsened graph, comprising closure, preservation of the initial connectivity and spectral characterisation of the graph and/or graph type, among others, (see \cite{mult} for a more detailed review), is rather difficult to satisfy entirely, and priorities need to be set in keeping with the overall goal to be achieved. Since we are interested in a sparse graph wavelet representation, and the obtained sparsity $K$ may increase with the bandwidth $M$ of the adjacency matrix, as will be clarified later on, we favor a graph reconnection which reduces or maintains $M$ when conducting multiresolution analysis. Kron-reduction \cite{Kron} is a commonly used method, which employs a sub-matrix approximation scheme, thus taking into account the entire given graph Laplacian matrix for a more accurate dimensionality-reduced graph-representation, yet it often leads to denser graphs (and thus an increased matrix bandwidth) due to its maximum reconnection. In particular, given the graph Laplacian matrix ${\bf L}$ and an index set $V_{\alpha}$ of retained nodes, Kron-reduction evaluates the graph-Laplacian matrix $\tilde{{\bf L}}$ of the coarsened graph as
 \[ \tilde{{\bf L}}={\bf L}(V_{\alpha},V_{\alpha})-{\bf L}(V_{\alpha},V_{\alpha}^{\complement}){\bf L}(V_{\alpha}^{\complement},V_{\alpha}^{\complement})^{-1}{\bf L}(V_{\alpha},V_{\alpha}^{\complement})^T.\]
In the traditional domain, the downsampled signal samples are `reconnected' through a simple stacking operation, whose graph-analogy on a simple cycle would correspond to the reconnection of $2$-hop neighbours, yet a straight-forward graph generalization is hindered by the overall complex connectivity of a graph, making it unclear to which extent one generally needs to reconnect downsampled nodes.
For a multilevel sparse graph wavelet representation, we therefore resort to two variations, both of which preserve circularity with little or no reconnection, to determine the coarse graph $\tilde{G}$: $(1)$ we do not reconnect nodes, and only keep existing edges (with exception of maintaining $s=1\in S$ to ensure that $\tilde{G}$ is connected), $(2)$ we (re-)connect a subset of nodes, such that $\tilde{G}$ is identical in structure to the initial $G$, i.e. it has the same generating set $S$. \\
We note that while the former approach leads to the sparsest possible solution, as given an $M$-connected (banded) circulant graph, we continuously remove edges resulting from odd elements in the generating set, reducing its band, it fails to preserve the global connectivity of the initial graph, which for the objective at hand, we deem of secondary importance. In addition, it produces the trivial simple cycle for bipartite circulant graphs. The latter approach, while leading to a slightly less sparse representation due to the constant bandwidth, reconnects a subset of nodes, which were initially connected via a path, and preserves connectivity through an exact replication in lower dimension. In the complementary work on sampling theory on graphs (\cite{acha2}, Lemma $4.2$), it is further proved that the latter approach specifically preserves spectral graph information.
\section{Families of Spline Wavelets on Circulant Graphs}
The concept of graph-specific smoothness of a signal ${\bf x}$ on a graph $G$ with graph Laplacian matrix ${\bf L}$ has been introduced via the graph Laplacian quadratic form $S_2({\bf x})={\bf x}^T{\bf L}{\bf x}$, and successfully leveraged in i.a. denoising schemes \cite{shu}, however, it provides little information on the actual sparsity of graph signals, as measured per the $l_0$-norm $||{\bf L}{\bf x}||_0$. Therefore, in this work, we choose to adhere to the standard, graph-independent, notion of annihilation of polynomial or exponential signals (for a given labelling), when considering compressibility of graph signals. \\
In classical signal processing, a high-pass filter ${\bf h}$ with taps $h_k$ is known to have $N$ vanishing moments when it is orthogonal with respect to the subspace of polynomials of up to degree $N-1$, that is when 
\[m_n=\sum_{k\in\mathbb{Z}} h_k k^n=0,\quad\text{for}\enskip n=0,...,N-1,\]
where $m_n$ is the $n$-th order moment of ${\bf h}$. Coincidentally, a filter $H(z)$ with $N$ vanishing moments is characterized by $N$ zeros at $z=1$. In \cite{Coifman}, this property is extended to graphs and manifolds by defining the number of vanishing moments of a scaling function as the number of eigenfunctions of the given diffusion operator ${\bf T}$ to which the former is orthogonal, up to a precision measure. Further, in \cite{vanla} the eigenvectors of the graph Laplacian are selected as the basis for `generalized vanishing moments' on graphs. However, this definition of vanishing moments on graphs does not accommodate equivalencies between the graph and traditional domain, when considering for instance discrete periodic (time) signals on a simple cycle.
\\
Contrary to the above, we therefore choose to maintain the traditional definition throughout this work for purposes of illustrating analogies to our developed spline wavelet theory on circulant graphs, yet will briefly revisit this interpretation when discussing the properties and implications of the novel e-graph Laplacian operator. 
\begin{defe}
A graph signal ${\bf p}\in \mathbb{R}^N$ defined on the vertices of a graph $G$ is (piecewise) polynomial if its labelled sequence of sample values, with value $p(i)$ at node $i$, is the discrete, vectorized version of a standard (piecewise) polynomial, such that ${\bf p}=\sum_{j=1}^K {\bf p}_j \circ {\bf 1}_{\lbrack t_j,t_{j+1})}$, where $t_1=0$ and $t_{K+1}=N$, with pieces
$p_j(t)=\sum_{d=0}^D a_{d,j} t^d,\enskip j=1,...,K$, for $t\in\mathbb{Z}^{\geq 0}$, coefficients $a_{d,j} \in\mathbb{R}$, and maximum degree $D=deg(p_j(t))$.
\end{defe}
In the following discussion, we develop a set of novel Graph Wavelet Transform (GWT) designs, which are localised in the vertex domain, critically sampled, invertible, and finally, tailored to the annihilation of polynomial graph signals, the foundation of which is laid by the ensuing result:
\begin{lem} \label{lem31}For an undirected, circulant graph $G=(V,E)$ of dimension $N$, the associated representer polynomial $l(z)=l_0 +\sum_{i=1}^{M} l_{i} (z^i+z^{-i})$ of graph Laplacian matrix ${\bf L}$, with first row $\lbrack l_0 \ l_1\ l_2 \quad ... \quad l_2 \ l_{1}\rbrack$, has two vanishing moments. Therefore, the operator ${\bf L}$ annihilates polynomial graph signals of up to degree $D=1$, subject to a border effect determined by the bandwidth $M$ of ${\bf L}$, provided $2M<<N$.\end{lem}
\begin{proof}
The representer polynomial of ${\bf L}$ with degree $d=\sum_{i=1}^M 2d_i$ per node and symmetric weights $d_i=A_{j,(i+j)_N}$, can be expressed as:
\[ l(z)=(-d_M z^{-M}-...-d_1z^{-1}+d-d_1 z-...-d_M z^M)=\sum_{i=1}^M d_i (z^i-1)(z^{-i}-1),\]
whereby factors are divisible by $(z^{\pm 1}-1)$ respectively using the equality $z^n-1=(z-1)(1+z+...+z^{n-1})$, thus proving that the matrix ${\bf L}$ has two vanishing moments.
Therefore, for a sufficiently small $M$ with respect to the dimension $N$ of the graph $G$, or in other words, if the adjacency matrix of $G$ is a symmetric, banded circulant matrix of bandwidth $M$, the corresponding ${\bf L}$ annihilates linear polynomial graph signals on $G$ up to a boundary effect.
\end{proof}
Given that ${\bf L}$ is circulant, we can generalize this property by considering ${\bf L}^k,k\in\mathbb{N}$, which has $2k$ vanishing moments due to the equivalency between polynomial and circulant matrix multiplication. In general, we observe that for an arbitrary graph $G$, whose graph Laplacian matrix ${\bf L}$ has rows of the above form, we may achieve similar annihilation, yet this property is not carried over to higher order $k$. Following the interpretation of the graph Laplacian matrix as a high-pass filter, this insight gives rise to a new range of graph wavelet filterbanks, whose high-pass filters can annihilate higher-order polynomial graph signals for a sparse graph wavelet domain representation, as we will demonstrate in the next sections. 
\subsection{Graph Spline Wavelets}
We begin by showing that the `spline-like' graph wavelet transform (see Eqns. (\ref{eq:eq1})-(\ref{eq:eq2})) can be generalized to higher order by raising its filters to the $k$-th power, thereby incorporating the previously detected vanishing moment property:
\begin{thm}\label{thm31}
Given the undirected, and connected circulant graph $G=(V,E)$ of dimension $N$, with adjacency matrix ${\bf A}$ and degree $d$ per node, we define the higher-order graph-spline wavelet transform (HGSWT), composed of the low-and high-pass filters
\begin{equation}\label{eq:t31}{\bf H}_{LP}=\frac{1}{2^k}\left({\bf I}_N+\frac{{\bf A}}{d}\right)^k\end{equation}
\begin{equation}\label{eq:t231}{\bf H}_{HP}=\frac{1}{2^k}\left({\bf I}_N-\frac{{\bf A}}{d}\right)^k\end{equation}
whose associated high-pass representer polynomial $H_{HP}(z)$ has $2k$ vanishing moments. This filterbank is invertible for any downsampling pattern, as long as at least one node retains the low-pass component, while the complementary set of nodes retains the high-pass components.
\end{thm}
\noindent \textit{Proof.} See Appendix $A.1$.\\
\\
In particular, given a graph signal ${\bf p}\in\mathbb{R}^N$ defined on $G$, the \textit{HGSWT} yields
\[{\bf \tilde{p}}=\left(\frac{1}{2}({\bf I}_N+{\bf K}){\bf H}_{LP}+\frac{1}{2}({\bf I}_N-{\bf K}){\bf H}_{HP}\right){\bf p}=\frac{1}{2}({\bf I}_N+{\bf K}) {\bf \tilde{p}}_{LP}+\frac{1}{2}({\bf I}_N-{\bf K}){\bf \tilde{p}}_{HP}\]
where ${\bf K}$ is a diagonal sampling matrix, with $K_{i,i}=1$ at $i=0,2,...,N-2$ and $K_{i,i}=-1$ otherwise, i.e. we downsample w.r.t. $s=1\in S$ and retain even-numbered nodes. The resulting signals ${\bf \tilde{p}}_{LP}, {\bf \tilde{p}}_{HP}\in\mathbb{R}^{N}$ represent low-and high-pass versions of ${\bf p}$ on $G$ within a $k$-hop local neighborhood $N(i,k), \forall i\in V$. The higher the degree of the filterbank, the higher the number of vanishing moments, and the less localized it becomes in the vertex domain. In Figure $2$, the low-and high-pass graph filter functions of the \textit{HGSWT} are plotted for $k=2$ and furthermore the spread in the vertex domain is illustrated for a sample circulant graph. 
\\
\begin{figure}[tbp]
	\centering
	{\includegraphics[width=1.5in]{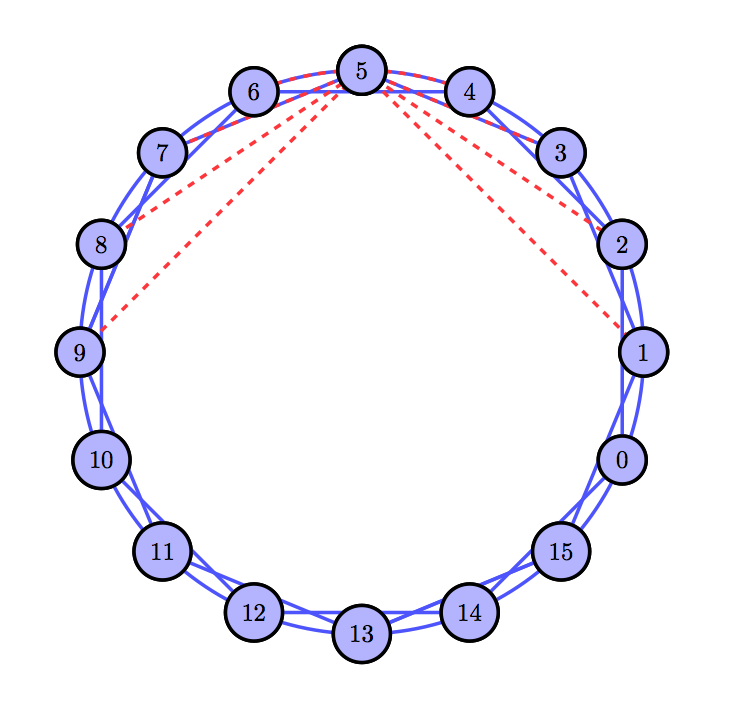}}%
	{\includegraphics[width=3in]{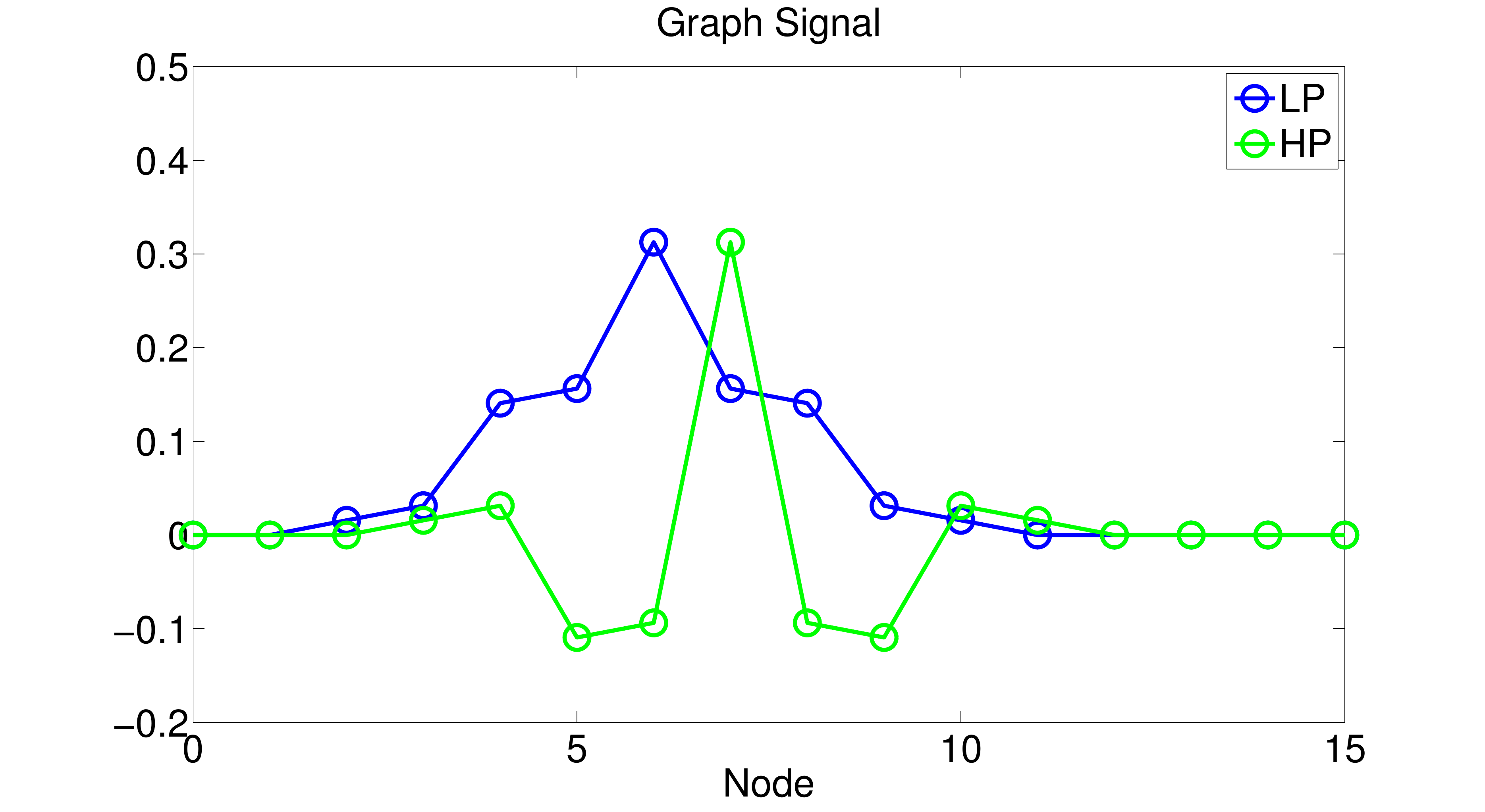}}%
		\caption{Localization of the HGSWT filters for $k=2$ in the graph vertex domain for graph $G$ with $S=\{1,2\}$: shown at vertex $v=5\in V$ on $G$ (left), and the corresponding graph filter functions at alternate vertices.}
\end{figure}
\\
A bipartite circulant graph $G$ is characterized by a generating set $S$ which contains only odd elements $s_k\in S$ for even dimension $N$, with the simple cycle $S=\{1\}$ as a natural example; we note the following interesting property of the \textit{HGSWT}, when $G$ is such:
\begin{cor}\label{cor31}
When $G$ is  an undirected, circulant, bipartite graph, with adjacency matrix ${\bf A}$ of bandwidth $M$, the polynomial representation $H_{LP}(z)$ of the low-pass filter ${\bf H}_{LP}$ in Eq. (\ref{eq:t31}) can reproduce polynomial graph signals up to degree $2k-1$, subject to a border effect determined by the bandwidth $Mk$ of ${\bf H}_{LP}$, provided $2Mk<<N$.
\end{cor}
\begin{proof} Similarly, as in Lemma \ref{lem31}, we can express the representer polynomial as
\[H_{LP}(z)=\frac{1}{(2d)^k}(d_M z^{-M}+...+d_1z^{-1}+d+d_1 z+...+d_M z^M)^k\]\[=\left(\frac{1}{2d}\sum_{1\leq i \leq M, i\in2\mathbb{Z}^{+}+1}d_i (z^i+1)(z^{-i}+1)\right)^k\] 
and note that the RHS factors $(z+1)^k(z^{-1}+1)^k$, since $(z^i+1)$ has a root at $z=-1$ only for $i\in 2\mathbb{Z}+1$. According to the Strang-Fix condition (\cite{strang}, \cite{vet}), this is necessary and sufficient for ensuring the reproduction of polynomials. 
\end{proof}
\noindent It becomes evident that our use of the spline-wavelet terminology is well-founded, since in the case of bipartite circulant graphs the polynomial reproduction property can be generalised to higher order $k$, and graph filters (\ref{eq:t31})-(\ref{eq:t231}) respectively reproduce and annihilate polynomial graph signals up to degree $n=2k$-$1$, bridging the gap to the traditional domain. We will elaborate more thoroughly on similarities with the classical spline and spline wavelets in Sect. $3.3$.\\
Further, the matrix ${\bf Q}={\bf D}+{\bf A}$, chosen as the low-pass filter in present constructions, is also known as the \textit{signless Laplacian} and has been studied for its spectral properties \cite{signless2}. For a connected graph, ${\bf Q}$ is a positive semi-definite matrix, whose smallest (simple) eigenvalue is $0$ if and only if the graph is bipartite; its multiplicity is further equal to the number of connected bipartite components \cite{signless1}. A more generalized form of this result has been independently derived in Cor. \ref{cor33} of Sect. $3.2$.\\
In particular, the characteristic polynomials, i.e. the eigenvalues, of the signless and traditional graph Laplacian are known to be the same for a bipartite graph. Notably, the eigenvalues of the bipartite adjacency matrix are symmetric with respect to zero \cite{chung}. This is further illustrated within the derived reproduction property of Cor. \ref{cor31} for circulant bipartite graphs; here, the frequency parameters $z=e^{\frac{2\pi i k}{N}}$ and $-z=e^{\frac{2\pi i (k+N/2)}{N}}$, $k=0,...,N-1$, as incorporated in the relation $H_{LP}(-z)=H_{HP}(z)$, induce graph-filter eigenvalues that are shifted by $N/2$ in their position within the DFT-ordered spectrum, following the spectral folding $A(-z) = -A(z)$ of the representer polynomial $A(z)$ of the adjacency matrix ${\bf A}$.

\subsection{Graph E-Spline Wavelets}
Inspired by the generalized framework of cardinal exponential splines \cite{espline} in the classical domain, we proceed to identify a new class of graph signals and graph wavelets which maintain and extend these properties to the graph domain.
\begin{defe}
A complex exponential polynomial graph signal ${\bf y}\in\mathbb{C}^N$ with parameter $\alpha\in\mathbb{R}$, is defined such that node $j$ has sample value $y(j)=p(j) e^{i \alpha j}$, for  polynomial ${\bf p}\in\mathbb{R}^N$ of degree $\textit{deg}(p(t))$.\end{defe}
In accordance with the definition of e-splines in the classical domain via a differential operator, and our previous result on the graph Laplacian, we introduce the e-graph Laplacian as a generalized graph difference operator:
\begin{defe}
Let $G=(V,E)$ be an undirected, circulant graph with adjacency matrix ${\bf A}$ and degree $d=\sum_{j=1}^M 2 d_j$ per node with symmetric weights $d_{j}=A_{i,(j+i)_N}$. Then the parameterised e-graph Laplacian of $G$ is given by $\tilde{{\bf L}}_{\alpha}=\tilde{{\bf D}}_{\alpha}-{\bf A}$, with exponential degree $\tilde{d}_{\alpha}=\sum_{j=1}^M 2 d_j \cos(\alpha j)$.
\end{defe}
The standard graph Laplacian ${\bf L}$ can be therefore regarded as a special case of the e-graph Laplacian $\tilde{{\bf L}}_{\alpha}$ for $\alpha=0$, however, with $\tilde{d}_{\alpha}\leq d$ the matrix ceases to be positive semi-definite otherwise. 
\begin{lem}\label{lem32}
For an undirected, circulant graph $G=(V,E)$ of dimension $N$, the associated representer polynomial $\tilde{l}(z)=\tilde{l}_0 +\sum_{i=1}^{M} \tilde{l}_{i} (z^i+z^{-i})$ of the e-graph Laplacian matrix $\tilde{{\bf L}}_{\alpha}$, with first row $\lbrack \tilde{l}_0 \ \tilde{l}_1\ \tilde{l}_2 \quad ... \quad \tilde{l}_2 \ \tilde{l}_{1}\rbrack$, has two vanishing exponential moments, i.e. the operator $\tilde{{\bf L}}_{\alpha}$ annihilates complex exponential polynomial graph signals with exponent $\pm i\alpha$ and $\textit{deg}(p(t))=0$. Unless $\alpha=\frac{2\pi k}{N}$ for $k\in\lbrack 0, N-1\rbrack$, this is subject to a border effect determined by the bandwidth $M$ of $\tilde{{\bf L}}_{\alpha}$, provided $2M<<N$.
\end{lem}
\begin{proof}
Consider the representer polynomial $\tilde{l}(z)$ of $\tilde{{\bf L}}_{\alpha}$:
\[\tilde{l}(z)=\sum_{j=1}^M 2 d_j \cos(\alpha j)-d_j(z^{j}+ z^{-j})
=\sum_{j=1}^M d_j (1-e^{i\alpha j}z^j )(1-e^{-i\alpha j}z^j )(-z^{-j})\]
where we note that $(1-e^{\mp i\alpha}z^{-1})$ is a factor of $(1-e^{\pm i\alpha j}z^j)$, which corresponds to two exponential vanishing moments \cite{espline}.
\end{proof}
\noindent As we require the weights of a graph wavelet filter ${\bf H}$ to be symmetric and real-valued, the construction of $\tilde{{\bf L}}_{\alpha}$ with polynomial factors $(1-e^{i\alpha}z^{-1})(1-e^{-i\alpha}z^{-1})$ ensures this. \\
The focus in this work is restricted to the class of complex exponential polynomial graph signals of the form $y(t)=p(t) e^{i\alpha t},\alpha \in\mathbb{R}$, which can be represented by trigonometric splines in the traditional domain, however, one can expand the framework by letting $\alpha=-i\beta,\beta\in\mathbb{R}$. In particular, this parameterization gives rise to real exponential polynomials, which can be represented in terms of hyperbolic functions of the form $(\cosh(\beta t),\sinh(\beta t))$, thereby inducing the class of hyperbolic splines \cite{espline}, with redefined e-degree $\tilde{d}=\sum_{k=1}^M 2 d_k \cos(-i \beta k)=\sum_{k=1}^M 2 d_k \cosh(\beta k)$.\\
\\
To provide an intuition behind the structure of the e-graph Laplacian, we make the following remark:
\begin{rmk}
The eigenvalues $\{\lambda_j\}_{j=0}^{N-1}$ of a circulant matrix, and in particular of ${\bf A}$ in Lemma \ref{lem32}, can be expressed as $\lambda_j=\sum_{k=1}^M 2 d_k \cos\left(\frac{2 \pi k j}{N}\right),\quad j=0,...,N-1$. We note that all circulant matrices have the same eigenbasis ${\bf U}$. Hence, when we restrict $\alpha=\frac{2\pi j}{N}$ and $j\in\lbrack 0, N-1\rbrack$, the nullspace of the corresponding e-graph Laplacian $\tilde{{\bf L}}_{\alpha}$ consists of its $j$-th eigenvector ${\bf u}_j$, where ${\bf u}_j$ represents a complex exponential graph signal with $\alpha=\frac{2\pi j}{N}$ and $\textit{deg}(p(t))=0$. In particular, we have $\tilde{d}_{\alpha}=\lambda_j$, and for $\alpha=0$ this becomes the maximum eigenvalue $\tilde{d}_0=d=\lambda_{max}$, whose associated eigenvector is the all-constant ${\bf u}_{max}={\bf 1}_N$, i.e. the nullspace of standard graph Laplacian ${\bf L}$. This facilitates the reinterpretation of the e-graph Laplacian as $\tilde{{\bf L}}_{\alpha}=\lambda_j {\bf I}_N-{\bf A}$, for $\alpha=\frac{2\pi j}{N}$ and $j\in\lbrack 0, N-1\rbrack$, or more generally, $(\lambda_j {\bf I}_N-{\bf A}){\bf u}_j={\bf 0}_N$, with $\tilde{{\bf L}}_{\alpha}$ representing the shift of ${\bf L}$ by $\lambda_j-d=-\tilde{\lambda}_j$ toward annihilation of ${\bf u}_j$, where $\{\tilde{\lambda}_j\}_{j=0}^{N-1}$ denotes the spectrum of ${\bf L}$. Depending on eigenvalue multiplicities, the nullspace of $\tilde{{\bf L}}_{\alpha}$ is accordingly extended. 
\end{rmk}
Revisiting our initial discussion on vanishing moments, it becomes evident that the annihilation property of the e-graph Laplacian is related to the definition of vanishing moments on graphs by Coifman et al. in \cite{Coifman}, as the nullspace of the former consists of (a subset of) its eigenvectors; however, we do not extend this definition up to a precision metric, and note that our chosen operator is parametric. The significance of our approach lies in the fact that annihilation is also local, and not restricted to the nullspace of the operator $\tilde{{\bf L}}_{\alpha}$. Even though we are mainly interested in properties pertaining to circulant graphs, one may consider extending the idea of a `nullspace-shifted' graph Laplacian operator to all regular graphs, for which it is known that ${\bf L}$ and ${\bf A}$ share the same eigenbasis \cite{chung}, as well as arbitrary graphs, extending the classes of graph signals which can be annihilated on their nodes. \\
In addition, while the ordered eigenvalues of $\tilde{{\bf L}}_{\alpha}$ are no longer nonnegative, their interpretation as graph frequencies which order the corresponding eigenvectors in terms of the number of their oscillations (zero crossings) \cite{shu} remains valid, with the only difference that the graph frequency $d-\lambda_j$ of ${\bf L}$ becomes the new zero or DC-frequency in $\tilde{{\bf L}}_{\alpha}=(\lambda_j{\bf I}_N-{\bf A})$. Further, the e-graph Laplacian quadratic form can be expressed as $\tilde{S}_{\alpha2}({\bf x})={\bf x}^T \tilde{{\bf L}}_{\alpha}{\bf x}=(\tilde{d}_{\alpha}-d)||{\bf x}||_2^2 +{\bf x}^T {\bf L} {\bf x}=(\tilde{d}_{\alpha}-d)||{\bf x}||_2^2 +S_{2}({\bf x})$.\\
\\
Based on these insights, we proceed to design a graph e-spline wavelet filterbank, following a similar line as the (higher-order) graph-spline wavelet filterbank. In order to generalize the types of graph signals which can be reproduced and/or annihilated by a GWT, we incorporate multiple parameters $\vec{\alpha}=(\alpha_1,...,\alpha_T)\in\mathbb{R}^T$ via a simple circular convolution of the graph filter functions, resulting in an invertible transform:
\begin{thm}\label{thm32}
The higher-order graph e-spline wavelet transform (HGESWT) on a connected, undirected circulant graph $G$, is composed of the low-and high-pass filters
\begin{equation}\label{eq:t32}{\bf H}_{LP_{\vec{\alpha}}}=\prod_{n=1}^T\frac{1}{2^k} \left(\beta_n{\bf I}_N+\frac{{\bf A}}{d}\right)^k\end{equation}
\begin{equation}\label{eq:t232}{\bf H}_{HP_{\vec{\alpha}}}=\prod_{n=1}^T \frac{1}{2^k}\left(\beta_n{\bf I}_N-\frac{{\bf A}}{d}\right)^k\end{equation}
where ${\bf A}$ is the adjacency matrix, $d$ the degree per node and parameter $\beta_n$ is given by $\beta_n=\frac{\tilde{d}_{\alpha_n}}{d}$ with $\tilde{d}_{\alpha_n}=\sum_{j=1}^M 2 d_j \cos(\alpha_n j)$ and $\vec{\alpha}=(\alpha_1,...,\alpha_T)$. Then the high-pass filter annihilates complex exponential polynomials (of deg$(p(t))\leq k-1$) with exponent $\pm i \alpha_n$ for $n=1,...,T$. The transform is invertible for any downsampling pattern as long as the eigenvalues $\gamma_i$ of $\frac{{\bf A}}{d}$ satisfy $|\beta_n|\neq |\gamma_i|,\enskip i=0,...,N-1$, under either of  the sufficient conditions
 \\
$(i)$ $k\in2\mathbb{N}$, or \\
$(ii)$ $k\in\mathbb{N}$ and $\beta_n, T$ such that $\forall \gamma_i, f(\gamma_i)=\prod_{n=1}^T(\beta_n^2-\gamma_i^2)^k> 0$ or $f(\gamma_i)< 0$.\\
If parameters $\beta_n$, are such that $\beta_n= \gamma_i$, for up to $T$ distinct values, the filterbank continues to be invertible under the above as long as $\beta_n\neq 0$ and at least $\sum_{i=1}^{T} m_i$ low-pass components are retained at nodes in set $V_{\alpha}$ such that $\{{\bf v}_{+i, k} (V_{\alpha})\}_{i=1, k=1}^{i=T, k=m_i}$ (and, if eigenvalue $-\gamma_i$ exists, complement $\{{\bf v}_{-i, k}({V_{\alpha}^{\complement}})\}_{i=1, k=1}^{i=T, k=m_i}$) form linearly independent sets, where $m_i$ is the multiplicity of $\gamma_i$ and $\{{\bf v}_{\pm i,k}\}_{k=1}^{m_i}$ are the eigenvectors respectively associated with $\pm \gamma_i$.
\end{thm}
\noindent \textit{Proof.} See Appendix $A.2$. \\
\\
The essential property of the above transform, which is captured in the proof, is that invertibility is governed by the parameters $\beta_n$. In particular, in the case where the chosen $\beta_n$ coincide with the magnitude(s) of certain eigenvalues $\{\gamma_i\}_i$ of normalized adjacency matrix $\frac{{\bf A}}{d}$, a suitable downsampling pattern needs to be selected such that the eigenvectors $\{{\bf v}_{i, k}\}_{i,k}$ associated with $\{\gamma_i\}_i$ remain linearly independent after downsampling. 
It is easily inferred that this can become challenging for certain graph topologies as well as when the number $T$ of such parameters $\beta_n$ is large. In particular, circulant adjacency matrices may exhibit large eigenvalue multiplicities as a result of increased graph connectivity. For instance, the normalized adjacency matrix of an unweighted complete (and hence circulant) graph of dimension $N$ has $\gamma_i=-\frac{1}{d}$ of multiplicity $N-1$ and $\gamma_{max}=1$, which is simple. In this case, the transform is invertible for $\beta=1$, but not for $\beta=-\frac{1}{d}$, when downsampling is conducted w.r.t. $s=1\in S$. This may be remedied i.a. by introducing distinct edge weights. \\
We further note that the graph filter powers $k$ may also be chosen to differ for each unique factor, provided that the invertibility conditions remain satisfied. Moreover, Thm. \ref{thm32} equivalently applies to real exponential polynomial signals with e-degree parameterization of the form $\tilde{d}_{i\alpha}=\sum_{k=1}^M 2d_k \cosh (\alpha k)$, $\cosh(x)\in\lbrack 1\enskip \infty )\enskip \forall x$. In particular, for this instance the invertibility conditions simplify and take the form of those in Thm. \ref{thm31} due to the property $\left|\frac{\tilde{d}_{i\alpha}}{d}\right|\geq 1\geq |\gamma_j|,\enskip j\in \lbrack 0\enskip N-1\rbrack$, facilitating a more convenient signal analysis.

Further, one deduces that the \textit{HGESWT} of Thm. \ref{thm32} converges to the \textit{HGSWT} of Thm. \ref{thm31} \footnote{It should be noted that invertibility of the wavelet transforms in Thms \ref{thm31} and \ref{thm32} does not require ${\bf A}$ to be a circulant matrix. In particular, for the former, it is sufficient that $\frac{{\bf A}}{d}$ be replaced by the normalized adjacency matrix $\hat{{\bf A}}={\bf D}^{-1/2}{\bf A}{\bf D}^{-1/2}$ of an undirected connected graph \cite{chung}, and for the latter, its eigenbasis ${\bf V}$ is subject to similar constraints, if chosen parameters $\beta_n$ coincide with the eigenvalues of $\hat{{\bf A}}$ (see Appendix $A.2$). However, (higher-order) vanishing moments are lost and downsampling becomes less intuitive and/or accurate when applying these transforms to non-circulant graphs.} for $\alpha\rightarrow 0$ and $\beta=\frac{\tilde{d}_{\alpha}}{d}\rightarrow 1$, in which case the conditions on invertibility are relaxed, and, as a consequence of this structural similarity, one can detect similar reproduction properties for circulant bipartite graphs as in Cor. \ref{cor31}:
\begin{cor}\label{cor32}
Let $G=(V,E)$ be an undirected, bipartite circulant graph with adjacency matrix ${\bf A}$ of bandwidth $M$, and e-degree $\tilde{d}_{\alpha}$. Then the low-pass filter ${\bf H}_{LP_{\alpha}}$ of Eq. (\ref{eq:t32}) reproduces complex exponential polynomial graph signals ${\bf y}$ with exponent $\pm i\alpha$, up to a border effect determined by the bandwidth $Mk$ of ${\bf H}_{LP_{\alpha}}$, provided $2Mk<<N$.
\end{cor}
\begin{proof} The representer polynomial $H_{LP_{\alpha}}(z)$ is of the form
\[H_{LP_{\alpha}}(z)=\frac{1}{(2d)^k}(d_M z^{-M}+...+d_1 z^{-1}+\sum_{j=1}^M 2 d_j \cos(\alpha j)+d_1 z^{1}+...+d_M z^{M})^k\]\[=\left(\frac{1}{2d}\sum_{j=1}^M d_j (1+e^{i\alpha j}z^j )(1+e^{-i\alpha j}z^j )(z^{-j})\right)^k,\enskip j\in2\mathbb{Z}^{+}+1,\]
where $(1+e^{\mp i\alpha }z^{-1})$ is a factor of $(1+e^{\pm i\alpha j}z^j)$ if $j$ is odd, i.e. the elements in the generating set of $G$ are odd. This filter therefore satisfies the generalized Strang-Fix conditions for the reproduction of exponentials \cite{esplinewav}. 
\end{proof} 
\noindent At last, we deduce the following property pertaining to the low-pass filter, which we will leverage for further graph wavelet constructions:
\begin{cor}\label{cor33}
Let $G=(V,E)$ be an undirected, circulant graph with adjacency matrix ${\bf A}$ and degree $d=\sum_{j=1}^M 2 d_j$ per node with symmetric weights $d_{j}=A_{i,(j+i)_N}$. Then the low-pass filter ${\bf H}_{LP_{\vec{\alpha}}}$ of Eq. (\ref{eq:t32}) is invertible unless $(i)$ $G$ is bipartite while $\beta_n$ satisfies $|\beta_n|= |\gamma_i|$ or $(ii)$ $\beta_n= -\gamma_i,\enskip i\in\lbrack 0\enskip N-1\rbrack$. \end{cor}
\noindent \textit{Proof.} See Appendix $A.3$.\\
\\
Prior to broadening the range of graph wavelet constructions on the basis of the aforementioned families, we conduct further analysis of the characteristic structure of the latter. Here, we focus on a bipartite circulant graph scenario, as this case is particularly relevant due to its dual vanishing moment property. Contrary to standard biorthogonal wavelet filterbanks, the proposed graph spline wavelet constructions exhibit well-defined analysis filters, while their corresponding synthesis filters lack a concrete characterization. Let the general analysis matrix of a graph wavelet transform be denoted by \[{\bf W}=\begin{bmatrix} {\bf \Psi}_{\downarrow 2}{\bf H}_{LP}\\{\bf \Phi}_{\downarrow 2}{\bf H}_{HP}\end{bmatrix},\] for downsampling matrices ${\bf \Psi}_{\downarrow 2},{\bf \Phi}_{\downarrow 2}\in\mathbb{R}^{N/2\times N}$, which respectively retain even and odd numbered nodes, and graph-based low-and high-pass filters ${\bf H}_{LP},{\bf H}_{HP}$, whose respective representer polynomials are expressed by $H_{LP}(z)$, $H_{HP}(z)$. Given that ${\bf W}$ is invertible, and contains two sets of basis functions, their corresponding duals, which we denote by $\tilde{H}_{LP}(z), \tilde{H}_{HP}(z)$ exist in ${\bf W}^{-1}=\begin{bmatrix} ({\bf \Psi}_{\downarrow 2}\tilde{{\bf H}}_{LP})^T& ({\bf \Phi}_{\downarrow 2}\tilde{{\bf H}}_{HP})^T\end{bmatrix}$, by definition of a biorthogonal system \cite{biorref}. In general, the following relations hold for $i,j\in\{LP,HP\}$
\[H_i(z)\tilde{H_i}(z)+H_i(-z)\tilde{H_i}(-z)=2\]
\[H_i(z)\tilde{H_j}(z)+H_i(-z)\tilde{H_j}(-z)=0,\enskip i\neq j\]
where $H_{HP}(z)=z H_{LP}(-z)$ is established in proposed constructions.\footnote{Here, in an abuse of notation, we incorporate a shift $z$, equivalently to the traditional $z$-transform, to signify that odd-numbered rows retain the high-pass component.} By substitution into the above, we observe that $\tilde{H}_{HP}(z) =z^{-1}\tilde{H}_{LP}(-z)$ must equivalently hold for the dual pair. This can be proved directly through the simple inversion of the analysis modulation matrix (as is done for classical perfect reconstruction filterbanks, see e.g. \cite{strangbook}), which we omit for brevity, and further reveals that the derived synthesis filters are rational functions whose zeros coincide with those of the analysis filters, i.e. $\tilde{H}_{HP}(z)$ and ${H}_{HP}(z)$ have the same vanishing moments.

\subsection{Splines on Graphs}
In the classical domain of signal processing, the $B$-spline of degree zero $\beta^0_{+}$, representing the box-function,
\[\beta^0_{+}(x)=\left\{
                \begin{array}{ll} 1, \enskip x\in\lbrack 0,1)\\
0, \enskip\text{otherwise}\end{array}
              \right.\]
is defined through the action of discrete (finite difference) operator $\Delta_{+}\{\cdot\}$, with $z$-transform $(1-z^{-1})$, on the step function $x^0_{+}$ such that $\beta_{+}^0(x)=\Delta_{+}x^0_{+}$; here $x^0_{+}=D^{-1}\{\delta(x)\}$ is the Green's function of continuous first-order differential operator $D\{\cdot\}$ \cite{splines}. A spline of degree $n$ is then obtained through $(n+1)$-fold convolution of the box function $\beta_{+}^{(0,...,0)}=\beta^n_{+}(x)=\beta^0_{+}*\beta^0_{+}*...*\beta_{+}^0(x)$, and similarly constructed through the higher-order operator $\Delta_{+}^{n+1}\{\cdot\}$, with $z$-transform $(1-z^{-1})^{n+1}$, such that $\beta^n_{+}(x)=\frac{\Delta_{+}^{n+1} x^n_{+}}{n!}$, where $x^n_{+}$ is the one-sided power function \cite{splines}. \\
The connection between the continuous and discrete time domain is established via the identity\footnote{Here, $S'$ denotes Schwartz's class of tempered distributions.}: $\forall f\in S',\enskip \Delta^m_{+} \{f\}=\beta^{m-1}_{+}*D^m \{f\}$, and more generally via $\Delta_{+}^{\alpha}\{f\}=\beta_{+}^{\alpha}*(D-\alpha I)\{f\}$ for operator $\Delta_{+}^{\alpha}(z)=(1-e^{\alpha}z^{-1})$
and exponential spline $\beta^{\alpha}_{+}$, with $\alpha\in\mathbb{C}$. In particular, $\Delta_{+}^{\alpha}\{\cdot\}$ can be regarded as a discrete approximation of the continuous $(D-\alpha I)\{\cdot\}$ (\cite{esplines2}, \cite{splines}, \cite{espline}). Such (exponential) polynomial splines, as solutions of certain variational problems, form a subset of the more generalized variational splines (\cite{varsplines}, \cite{var2}). \\
\\
The graph Laplacian matrix can be viewed as an approximation of the continuous Laplacian operator $-\nabla^2$ via the discrete Laplacian, which motivates the consideration of splines on graphs, nevertheless, a brute-force generalization of standard definitions is primarily hindered by the fact that the former is singular. In \cite{Pesenson}, variational splines on graphs are  determined as the Green's functions of the approximate graph differential operator $({\bf \mathcal{L}}+\epsilon{\bf I}_N)^{t}$ for small $\epsilon>0$, $t\in\mathbb{R}^{+}$, which minimize the Sobolev norm, of the form $({\bf \mathcal{L}}+\epsilon{\bf I}_N)^{-t}{\bf e}_i$ for normalized graph Laplacian ${\bf \mathcal{L}}$ and elementary basis vector ${\bf e}_i\in\mathbb{R}^N$ with $e_i(i)=1$ and $e_i(k)=0$ for $k\neq i$. Furthermore, Chung et al. \cite{green1} define the Green's function of a connected graph, without a direct reference to splines, as $G=\sum_{\lambda_j> 0}\frac{1}{\lambda_j}{\bf u}_j {\bf u}_j^H$,  with normalized graph Laplacian eigenvectors ${\bf u}_j$ and associated eigenvalues $\lambda_j$, and propose closed-form expressions for elementary cases, including the simple cycle graph, which are further extended to Cartesian graph products \cite{green2}.\\ 
\\
Spline wavelet transforms in the Euclidean domain are commonly characterized by (dual) scaling functions which are (combinations of) polynomial splines, with the Cohen-Daubechies-Feauveau wavelet as a prominent example \cite{splines}. In order to provide a more intuitive link between the spline-like properties of our graph wavelet functions and the traditional B-spline, we consider the case of a graph signal residing on the vertices of a simple cycle graph, which we denote with $G_{S_1=(1)}$, as the least connected example of a circulant graph. This graph-representation can be regarded as an analogy to a periodic signal in the discrete domain, where existing edges indicate the sequence of sample values \cite{shu}.\\ 
The rows and columns of the low-pass filter matrix in Eq. (\ref{eq:t31}) of the \textit{HGSWT}, given by 
\[{\bf H}_{LP}:=\scalemath{0.75}{\begin{bmatrix}
0.5 & 0.25 & 0& \cdots     &0& 0.25  \\
 0.25 &  0.5 &0.25&0 & \cdots&0  \\
 &   &\ddots & \ddots  &&  \\
 &  & & &&\\
 &  & & &&\\
 &   & &  \ddots &\ddots&  \\
&  & & &&\\
0.25&  0 &  \cdots        &  0&0.25&  0.5
\end{bmatrix}}^k,\]	
produce traditional higher-order splines via convolution of the discrete linear spline $\tilde{\beta}_{+}^{(0,0)}(t)$ (for $k=1$) with itself, thus creating the notion of a spline-wavelet filterbank. The corresponding (high-pass) graph Laplacian in Eq. (\ref{eq:t231}), not only provides the stencil approximation of the second order differential operator for certain types of graphs such as lattices \cite{Ekambaram3}, but in the case of a simple cycle, and symmetric circulant graphs by extension, gains the actual vanishing moment property. 
It therefore appears that this spline property can be directly extended to bipartite circulant graphs, whose associated \textit{HGSWT} filters retain both the reproduction and annihilation property of traditional spline-wavelets, as noted in Cor. \ref{cor31}.
\begin{figure}[tbp]
	\centering
	{\includegraphics[width=3in]{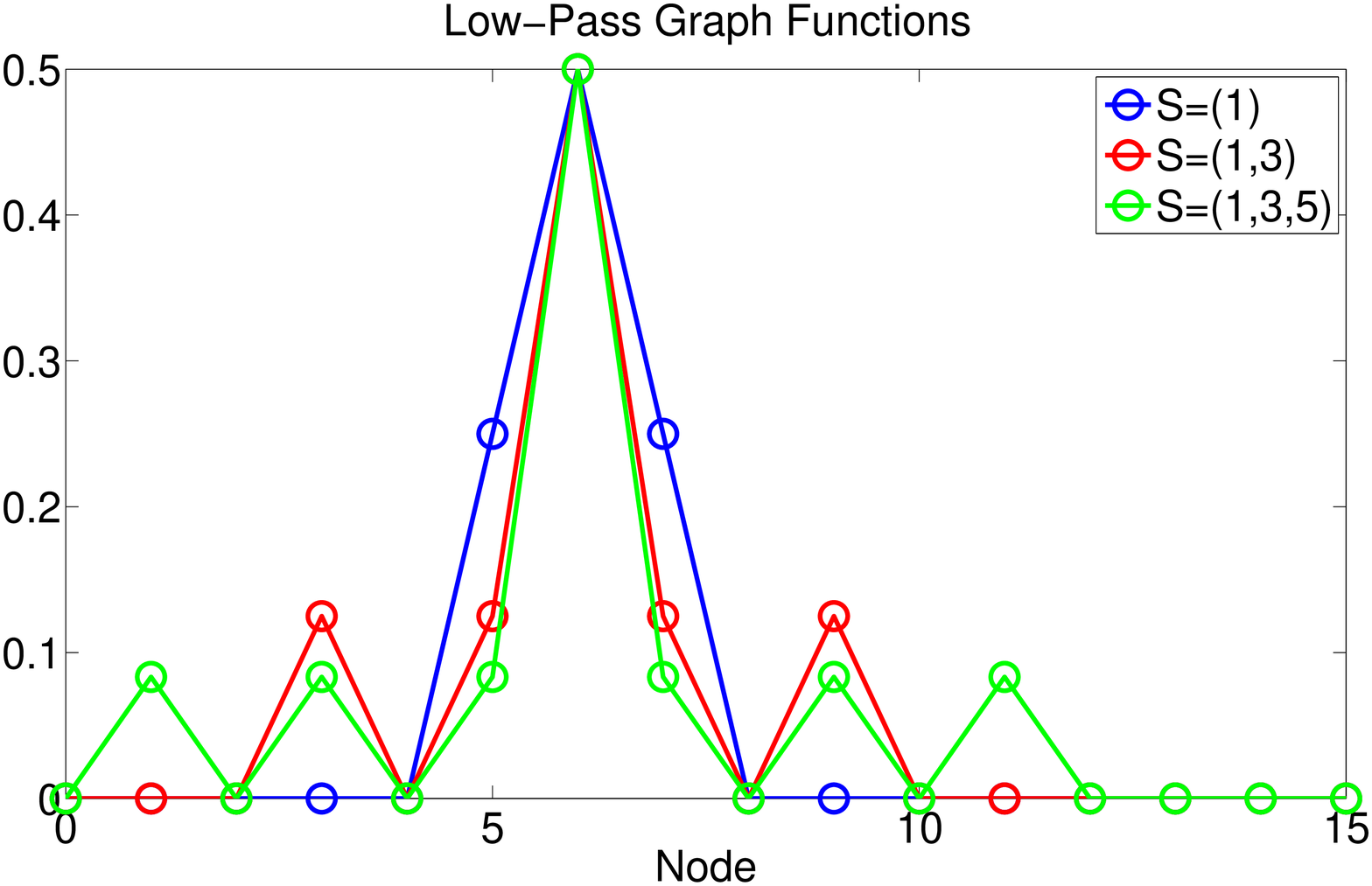}}%
	{\includegraphics[width=3in]{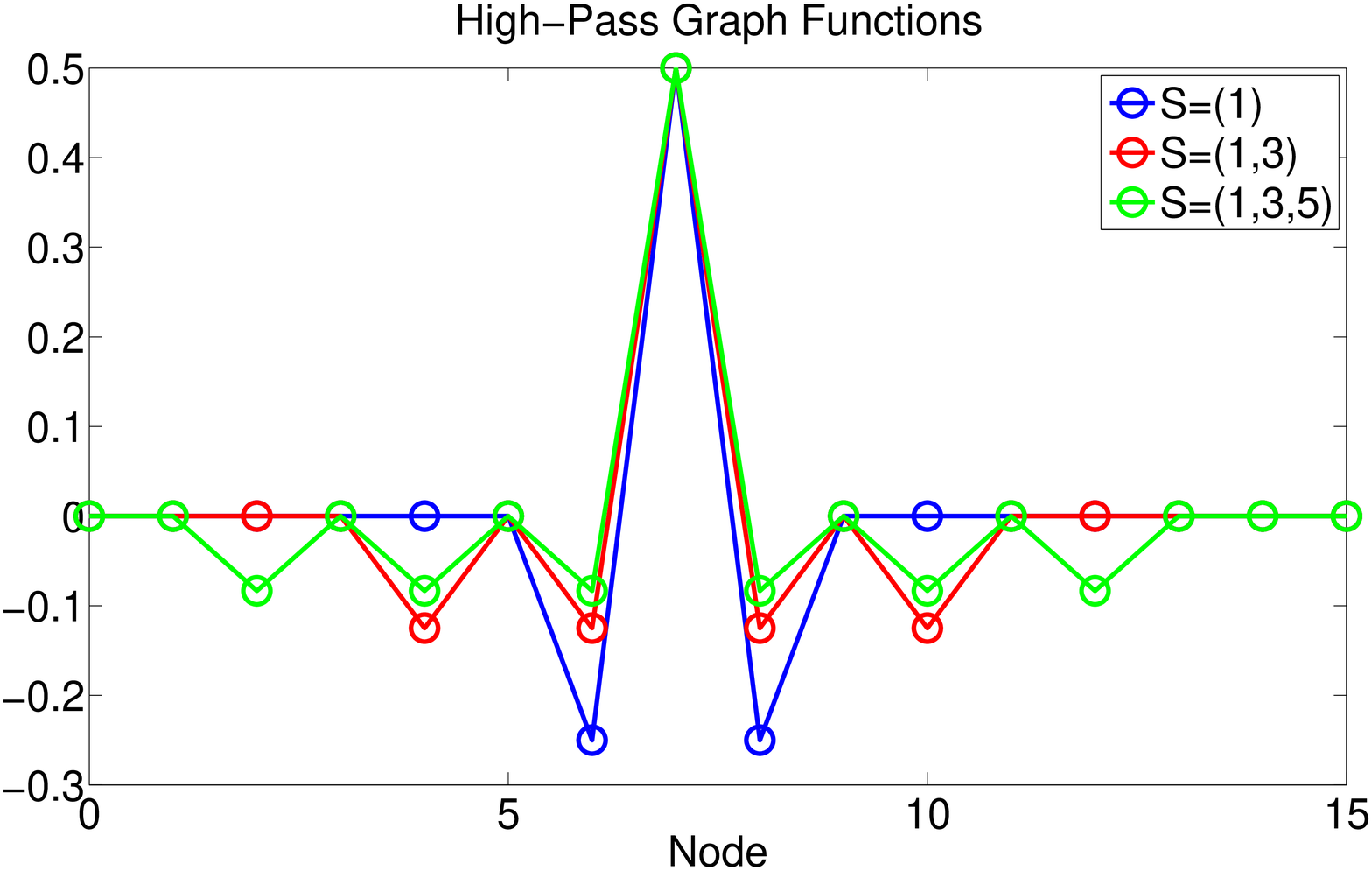}}%
		\caption{The \textit{HGSWT} filter functions at $k=1$ for different bipartite circulant graphs, $N=16$.}
\end{figure}
\noindent Similarly, the rows and columns of the low-pass filter matrix ${\bf H}_{LP_{\alpha}}$ in Eq. (\ref{eq:t32}) of the \textit{HGESWT} (at $k=1$) describe a second-order e-spline, arising from two convolved complex conjugate first-order e-splines
\[{\bf H}_{LP_{\alpha}}:=\scalemath{0.75}{\begin{bmatrix}
	 0.5\cos(\alpha) & 0.25 & 0& \cdots     &0& 0.25  \\
 0.25 &  0.5\cos(\alpha) &0.25&0 & \cdots&0  \\
 &   &\ddots & \ddots  &&  \\
 &  & & &&\\
 &  & & &&\\
 &   & &  \ddots &\ddots&  \\
&  & & &&\\
0.25&  0 &  \cdots        &  0&0.25&  0.5 \cos(\alpha)
\end{bmatrix}}.\]
By considering powers of ${\bf H}_{LP_{\alpha}}$, we obtain the polynomial e-spline basis functions, while the multiplication by low-pass filters of different parameters $\alpha_n$, as in Thm \ref{thm32},  results in convolved heterogeneous e-spline basis functions. From Cor. \ref{cor32}, we gather that the exponential polynomial reproductive properties can be similarly extended to bipartite circulant graphs.
\begin{figure}[tbp]
	\centering
	{\includegraphics[width=3in]{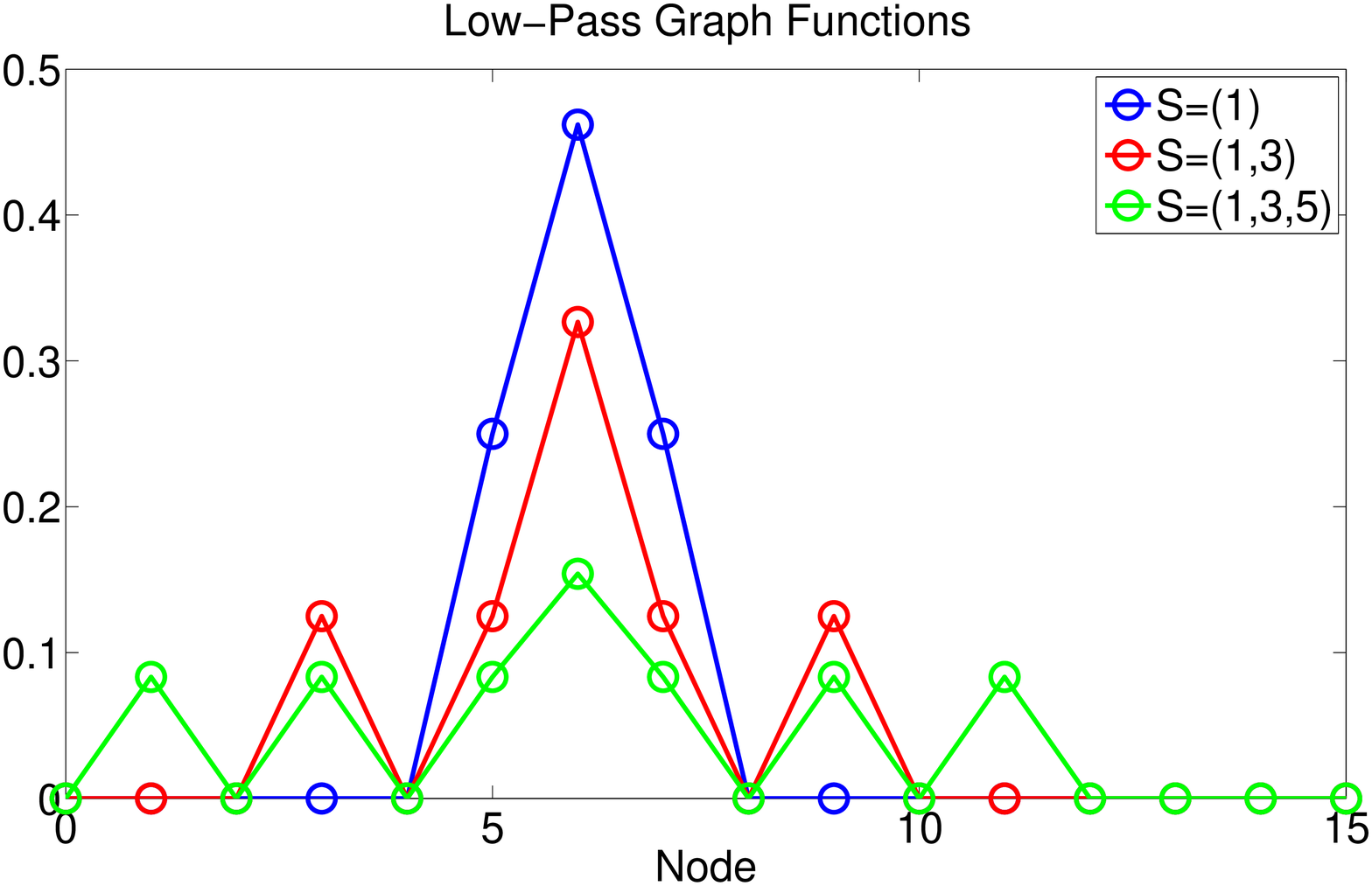}}%
	{\includegraphics[width=3in]{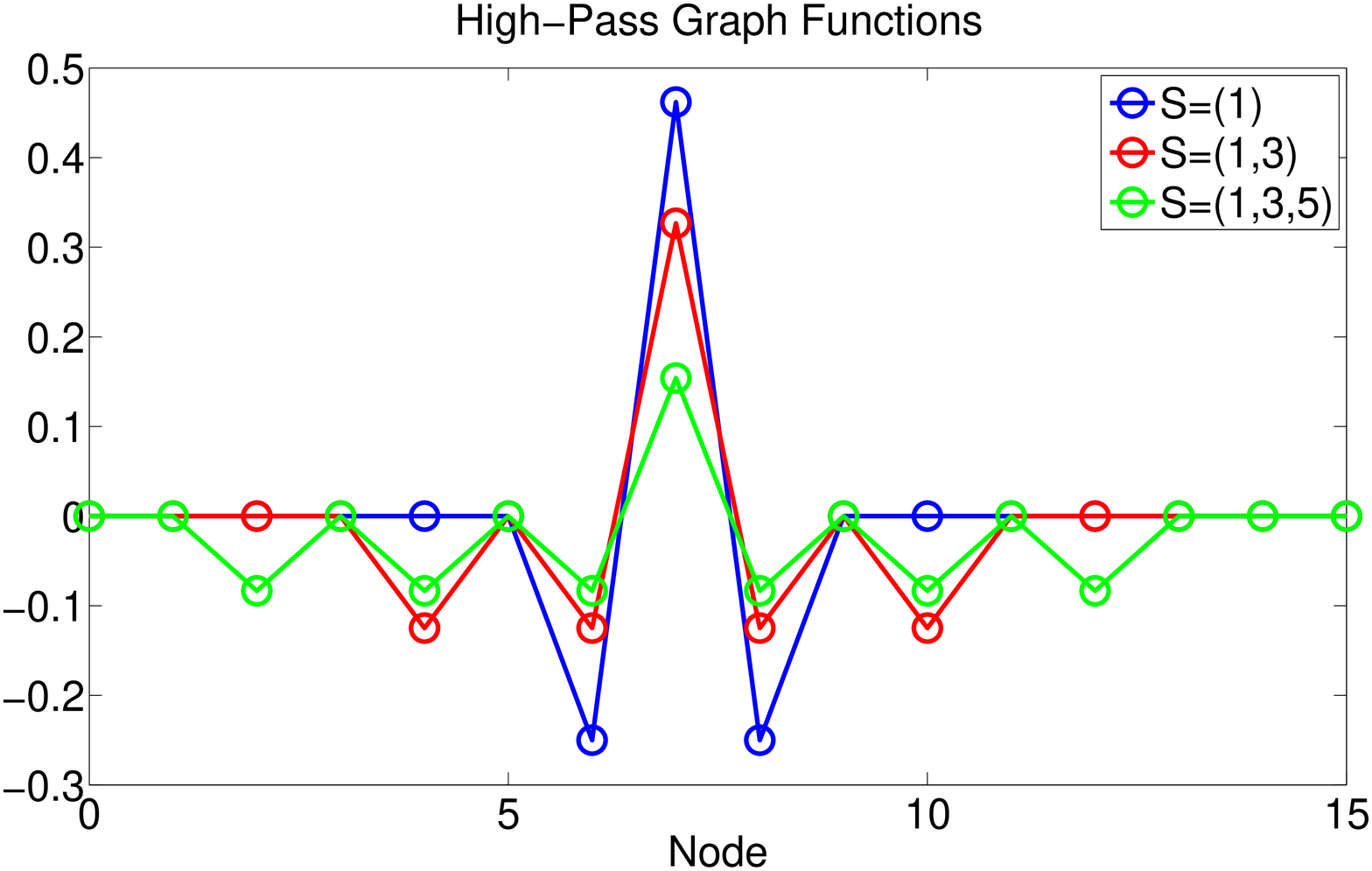}}%
		\caption{The HGESWT filter functions ($k=1$) for different bipartite circulant graphs at $\alpha=\frac{2\pi}{N}$, $N=16$.}
\end{figure}

\begin{table}[h]
\begin{center}
\scalebox{0.7}{
\begin{tabular}{ l c c || c c r }

\hline
  e-Spline & Continuous Operator & Order & Graph (e-)Spline &  Matrix Operator & Order \\ \hline
  $\beta^{(0,0)}(t)$ & $D^2\{\}=\frac{d^2}{dt^2}$& 2 & $(2d{\bf I}_N-\tilde{{\bf L}}_{0}){\bf e}_i$ &$\tilde{{\bf L}}_{0}$ & 2\\
  
  $\beta^{(0,..,0)}(t)$ & $D^{2n}\{\}=\frac{d^{2n}}{dt^{2n}}$& 2n & $(2d{\bf I}_N-\tilde{{\bf L}}_{0})^n{\bf e}_i$ &$\tilde{{\bf L}}^n_{0}$ & 2n\\
  
  $\beta^{(i\alpha ,-i\alpha )}(t)$ & $(D-i \alpha  I)*(D+i\alpha I)$\{\}& 2 & $(2\tilde{d}_{\alpha}{\bf I}_N-\tilde{{\bf L}}_{\alpha}){\bf e}_i$ &$\tilde{{\bf L}}_{\alpha}$ & 2\\
  
  $\beta^{(i\alpha,-i\alpha,...,i\alpha,-i\alpha)}(t)$ &  $(D-i\alpha I)^n*(D+i\alpha I)^n$\{\}& 2n & $(2\tilde{d}_{\alpha}{\bf I}_N-\tilde{{\bf L}}_{\alpha})^n{\bf e}_i$ &$\tilde{{\bf L}}^n_{\alpha}$ & 2n\\
  
  $\beta^{(i\alpha_1,-i\alpha_1,...,i\alpha_m,-i\alpha_m)}(t)$ &  $\prod_{t=1}^m(D-i\alpha_t I)^n*(D+i\alpha_t I)^n$\{\}& 2mn & $\prod_{t=1}^m (2\tilde{d}_{\alpha_t}{\bf I}_N-\tilde{{\bf L}}_{\alpha_t})^n{\bf e}_i$ &$\prod_{t=1}^m \tilde{{\bf L}}^n_{\alpha_t}$ & 2mn\\
\end{tabular}
}\end{center}
\caption{Comparison between the continuous e-Spline and Graph e-Spline Definitions}
\end{table}
In particular, for a bipartite circulant graph, we may interpret the resulting graph spline-like function, which we denote as column $i$ of the generalized low-pass operator $(\tilde{d}_{\alpha}{\bf I}_N+{\bf A}){\bf e}_i=(2\tilde{d}_{\alpha}{\bf I}_N-\tilde{{\bf L}}_{\alpha}){\bf e}_i$ for e-degree $\tilde{d}_{\alpha}=\sum_{k=1}^M 2 d_k\cos(\alpha k)$, as a convolution of the discrete (e-)spline $\tilde{\beta}_{+}^{(i\alpha,-i\alpha)}(t)$ of order $2$ with a function $\phi_G(t)$, which depends on the connectivity of the graph at hand. Thereby, we introduce a link to graphs and make the notion of a `graph-spline' which converges to the `classical' discrete spline as $G_S\rightarrow G_{S_1=(1)}$ for an arbitrary circulant bipartite graph $G_S$ with generating set $S$, more concrete. Figures $3$ and $4$ compare the low-and high-pass functions of the \textit{HGSWT} and \textit{HGESWT} respectively for different bipartite graph examples which correspond to second-order graph (e-)splines, with the traditional (e-)spline represented through the simple cycle.
\\
Our motivation for the use of the spline-terminology for functions $(2\tilde{d}_{\alpha}{\bf I}_N-\tilde{{\bf L}}_{\alpha}){\bf e}_i$ originates from its reproduction properties for (exponential) polynomials and hence structural similarity with its classical counterparts, as well as its definition through a suitable differential operator, i.e. the parameterised graph Laplacian $\tilde{{\bf L}}_{\alpha}$; however, it should be clarified that these do not constitute Green's functions of $\tilde{{\bf L}}_{\alpha}$. 
In comparison, the variational graph splines in \cite{Pesenson} of the form $({\bf \mathcal{L}}+\epsilon{\bf I}_N)^{-t}{\bf e}_i,\enskip t>0$ inherit only approximate properties for (exponential) polynomial reproduction when ${\bf \mathcal{L}}$ is circulant, and contrary to the proposed spline constructions, the former are neither well-characterized on the graph nor compactly supported, and therefore of lesser interest as basis functions for graph wavelets. \\
\\
\noindent We summarize and compare the classes of spline-like functions on bipartite circulant graphs and their classical continuous counterparts in Table $1$ in relation to order, with the symmetrization $\beta^n(x)=\beta^n_{+}\left(x+\frac{n+1}{2}\right)$.
\subsubsection{The directed graph spline}
The collective of results pertaining to vanishing moments of graph operators can be extended to the case when ${\bf A}$ is the adjacency matrix of a directed circulant graph $\vec{G}_S$, and the corresponding graph Laplacian is replaced by a first order (normalized) difference operator of the form ${\bf S}={\bf I}_N-\frac{{\bf A}}{d}$, similarly defined as in \cite{moura}. Let edge $(i,(i+s_k)_N)$ in $\vec{G}_S$ be directed from node $i$ to $(i+s_k)_N$, for $s_k\in S$; as a result of degree-regularity (i.e. the in-and out-degrees of each node are the same) and circularity, ${\bf A}$ maintains the DFT-matrix as its basis. The representer polynomial of operator ${\bf S}$ then possesses one vanishing moment, which can be generalized to higher order $k$, while the degree-parameterised $\tilde{{\bf S}}_{\pm\alpha}=\frac{\tilde{d}_{\pm\alpha}}{d}{\bf I}_N-\frac{{\bf A}}{d}$, featuring the, now complex, e-degree $\tilde{d}_{\pm\alpha}=\sum_{k=1}^M d_k e^{\pm i\alpha k}$, per node, possesses one vanishing exponential moment, i.e. $\tilde{{\bf S}}_{\pm\alpha}$ respectively annihilate exponential graph signals with exponent $\pm i\alpha$. 
Generalizations also apply to the reproduction properties of low-pass filters of the form ${\bf H}_{LP_{{\pm\alpha}}}=\frac{1}{2^k}\left(\frac{\tilde{d}_{\pm\alpha}}{d}{\bf I}_N+\frac{{\bf A}}{d}\right)^k$ in the bipartite case, whereby (linear combinations of) the rows of ${\bf H}_{LP_{{\pm\alpha}}}$ (at $k=1$) reproduce exponentials with reversed exponent $\mp i\alpha$ (and vice versa for the columns). \\
The invertibility of the graph spline wavelet filterbank construction in Thm \ref{thm31} remains intact for directed graphs at $k=1$, as by the Perron Frobenius Thm. \cite{frob} for nonnegative matrices, $\frac{{\bf A}}{d}$ maintains an (albeit complex) spectrum with $|\gamma_i|< \gamma_{max}$ and $\gamma_{max}=1$ of multiplicity $1$ corresponding to eigenvector ${\bf 1}_N$; here ${\bf A}$ is required to be primitive, i.e. ${\bf A}^k>0$ for some $k\in\mathbb{N}$ to ensure $|\gamma_i|<\gamma_{max}$ \cite{frob}. Otherwise, invertibility of the transform for all remaining graph cases depends on the downsampling pattern and requires that the `downsampled' eigenvectors associated with $\gamma_i$ and  $-\gamma_i$ (for $|\gamma_i|=1$) form linearly independent sets respectively. 
Similarly, the proof of the graph e-spline wavelet transform in Thm \ref{thm32} may be extended to accommodate directed graphs under further restrictions. \\
\\
The study of the directed graph case is covered to a lesser extent in spectral graph theory, however, \cite{chungdir} defines the combinatorial graph Laplacian of a directed graph as a symmetrization via the probability transition matrix ${\bf P}$. When the graph is directed, strongly connected and circulant, this would correspond to $\vec{{\bf \mathcal{L}}}={\bf I}_N-\frac{{\bf A}+{\bf A}^H}{2d}$, which is equivalent to the normalized graph Laplacian of its undirected counterpart. 
This definition gives rise to interesting generalizations, where the undirected normalized e-graph Laplacian can be expressed via the decomposition \[\tilde{{\bf \mathcal{L}}}_{\alpha}=\frac{\tilde{d}_{\alpha}}{2d}{\bf I}_N-\frac{{\bf A}+{\bf A}^H}{2d}=\frac{\tilde{{\bf S}}_{\alpha}+\tilde{{\bf S}}_{\alpha}^H}{2}=\frac{\tilde{{\bf S}}_{-\alpha}+\tilde{{\bf S}}^H_{-\alpha}}{2},\] with $\tilde{d}_{\alpha}=\tilde{d}_{+\alpha}+\tilde{d}_{-\alpha}$. In other words, the annihilation property of the e-graph Laplacian is preserved in the case of a directed circulant graph through the Hermitian transpose, where $\tilde{d}_{+\alpha}$ can be simultaneously interpreted as the degree which annihilates complex exponentials with exponent $+i\alpha$ on the graph of ${\bf A}$ and $-i\alpha$ on the graph of ${\bf A}^H$ (and vice versa for $\tilde{d}_{-\alpha}$). This ties in with our previous discussion on graph spline similarities and analogies, as in the undirected case, the (e-)graph Laplacian operator is a graph extension of a traditional second order derivative operator, thereby giving rise to graph spline-like functions and associated wavelets in degree steps of $2$, suggesting that the directed first-order graph difference operator $\tilde{{\bf S}}_{\alpha}$  provides an extension to the traditional first-order differential operator. For the directed cycle, we therefore ascertain a comprehensive analogy with the traditional spline and e-spline definitions. \\

As we are primarily interested in real valued node degrees and (symmetric) graph filters, we will continue to focus on the undirected case. 

\subsection{Complementary Graph (E-)Spline Wavelets}
The introduced wavelet transforms can annihilate (complex exponential) polynomial graph signals in the high-pass branch, yet do not reproduce such in the low-pass branch, unless the graph at hand is also bipartite, as shown in Cors. \ref{cor31} and \ref{cor32}. In addition, it becomes apparent that while the proposed filterbanks are well-defined in the analysis domain, they lack a straightforward synthesis representation. At last, we point to the fact that both low-and high-pass filters have compact support of the same length $2Mk+1$, based on the given ${\bf A}$, while the support of their corresponding synthesis filters is comparatively larger (exponentially decaying).\\

In light of this, we develop a new class of graph wavelet filterbanks on circulant graphs, by making use of traditional spectral factorization techniques ordinarily employed for the creation of biorthogonal perfect reconstruction filterbanks (also used in \cite{ortega3} to create bipartite spectral graph filters), in order to satisfy the additional desired properties. Here, within a generalized approach, one can specifically tailor the design of analysis and synthesis filters to incorporate variable reproduction and annihilation properties, while ensuring localization and compact support in the vertex domain. \\
In particular, since the filtering operation of a discrete-time signal in traditional signal processing can be defined as the matrix-vector product between a circulant matrix and the given signal-vector, spectral factorisation in the $z$-domain is directly applicable to circulant graphs, establishing a convenient analogy to the graph domain. Thus, we can achieve that the filter matrices in the analysis and synthesis branch are of finite and `balanced' bandwidth for $2Mk<N$. We further note that this type of filterbank, contrary to the preceding, facilitates only the standard alternating downsampling pattern on circulant graphs, where every other node is skipped.\\
\\
Let us consider the simple spline case first. Given analysis high-pass filter $H_{HP}(z)=\frac{l(z)^k}{(2d)^k}$ with $2k$ vanishing moments, we can determine the synthesis lowpass filter as $\tilde{H}_{LP}(z)=H_{HP}(-z)$ and thus, the analysis lowpass filter $H_{LP}(z)$ via the biorthogonality relations of a traditional filterbank \cite{biorref}, with $P(z)=H_{LP}(z)\tilde{H}_{LP}(z)$ subject to the constraint of the half-band condition\footnote{Since the set up of our filterbank is such that even-numbered nodes retain the low-pass and odd-numbered nodes the high-pass component, we are technically considering orthogonality between the shifted $z H_{HP}(z)$ and its dual $\tilde{H}_{LP}(z)=-z^{-1}((-z)H_{HP}(-z))$, but omit this notation for simplicity.} $P(z)+P(-z)=2$.
By requiring that the resulting filter $H_{LP}(z)=\sum_{i=0}^T r_i (z^i +z^{-i})$ is symmetric (on an undirected graph), we obtain the equality
\begin{equation} \label{eq:pz} P(z)=1+\sum_{i=0}^L p_{2i+1} (z^{2i+1}+z^{-(2i+1)})=\frac{1}{(2d)^k}\left(d-\sum_{i=1}^M (-1)^{i} d_i (z^i+z^{-i})\right)^k\left(\sum_{i=0}^T r_i (z^i +z^{-i})\right)\end{equation}
where $P(z)$ is a polynomial of odd powers. In addition, we may further impose the restriction that the analysis and synthesis filters have an equal number of vanishing moments $2k$, by setting $H_{LP}(z)=(z+1)^{k}(z^{-1}+1)^{k}R(z)$, where $R(z)$ is the polynomial to be determined. \\ For $k=1$, we require the highest degree of each side of Eq. (\ref{eq:pz}) to be $2L+1=M+T$, and consider the $L+1$ constraints $p_{2n}=0$, $n=1,...,L$, and $p_0=1$, and $T+1$ unknowns $r_i$, $i=0,...,T$. Thus, to obtain a unique solution to the resulting linear system, we require $L=T=\frac{M+T-1}{2}$, or $T=M-1$. For a higher-order filterbank with $k>1$, the constraints change as follows: $T=L=\frac{M k+T-1}{2}$, or $T=M k-1$. If we further impose that both synthesis and analysis filters have an equal number of vanishing moments, we need to include the additional factor $(z+1)^k(z^{-1}+1)^k$ for $H_{HP}(z)=\frac{1}{(2d)^k}l(z)^k $, and require $T=M k+k-1$. \\
The necessary existence of a complementary analysis low-pass filter for a given high-pass filter of graph $G$, follows from the B\'{e}zout theorem:
\begin{thmm}[B\'ezout \cite{esplinewav}] 
Given $C(z)\in\mathbb{R}[z]$, there exists a polynomial $D(z)\in\mathbb{R}[z]$ such that \[C(z)D(z)+C(-z)D(-z)=2\] if and only if $C(z)$ has neither zero as a root, nor a pair of opposite roots. In this case, there exists a unique polynomial $D_0(z)\in\mathbb{R}[z]$ satisfying the above and such that $deg D_0(z)\leq C(z)-1$. The set of all polynomials $D(z)\in\mathbb{R}[z]$ that satisfy the above is \[\{D_0(z)+z\lambda (z^2)C(-z),\quad\lambda(z)\in\mathbb{R}[z]\}\]
\end{thmm} 
\noindent In our case, we observe that $C(z)=H_{HP}(-z)$ cannot have a zero root since $d=\sum_{i=1}^M 2d_i> 0$ (as we only consider nonnegative weights $d_i\geq 0$). Furthermore, $C(z)$ contains pairs of opposing roots such that $H_{HP}(-z)=H_{HP}(z)$ if the generating set $S$ of the graph at hand contains only even elements; however, as we assume that $G$ is connected with $s=1\in S$, this cannot occur in our framework. 
We can equivalently resort to spectral factorization for bipartite graphs, however, since both $P(z)$ and $H_{HP}(z)$ are odd degree polynomials, $R(z)$ is required to be of higher degree $T$ than the remaining factor in $P(z)$ in order to produce a non-trivial solution. This gives rise to an underdetermined linear system, which can be uniquely solved by imposing additional constraints on the coefficients $r_i$ (such as roots at $z=-1$). The proposed biorthogonal graph wavelet constructions for circulant graphs are captured in the following theorem:
\begin{thm}\label{thm33}
Given the undirected, and connected circulant graph $G=(V,E)$ of dimension $N$, with adjacency matrix ${\bf A}$ and degree $d$ per node, we define the higher-order `complementary' graph-spline wavelet transform (HCGSWT) via the set of analysis filters:
\begin{equation} {\bf H}_{LP,an}\stackrel{(*)}{=}{\bf C}\bar{{\bf H}}_{LP}=\frac{1}{2^k}{\bf C}\left({\bf I}_N+\frac{{\bf A}}{d}\right)^k\end{equation}
\begin{equation}{\bf H}_{HP,an}=\frac{1}{2^k}\left({\bf I}_N-\frac{{\bf A}}{d}\right)^k\end{equation}
and the set of synthesis filters:
\begin{equation}{\bf H}_{LP,syn}=c_1{\bf H}_{HP,an} \circ {\bf \mathit{I}}_{HP}\end{equation}
\begin{equation}{\bf H}_{HP,syn}=c_2{\bf H}_{LP,an} \circ {\bf \mathit{I}}_{LP}\end{equation}
where ${\bf H}_{LP,an}$ is the solution to the system from Eq. $(7)$ under specified constraints, with coefficient matrix ${\bf C}$ arising from the relation ${\bf H}_{LP,an}\bar{{\bf H}}_{LP}^{-1}$ where applicable (see Cor. \ref{cor33}). Here, $\circ$ is the Hadamard product, $c_i,i\in\{1,2\}$ are normalization coefficients, and ${\bf \mathit{I}}_{LP/HP}$ circulant indicator matrices with first row of the form $\lbrack 1 \enskip-1\enskip 1\enskip -1\enskip...\rbrack$.
\end{thm}
\noindent \textit{Proof.} Follows from above discussion.\\
\\
As a result of spectral factorization, the shape and vertex spread of ${\bf H}_{LP,an}$ does not coincide exactly with the adjacency matrix of the graph (and its powers), but rather encompasses a subset $S_i\subseteq \textit{N}(i,\tilde{k})$ of vertices, per node $i$ within its $\tilde{k}$-hop local neighborhood, whereby $\tilde{k}$ depends on the initial constraints we impose on $H_{LP,an}(z)$.\\
We therefore establish a structural link to the analysis branch of the \textit{HGSWT} in Thm \ref{thm31}, in particular, to the higher-order low-pass graph filter $\bar{{\bf H}}_{LP}$ of Eq. (\ref{eq:t31}), which is based on the adjacency matrix, via the graph filter given by symmetric circulant coefficient matrix ${\bf C}$ in $(*)$;  here, ${\bf C}$ can be determined via matrix inversion of $\bar{{\bf H}}_{LP}$ (when $G$ is non-bipartite, see Cor. \ref{cor33} for $\alpha=0$).\\
Moreover, we note that given initial graph signal ${\bf p}\in\mathbb{R}^N$ and letting $\tilde{{\bf p}}\in\mathbb{R}^N$ denote the graph signal in the GWT domain, as in Sect. $3.1$, the synthesis stage can be expressed as follows:
\[\left(\frac{1}{2}\left(({\bf I}_N+{\bf K}){\bf H}_{LP,syn}\right)^T+\frac{1}{2}\left(({\bf I}_N-{\bf K}){\bf H}_{HP,syn}\right)^T\right)\tilde{{\bf p}}={\bf p},\]
for $K_{i,i}=1$ at even-numbered positions and $K_{i,i}=-1$ otherwise.
\\
\\
We proceed analogously for the graph e-spline case. According to (\cite{esplinewav}, Thm. $1$), for a scaling filter $H_j(z)$ at level $j$ to reproduce a function of the form $P(t) e^{\gamma_m t}$, where $deg P(t) \leq(L_m -1)$ and  $L_m$ is the multiplicity of $\gamma_m$, it is necessary and sufficient that the former be divisible by the term $R_{2^j \vec{\gamma}}(z)=\prod_{m=1}^M(1+e^{2^j\gamma_m} z^{-1})$, $\forall j\leq j_0-1$, with $\vec{\gamma}=(\gamma_1,...,\gamma_M)^T\in\mathbb{C}^M$, i.e. satisfying the generalized Strang-Fix conditions for suitable $\vec{\gamma}$. Here, $H_j(z)$ must not contain roots of opposite sign.\\
\\
Mirroring the constructions in Thm. \ref{thm33} and given analysis high-pass filter $H_{HP_{\alpha}}(z)=\frac{\tilde{l}_{\alpha}(z)}{2d}$ with $2$ vanishing exponential moments, we determine the analysis lowpass filter $H_{LP_{\alpha}}(z)$, which can be expressed as an extension of Eq. (\ref{eq:t32}) via a coefficient matrix ${\bf C}$ (subject to constraints, see Cor. \ref{cor33}). By imposing the constraints of B\'{e}zout's Thm. \cite{esplinewav}, and setting $P(z)=H_{LP_{\alpha}}(z)H_{HP_{\alpha}}(-z)$, we arrive at an equality of the form
\begin{equation} \label{eq:ep} P(z)=1+\sum_{i=0}^L p_{2i+1} (z^{2i+1}+z^{-(2i+1)})=\frac{1}{2d}\left(\tilde{d}_{\alpha}-\sum_{i=1}^M (-1)^{i} d_i (z^i+z^{-i})\right)\left(\sum_{i=0}^T r_i (z^i +z^{-i})\right)\end{equation}
and solve the emerging linear system in a similar fashion as discussed for Eq. (\ref{eq:pz}) for unknown symmetric coefficients $r_i$ of $H_{LP_{\alpha}}(z)$.\\
Moreover, for the analysis and synthesis filters to have (an equal number of) vanishing moments, we require $H_{LP_{\alpha}}(z)=(z+ e^{i\alpha })(1+ e^{-i\alpha }z^{-1}) R(z)$, where $R(z)$ is the polynomial to be determined. This scheme can be generalised to higher order for $\vec{\alpha}=(\alpha_1,...,\alpha_T)$ such that $H_{HP_{\vec{\alpha}}}(z)=\prod_{n=1}^T\frac{\tilde{l}_{\alpha_n}(z)^k}{(2d)^k}$, and $H_{LP_{\vec{\alpha}}}(z)=\prod_{n=1}^T(z+2\cos(\alpha_n)+z^{-1})^k R(z)$, however, a solution $R(z)$ exists only if the remainder term in $P(z)$ does not contain zero and/or opposing roots \cite{esplinewav}. \\
The possibility of a multiresolution representation of the filterbank, is conditional upon the existence of real-valued filters, which maintain their reproduction/annihilation properties up to a certain level $j\leq J-1$. In the classical domain, the filters of a non-stationary biorthogonal exponential wavelet filterbank with exponent $\vec{\alpha}=(\alpha_1,...,\alpha_T)$ (as described in \cite{esplinewav}) do not contain roots of opposite sign nor the zero root at level $j$, as long as there are no distinct $\alpha, \alpha'$ in $\vec{\alpha}$ that satisfy $2^j(\alpha-\alpha')=i(2k+1)\pi$, for some $j\leq J-1$ and $k\in\mathbb{Z}$. Otherwise, a multilevel representation is only possible up to a finite level $J-1$, when this condition ceases to be fulfilled. This result becomes particularly relevant when considering e-spline graph wavelet filterbanks of higher order. For instance, opposing roots occur at $j=0$ for the case $e^{-i\alpha}=-e^{i\alpha}$, at $\alpha=\pi/2, 3\pi/2$, and we cannot create filterbanks for these parameters. 
The discussed approach gives rise to the following filterbank:
\begin{thm}\label{thm34}
Given the undirected, and connected circulant graph $G=(V,E)$ of dimension $N$, with adjacency matrix ${\bf A}$ and degree $d$ per node, we define the higher-order `complementary' graph e-spline wavelet transform (HCGESWT) via the set of analysis filters:
\begin{equation} {\bf H}_{LP_{\vec{\alpha}},an}\stackrel{(*)}{=}{\bf C}\bar{{\bf H}}_{LP_{\vec{\alpha}}}={\bf C}\prod_{n=1}^T\frac{1}{2^k}\left(\beta_n {\bf I}_N+\frac{{\bf A}}{d}\right)^k\end{equation}
\begin{equation}{\bf H}_{HP_{\vec{\alpha}},an}=\prod_{n=1}^T\frac{1}{2^k}\left(\beta_n{\bf I}_N-\frac{{\bf A}}{d}\right)^k\end{equation}
and the set of synthesis filters:
\begin{equation}{\bf H}_{LP_{\vec{\alpha}},syn}=c_1{\bf H}_{HP_{\vec{\alpha}},an} \circ {\bf \mathit{I}}_{HP}\end{equation}
\begin{equation}{\bf H}_{HP_{\vec{\alpha}},syn}=c_2{\bf H}_{LP_{\vec{\alpha}},an} \circ {\bf \mathit{I}}_{LP}\end{equation}
where ${\bf H}_{LP_{\vec{\alpha}},an}$ is the solution to the system from Eq. (\ref{eq:ep}) for $\vec{\alpha}$ under specified constraints, with coefficient matrix ${\bf C}$ arising from the relation ${\bf H}_{LP_{\vec{\alpha}},an}{\bf \bar{H}}_{LP_{\vec{\alpha}}}^{-1}$ where applicable (see Cor. \ref{cor33}). Here, $c_i,i\in\{1,2\}$ are normalization coefficients, and ${\bf \mathit{I}}_{LP/HP}$ are circulant indicator matrices with first row of the form $\lbrack 1 \enskip-1\enskip 1\enskip -1\enskip...\rbrack$.
\end{thm}
This coincides with the previous \textit{HCGSWT} in the simple spline case for $\vec{\alpha}={\bf 0}$.
\section{Graph Products and Approximations: A Multidimensional Extension}
In order to facilitate a generalization of our developed framework to arbitrary graphs, we require a means to compute circulant graph approximations to existing structures in a given network.\\
We have previously resorted to employing an adjacency matrix approximation scheme (\cite{globalsip},\cite{spie}), which determines the nearest circulant graph approximation $\tilde{G}$ to the given graph $G$ with adjacency matrix ${\bf A}\in\mathbb{R}^{N\times N}$ by minimizing the error norm $\min _{\tilde{{\bf A}}\in C_N}||{\bf A}^P-\tilde{{\bf A}}||_F$ over the space $C_N$ of all $N\times N$ circulant matrices. If $G$ is sparse or a posteriori sparsified by removing edges of small weight, the approximation can be subjected to a prior node relabelling $P$ based on the RCM-algorithm \cite{rcm} in order to minimize the bandwidth of ${\bf A}$. This facilitates a restructuring such that ${\bf A}^P$ is (locally) closer to circulant (sub-)structures and hence reduces the number of complementary edges in $\tilde{{\bf A}}$. The closed-form solution is therefore obtained as \[\tilde{{\bf A}}=\sum_{i=0}^{N-1}\frac{1}{N}\langle {\bf A}^P,{\bf \Pi}^i\rangle_F{\bf \Pi}^i,\]
for circulant permutation matrix ${\bf \Pi}$ with first row $\lbrack 0\enskip 1\enskip 0...\rbrack$, and the graph signals residing on $G$ can then be analyzed with respect to $\tilde{G}$. For an arbitrary graph $G$ featuring communities, we propose to perform graph partitioning (e.g. the normalized graph cut \cite{cut}) and compute the nearest circulant structures to the arising subgraphs $\{G_i\}_{i}$, so as to ultimately conduct wavelet analysis on the latter with respect to the partitioned subgraph signals $\{{\bf x}_i\}_i$. Fig. $5$ shows the resulting graph approximations and multiscale representation via the \textit{HGSWT} (at $k=1$, no reconnection) for a data-driven graph with weights \[w_{i,j}=e^{-\frac{d(x_i,x_j)^2}{\sigma^2}},\quad d(x,y)=|x-y|,\enskip i,j=0,...,N-1\]
with $\sigma$ as $10\%$ of the total range of $d(x,y)$, and random graph signal ${\bf x}$, prepared using \cite{toolbox}. The obtained representation is highly sparse as a consequence of the breadth-first traversal of the RCM algorithm, whereby ${\bf x}^P$ has reduced total variation $||{\bf x}^P||_{TV}$ 
\cite{spie}, and simultaneously ${\bf A}^P$ (and by extension $\tilde{{\bf A}}$) is of minimum bandwidth.\\
\begin{figure}[htb]
\centering
\begin{minipage}{3.5in}
	{\includegraphics[width=1.7in]{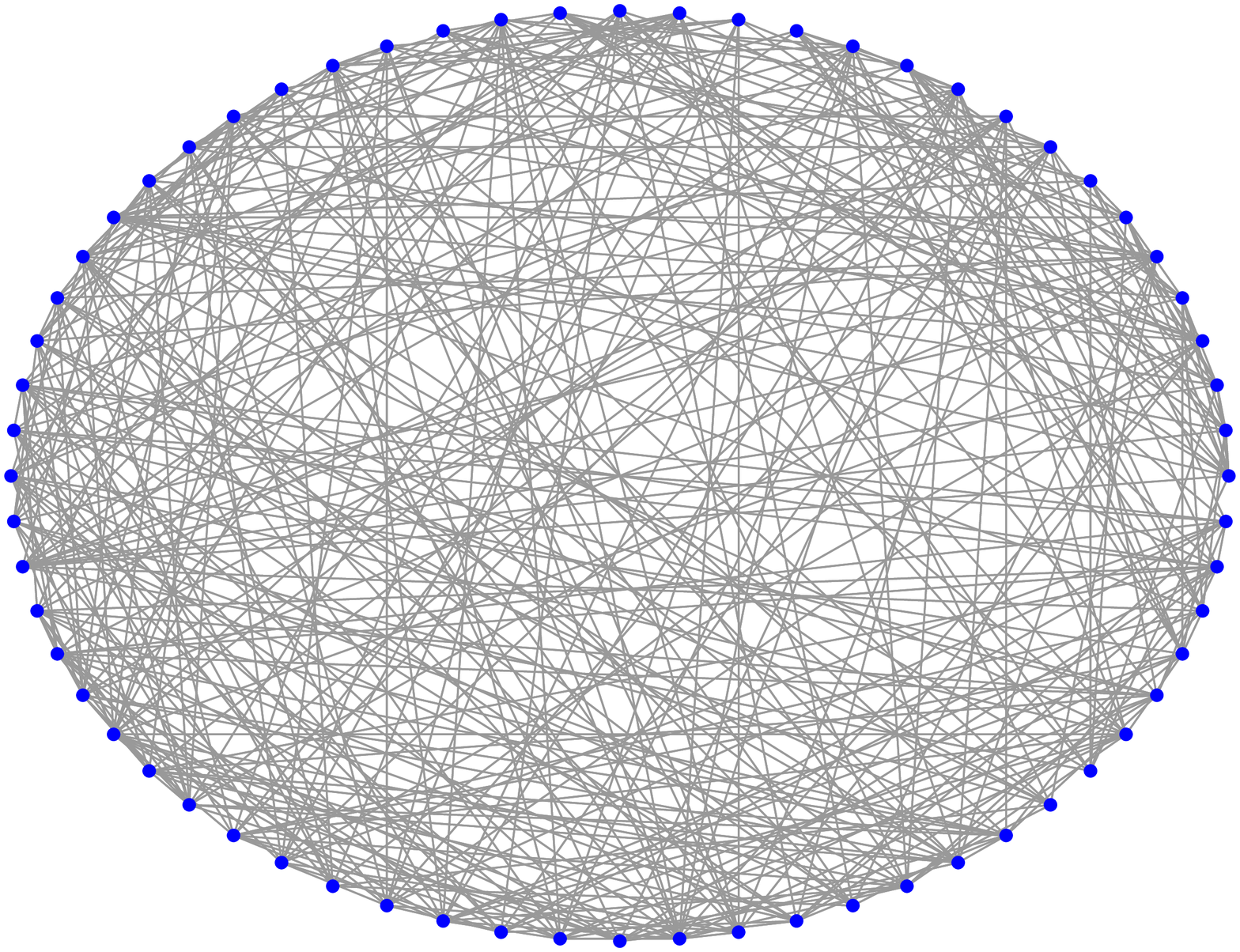}}%
	{\includegraphics[width=1.7in]{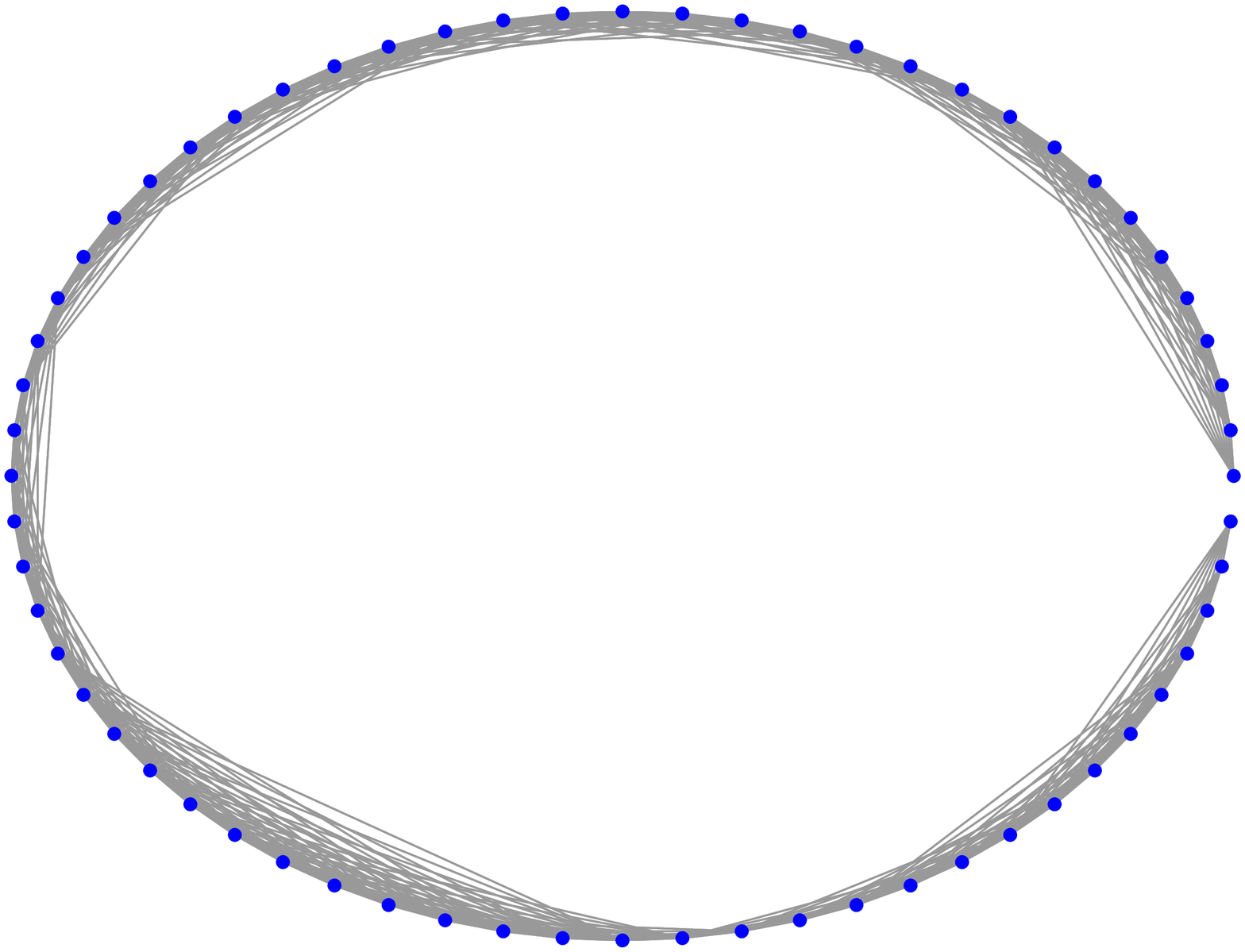}}\\
	\centering 
	{\includegraphics[width=2.2in]{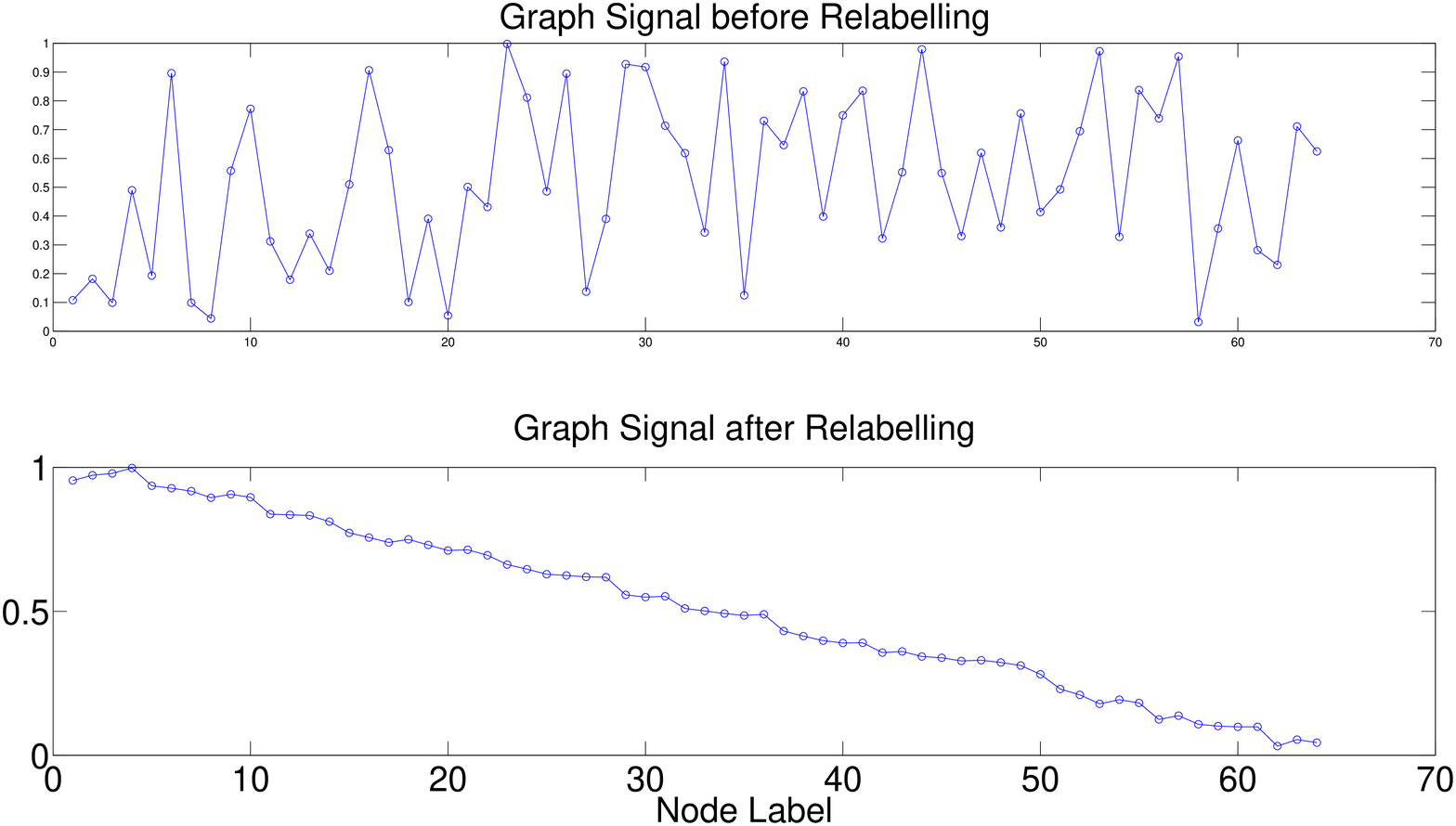}}
	\end{minipage}%
	\begin{minipage}{3.5in}
	{\includegraphics[width=3.1in]{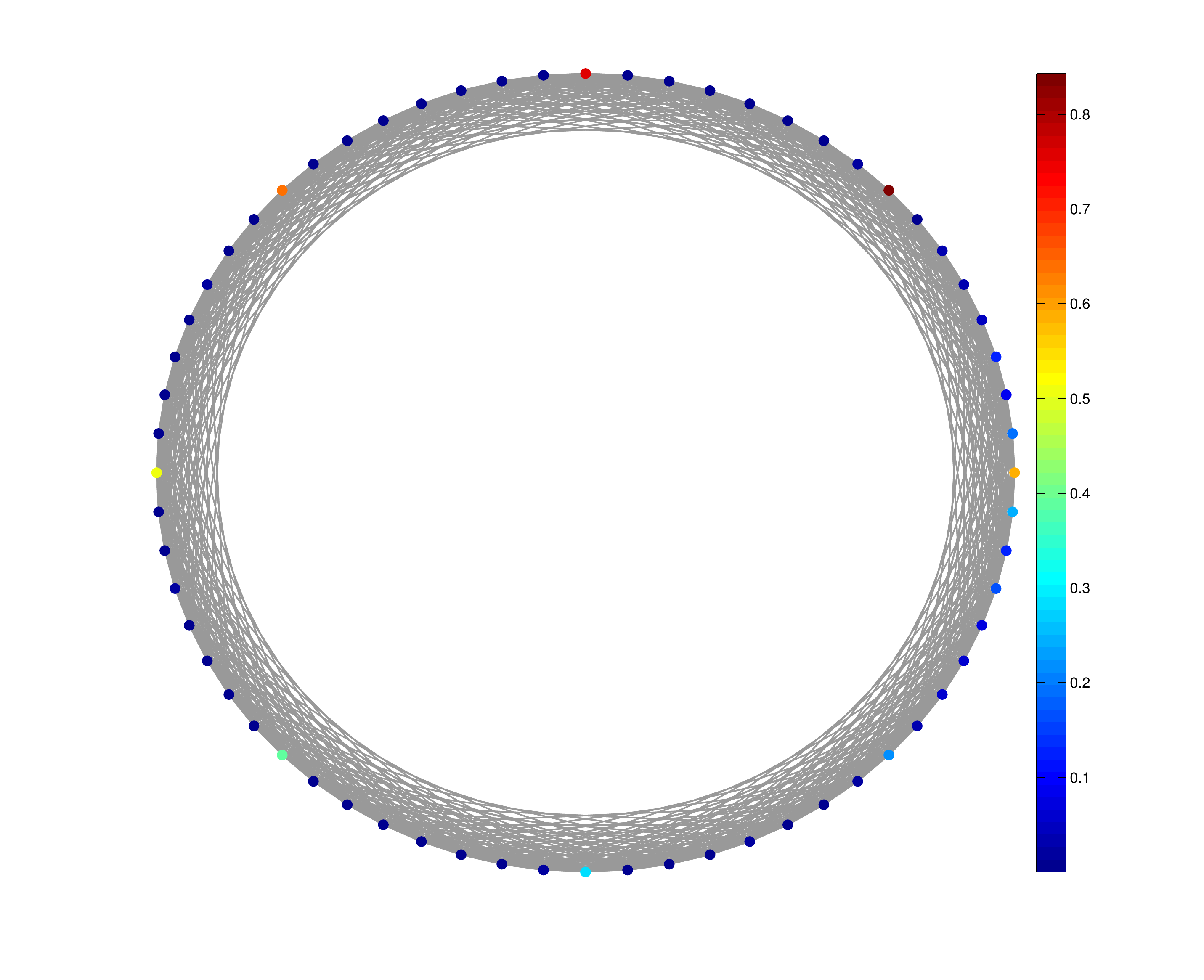}}
	
	\end{minipage}

		\caption{Original $G$ $(N=64)$ after thresholding of weights (top left), after RCM relabelling (top middle), multiscale \textit{HGSWT} representation (at $k=1$) of ${\bf x}$ (in magnitude) on $\tilde{G}$ for 3 levels (right), signal ${\bf x}$ before/after relabelling (bottom).}
\end{figure}
\\
As part of a more generalized motivation which facilitates the multi-(and lower-)dimensional processing and representation of signals on graphs, we wish to explore alternative approximation schemes for circulant graphs, and identify graph product approximations as a promising venue. In particular, given an arbitrary undirected graph, we consider its approximation as the graph product of circulant graphs.
\\
Graph products \cite{handbook} have been studied and applied in a variety of contexts for purposes of i.a. modelling realistic networks, and/or rendering matrix operations computionally efficient (\cite{kronapprox}, \cite{kronnetwork}, \cite{pits}). Their relevance for GSP was first discussed in \cite{bigdata} as a means of modelling and representation of complex data as graph signals defined on product graphs, with the potential of promoting a more efficient implementation of graph operations, such as graph filtering. Our motivation for considering graph products is twofold: 
$(1)$ we require a scheme which can decompose arbitrary graphs into circulant graphs, so as to facilitate the processing of graph signals with respect to the circulant approximations, 
and $(2)$ we wish to conduct operations in lower dimensional settings to increase efficiency. 
\subsection{Graph Products of Circulants}
The product $\diamond$ of two graphs $G_1=(V(G_1),E(G_1))$ and $G_2=(V(G_2),E(G_2))$, also referred to as factors, with respective adjacency matrices ${\bf A}_1\in\mathbb{R}^{N_1\times N_1}$ and ${\bf A}_2\in\mathbb{R}^{N_2\times N_2}$, is formed by letting the Cartesian product $V(G)=V(G_1)\times V(G_2)$ denote the new vertex set of the resulting graph $G$, and defining the new edge relations $E(G)$ according to the characteristic adjacency rules of the product operation, resulting in adjacency matrix ${\bf A}_{\diamond}\in\mathbb{R}^{N_1N_2\times N_1N_2}$. We identify four main graph products of interest:
\begin{itemize}
\item Kronecker product $G_1\otimes G_2$: ${\bf A}_{\otimes}={\bf A}_1\otimes {\bf A}_2$
\item Cartesian product $G_1\times G_2$: ${\bf A}_{\times}={\bf A}_1\times {\bf A}_2={\bf A}_1\otimes {\bf I}_{N_2} +{\bf I}_{N_1}\otimes {\bf A}_2$
\item Strong product $G_1\boxtimes G_2$: ${\bf A}_{\boxtimes}={\bf A}_1\boxtimes {\bf A}_2={\bf A}_1\otimes {\bf A}_2+{\bf A}_1\times {\bf A}_2$
\item Lexicographic product $G_1\lbrack G_2\rbrack$: ${\bf A}_{\lbrack \enskip \rbrack}={\bf A}_1\lbrack{\bf A}_2\rbrack={\bf A}_1\otimes {\bf J}_{N_2} +{\bf I}_{N_1}\otimes {\bf A}_2$
\end{itemize}
\noindent where ${\bf J}_{N_2}={\bf 1}_{N_2} {\bf 1}_{N_2}^T$. In particular, we note that the lexicographic product can be regarded as a variation of the Cartesian product, yet contrary to the others, it is not commutative for unlabelled graphs \cite{handbook}. In addition, the adjacency matrices ${\bf A}_{\diamond}$ for the first three products possess the same eigenbasis ${\bf V}={\bf V}_1\otimes {\bf V}_2$, for decompositions ${\bf A}_1={\bf V}_1{\bf \Gamma}_1{\bf V}_1^{H}$ and ${\bf A}_2={\bf V}_2{\bf \Gamma}_2{\bf V}_2^{H}$, and eigenvalues of the form ${\bf \Gamma}_{\diamond}={\bf \Gamma}_1\diamond {\bf \Gamma}_2$ (\cite{handbook}). \\
\\
Graph products have been employed to model realistic networks, due to their ability to capture present regularities such as patterns and recursive community growth \cite{kronnetwork}, and can thus provide suitable approximations to networks with inherent substructures, such as social networks consisting of similarly structured communities or time-evolving sensor networks \cite{bigdata}. For our ensuing analysis, we consider a scheme for arbitrary graphs which imposes the desired constraint of circularity on the individual factors \footnote{On a related note, \cite{kronapprox} identifies Kronecker product approximation as a means to facilitate efficiency for large structured least-squares problems in image restoration, which coincidentally marks an area where circulant approximations are used extensively.}.
\subsubsection{The Kronecker product approximation}
Given an arbitrary graph $G$ with adjacency matrix ${{\bf A}}$, we resort to a result from matrix theory \cite{loankron} which facilitates the approximate Kronecker product decomposition ${{\bf A}}\approx{\bf A}_1\otimes {\bf A}_2$ into circulant (adjacency) matrices ${\bf A}_i$ of suitably chosen dimension $N_i$, by solving the convex optimization problem 
\[\min_{{\bf C}_1^T vec({\bf A}_1)=0, {\bf C}_2^T vec({\bf A}_2)=0} ||{\bf A}-{\bf A}_1\otimes {\bf A}_2||_F\]
subject to linear constraints in the form of structured, rectangular matrices ${\bf C}_i$ with entries $\{0,1,-1\}$, which impose circularity on ${\bf A}_i$ via column-stacking operator $vec$. 
It can be shown that closed-form solutions $vec({\bf A}_i), i=1,2$ are obtained by solving a reduced unconstrained problem, after expressing the above as a rank-1 approximation problem (see \cite{loankron} for details).
In addition, we may also impose symmetry and bandedness (\cite{loankron}, \cite{pits}).\\
\subsubsection{Exact graph products}
Conversely, the general graph product of circulants gives rise to block-circulant structures (or sums thereof), and an example can be seen in Fig. $6$.
\begin{figure}[htb]\centering
\begin{minipage}{6in}
  \centering
  \raisebox{-0.5\height}{\includegraphics[height=1.2in]{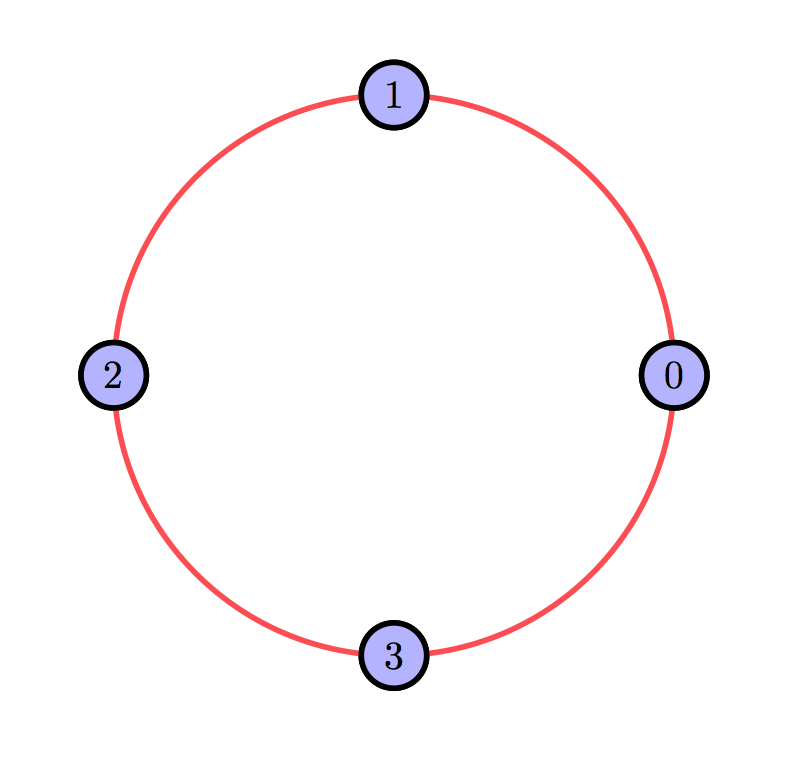}}
  \hspace*{.02in}
  ${\bf \times}$
  \hspace*{.06in}
 \raisebox{-0.4\height}{\includegraphics[height=1.2in]{circ2.png}}
 \hspace*{.03in}
 {\bf =}
 \hspace*{.03in}
  \raisebox{-0.5\height}{\includegraphics[height=2.3in]{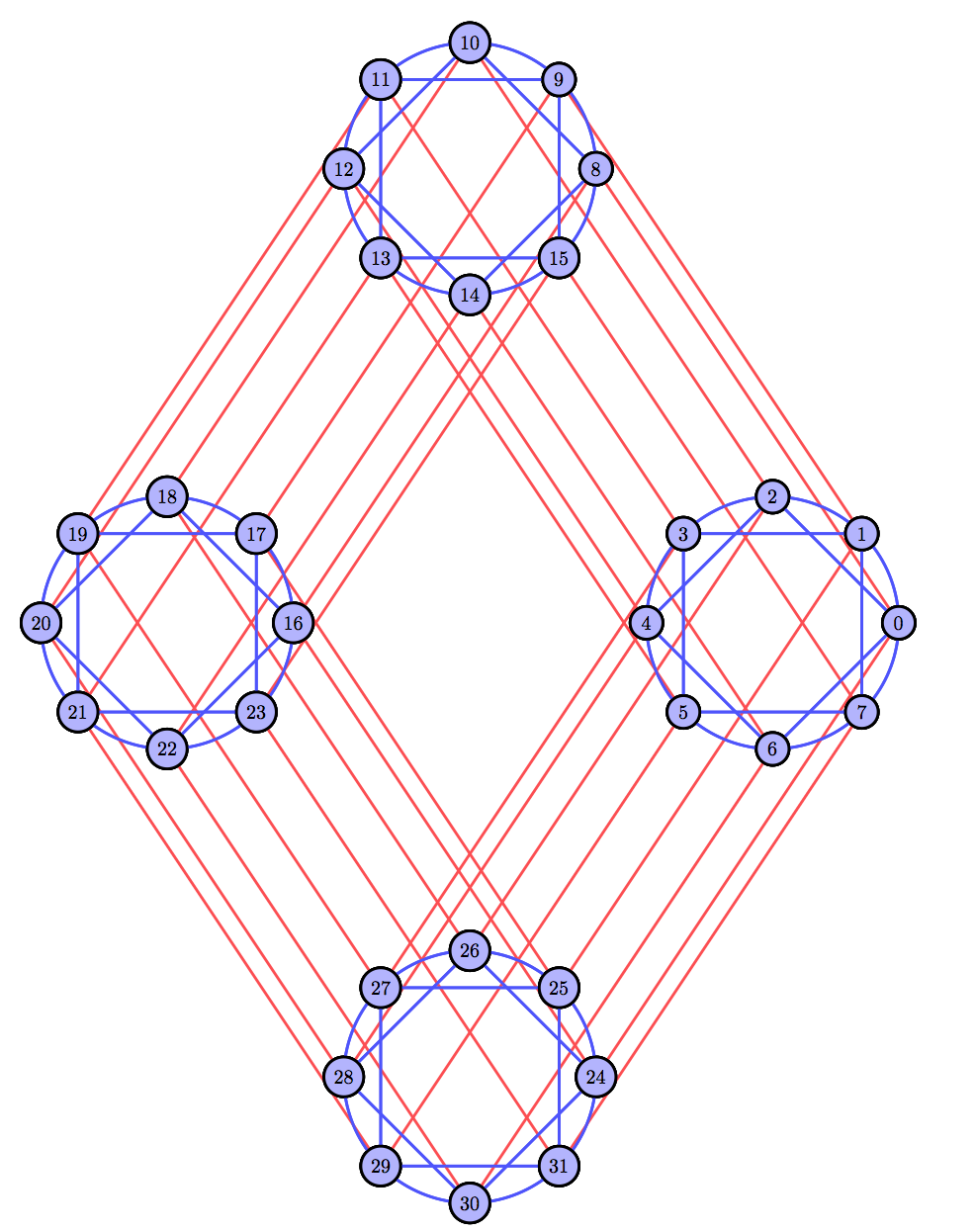}}
\end{minipage}
\caption{Graph Cartesian Product of two unweighted circulant graphs}
\label{fig:res}
\end{figure}
\\
There exists a subset of circulant graphs which can be represented as the graph products of circulant factors; while these cases are marginal, they similarly motivate decompositions for lower-dimensional processing.\\ 
Circulant graphs are not generally closed under graph products, with the exception of the lexicographic product \cite{lex}. In particular, for two circulant graphs $G_1=C_{N_1,S_1}$ and $G_2=C_{N_2,S_2}$ of respective dimensions $N_1$ and $N_2$ and with generating sets $S_1$ and $S_2$, the product $G_1\lbrack G_2\rbrack$ is isomorphic to the circulant graph $C_{N_1 N_2,S}$ with generating set $S=\left(\cup^{\lfloor\frac{N_2-1}{2}\rfloor}_{t=0} tN_1+S_1\right)\cup \left(\cup^{\lfloor\frac{N_2}{2}\rfloor}_{t=1} tN_1-S_1\right)\cup N_1 S_2$ \cite{lex}. Here, $C_{N_1N_2,S}$ is connected with $1\in S$ only if $G_1$ is connected with $1\in S_1$. The adjacency matrix ${\bf A}_{\lbrack\enskip\rbrack}$ is not circulant, but its isomorphism $\tilde{{\bf A}}_{\lbrack \enskip\rbrack}={\bf P}{\bf A}_{\lbrack\enskip \rbrack} {\bf P}^T$ is, where permutation matrix ${\bf P}$ performs the relabelling $\{0,...,N_1 N_2 -1\}\rightarrow \{0:N_2:N_1N_2-1, 1:N_2:N_1N_2-1,...,N_2-1:N_2:N_1N_2-1\}$ such that each product node $(g_{1,j},g_{2,k})\in V(G)$ is labelled as $g_{1,j}+N_1g_{2,k}$, for $g_{i,j}\in V(G_i)$.\\
Special cases of other graph products that are circulant with circulant factors are discussed in \cite{tensor}.

\subsection{Multi-dimensional Wavelet Analysis on Product Graphs}
In the following, we explore how the theory on circulant graph wavelet analysis can be extended to product graphs. Here, we operate under the assumption that the decomposition (and decomposition type) of an arbitrary graph into circulants is either known (exact or approximate), or unknown, in which case we can always resort to a Kronecker product approximation.
Before we can proceed, we need to define the graph Laplacian ${\bf L}_{\diamond}$ of product graphs as a relevant high-pass filter; its interpretation as an extension of the circulant graph Laplacian high-pass filter to higher dimensions, with associated property preservations, will be revisited in Sect 4.3. We note that the formation of ${\bf L}_{\diamond}$ is not a reflection of the adjacency matrix relations, except in the case of the Cartesian product (\cite{laplaceproduct1}, \cite{laplaceproduct2}):
\begin{itemize}
\item Kronecker product: ${\bf L}_{\otimes}={\bf D}_1\otimes {\bf D}_2-{\bf A}_1\otimes {\bf A}_2={\bf L}_1\otimes {\bf D}_2+{\bf D}_1\otimes {\bf L}_2-{\bf L}_1\otimes {\bf L}_2$
\item Cartesian product: ${\bf L}_{\times}={\bf D}_1\times {\bf D}_2-{\bf A}_1\times {\bf A}_2={\bf L}_1\otimes {\bf I}_{N_2} +{\bf I}_{N_1}\otimes {\bf L}_2$
\item Strong product: ${\bf L}_{\boxtimes}={\bf D}_1\boxtimes {\bf D}_2-{\bf A}_1\boxtimes {\bf A}_2={\bf L}_{\otimes}+{\bf L}_{\times}$
\item Lexicographic product: ${\bf L}_{\lbrack \enskip \rbrack}={\bf D}_1\lbrack{\bf D}_2\rbrack-{\bf A}_1\lbrack{\bf A}_2\rbrack={\bf I}_{N_1}\otimes {\bf L}_2+{\bf L}_1\otimes {\bf J}_{N_2}+{\bf D}_1\otimes (N_2 {\bf I}_{N_2}-{\bf J}_{N_2})$
\end{itemize}
For connected, regular graph factors $G_i$, equivalent relations between eigenbases ${\bf U}$ of ${\bf L}_{\diamond}={\bf U}{\bf \Lambda}_{\diamond}{\bf U}^H$ and ${\bf U}_i$ of ${\bf L}_i={\bf U}_i{\bf \Lambda}_i{\bf U}_i^{H}$ as for the adjacency matrices hold, with ${\bf U}={\bf U}_1\otimes {\bf U}_2$ and ${\bf \Lambda}_{\diamond}={\bf D}_{\diamond}- {\bf \Gamma}_{\diamond}$  \cite{laplaceproduct1}. These relations are further preserved for the Cartesian product when the $G_i$ are generic, while a nearer characterization for the remaining cases is still subject to investigation (\cite{laplaceproduct2}, \cite{laplaceproduct1}). The eigenvectors of ${\bf L}_{\lbrack\enskip\rbrack}$ are given by ${\bf U}_{\lbrack\enskip \rbrack}=\lbrack \{{\bf u}_{1,i}\otimes {\bf 1}_{N_2} \}_{i=1}^{N_1}|\{{\bf e}_i\otimes {\bf u}_{2,j}\}_{i=1,j=2}^{i=N_1,j=N_2}\rbrack$, when $G_1$ is connected \cite{laplaceproduct2}.\\
Therefore, under the first three products, the special case of two circulant graph factors, for which ${\bf U}={\bf V}$ can be represented as the 2D DFT matrix, reveals that each graph Laplacian eigenvector ${\bf u}_j$ of ${\bf L}_{\diamond}$ is the Kronecker product of the graph Laplacian eigenvectors of its factor graphs. This insight motivates the following graph signal definition for GSP:
\begin{defe}\label{def41}
Any graph signal ${\bf x}\in\mathbb{R}^N$, with $N=N_1 N_2$, can be decomposed as ${\bf x}=\sum_{s=1}^k{\bf x}_{s,1}\otimes {\bf x}_{s,2}=vec_r\{\sum_{s=1}^k{\bf x}_{s,1}{\bf x}_{s,2}^T\}$, where $vec_r\{\}$ indicates the row-stacking operation, or, equivalently, $\sum_{s=1}^k{\bf x}_{s,1}{\bf x}_{s,2}^T$ has rank $k$ with ${\bf x}_{s,i}\in\mathbb{R}^{N_i}$. For ${\bf x}$ residing on the vertices of an arbitrary undirected graph $G$, which admits the graph product decomposition of type $\diamond$, such that $G_{\diamond}=G_1 \diamond G_2$ and $|V(G_i)|=N_i$, we can redefine and process ${\bf x}$ as the graph signal tensors ${\bf x}_{s,i}$ on $G_i$.
\end{defe}
While ${\bf x}$ does not generally lie in the graph Laplacian eigenspace of the underlying graph, and alternative decompositions are possible, we adopt the above perspective as a promising interpretation of component-wise processing of graph signals defined on product graphs. We inspect the case ${\bf x}={\bf x}_1\otimes {\bf x}_2$ (for rank $k=1$) more closely in Sect. $4.3$, as it facilitates concrete claims on the smoothness and sparsity relations between a signal and its tensor components on a graph.

\subsubsection{Separable vs Non-separable Wavelet Analysis}
Let a graph signal ${\bf x}$ reside on the vertices of an arbitrary, undirected graph $G$ with graph product decomposition $G_{\diamond}=G_1\diamond G_2$, which can be exact or approximate, such that the factors $G_i$ are circulant with adjacency matrices ${\bf A}_i\in\mathbb{R}^{N_i\times N_i},i=1,2$ and connected with $s=1\in S_i,\enskip i=1,2$. We propose a non-separable and a separable wavelet transform on $G_{\diamond}$: the former operates on the product graph directly, while the latter operates on each factor graph independently, thereby omitting the inter-connections between the two factors arising through the graph product operation.\\
\\
We define the \textit{non-separable} graph wavelet transform on $G_{\diamond}$ with (symmetric) adjacency matrix ${\bf A}_{\diamond}$ as 
\[{\bf H}_{\diamond}=\frac{1}{2}({\bf I}_N+{\bf K})\prod_{n=1}^T \frac{1}{2^{k}}\left(\beta_{\diamond,n}{\bf I}_N +\frac{{\bf A}_{\diamond}}{d}\right)^k+\frac{1}{2}({\bf I}_N-{\bf K})\prod_{n=1}^T\frac{1}{2^{k}}\left(\beta_{\diamond,n}{\bf I}_N -\frac{{\bf A}_{\diamond}}{d}\right)^k,\enskip k\in\mathbb{N}.\]
If $\beta_{\diamond,n}=1,\forall n$, this is verifiably invertible for any downsampling pattern ${\bf K}$ as long as at least one low-pass component is retained, as a generalization of the circulant \textit{HGSWT} construction in Thm \ref{thm31}. Here, the fundamental properties which ensure this extension are that $G_{\diamond}$ is undirected, regular and connected \cite{brouwer12}, i.e. the spectrum of $\frac{{\bf A}_{\diamond}}{d}$ is such that $|\gamma_{\diamond, i}|\leq \gamma_{\diamond, max}=1$, with $\frac{{\bf A}_{\diamond}}{d}{\bf 1}_{N_1 N_2}={\bf 1}_{N_1 N_2}$ and $\gamma_{\diamond, max}$ of multiplicity $1$. It is known that $G_{\diamond}$ is connected under the Cartesian product for connected $G_i$ and under the Kronecker product, if in addition at least one $G_i$ is non-bipartite \cite{prod2}. 

Wavelet constructions with exponential degree parameters of the type of Thm. \ref{thm32} can be similarly extended to the first three product graph types under equivalent restrictions for parameters $|\beta_{\diamond, n}|\neq|\gamma_{\diamond, i}|$, i.e. conditions $(i)$-$(ii)$ of Thm. \ref{thm32} apply to the above. In particular, when $\beta_{\diamond, n}$ is of the form of an eigenvalue $\gamma_{\diamond, i}$ of $\frac{{\bf A}_{\diamond}}{d}$, invertibility of the transform is similarly conditional on the downsampling pattern and linear independence of the associated sampled eigenvectors.
Here, the multi-dimensional e-degree matrix $\beta_{\diamond, n}{\bf I}_N=\frac{\tilde{{\bf D}}_{\alpha_1} \diamond \tilde{{\bf D}}_{\alpha_2}}{d}=\frac{(\tilde{d}_{\alpha_1}{\bf I}_{N_1})\diamond (\tilde{d}_{\alpha_2}{\bf I}_{N_2})}{d}$ of the graph filters (transforms) may be tailored to the analysis of multi-dimensional smooth signals residing on the vertices of the graph product $G_{\diamond}=G_1\diamond G_2$, whose circulant (e-graph Laplacian) factors are respectively parameterized by e-degrees (or eigenvalues) $\tilde{d}_{\alpha_1}$ and $\tilde{d}_{\alpha_2}$; the relations between the graph Laplacian eigenvalues on each factor are specifically derived in the following Sect. $4.3$ as measures of signal smoothness on product graphs. 
As such, the non-separable GWT on a product graph of circulants constitutes a generalization of the \textit{HGESWT} to higher dimensional graphs, which incorporates the parameterization desirable for the analysis of the signal partitions (tensor factors ${\bf x}_i$ residing on $G_i$) of the multi-dimensional signal ${\bf x}$.

While the graph product $G_{\diamond}$ of circulants is not LSI as such, it is invariant with respect to shifts on its factors, i.e. matrix ${\bf P}_{\otimes}={\bf P}_{N_1}\otimes{\bf P}_{N_2}$, for circulant permutation matrices ${\bf P}_{N_i}$, commutes with filters on ${\bf A}_{\diamond}$. Therefore, we can conduct multiresolution analysis with respect to product graphs by performing downsampling and graph coarsening operations on the individual factors, with one level corresponding to operating on either $G_i$. 
For instance, given the product graph $G_{\diamond}$ of Fig. $6$, one may choose to downsample by $2$ on factor $G_2$ with respect to $s=1\in S_2$ and define the associated downsampling matrix ${\bf K}_2$; this creates the downsampling pattern ${\bf K}={\bf I}_{N_1}\otimes {\bf K}_2$ on $G_{\diamond}$, which skips every other node within each block of $G_2$ in $G_{\diamond}$, and one may subsequently redefine the sampled low-pass output on $G_1\diamond\tilde{G}_2$, where $\tilde{G}_2$ corresponds to the coarsened version of $G_2$ (see Fig. $7$).
\begin{figure}[htb]\centering
\begin{minipage}{6in}
  \centering
   \raisebox{-0.5\height}{\includegraphics[height=1.8in]{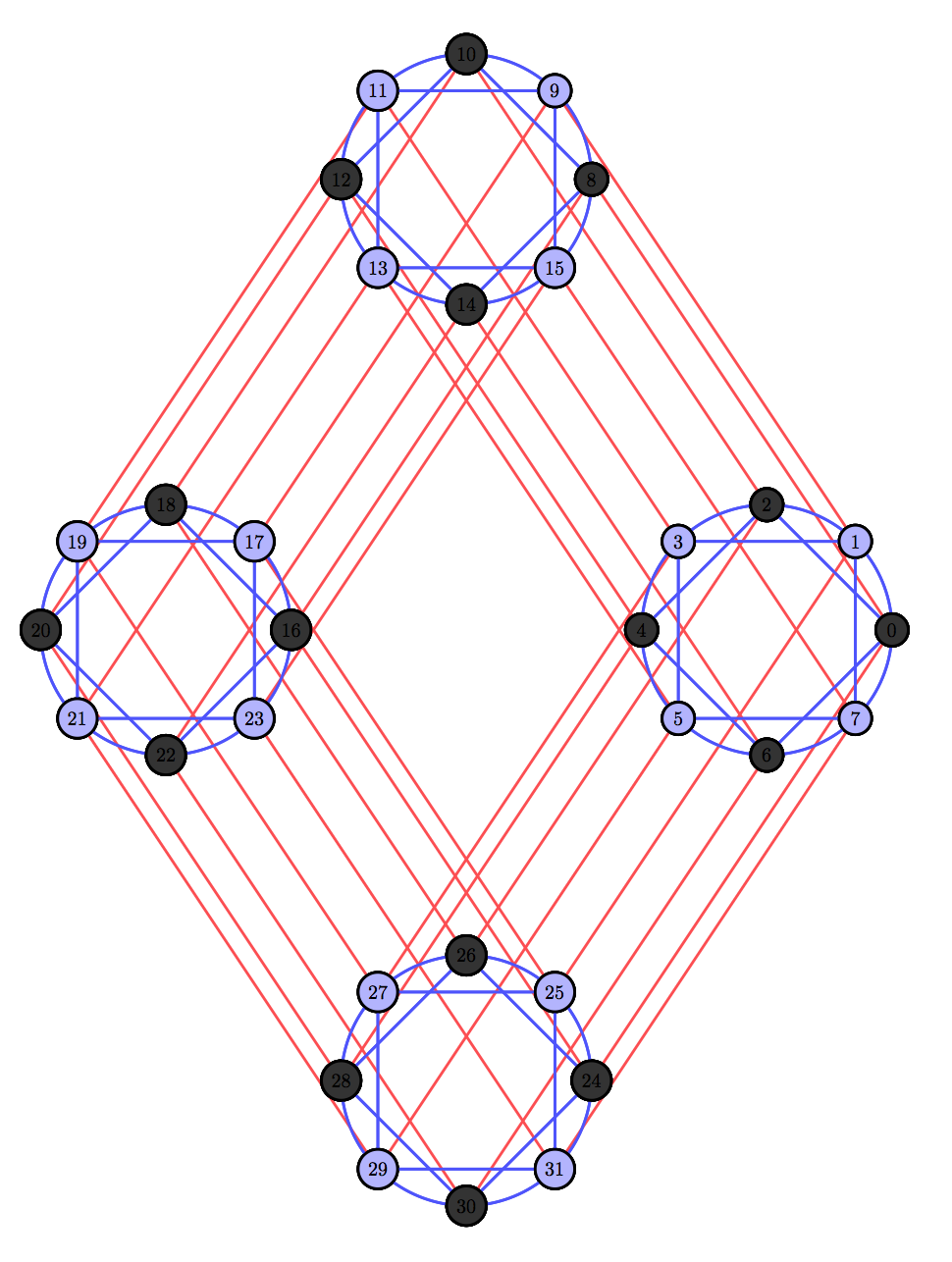}}
   \hspace*{.03in}
 ${\bf \rightarrow}$
 \hspace*{.03in}
  \raisebox{-0.5\height}{\includegraphics[height=1.8in]{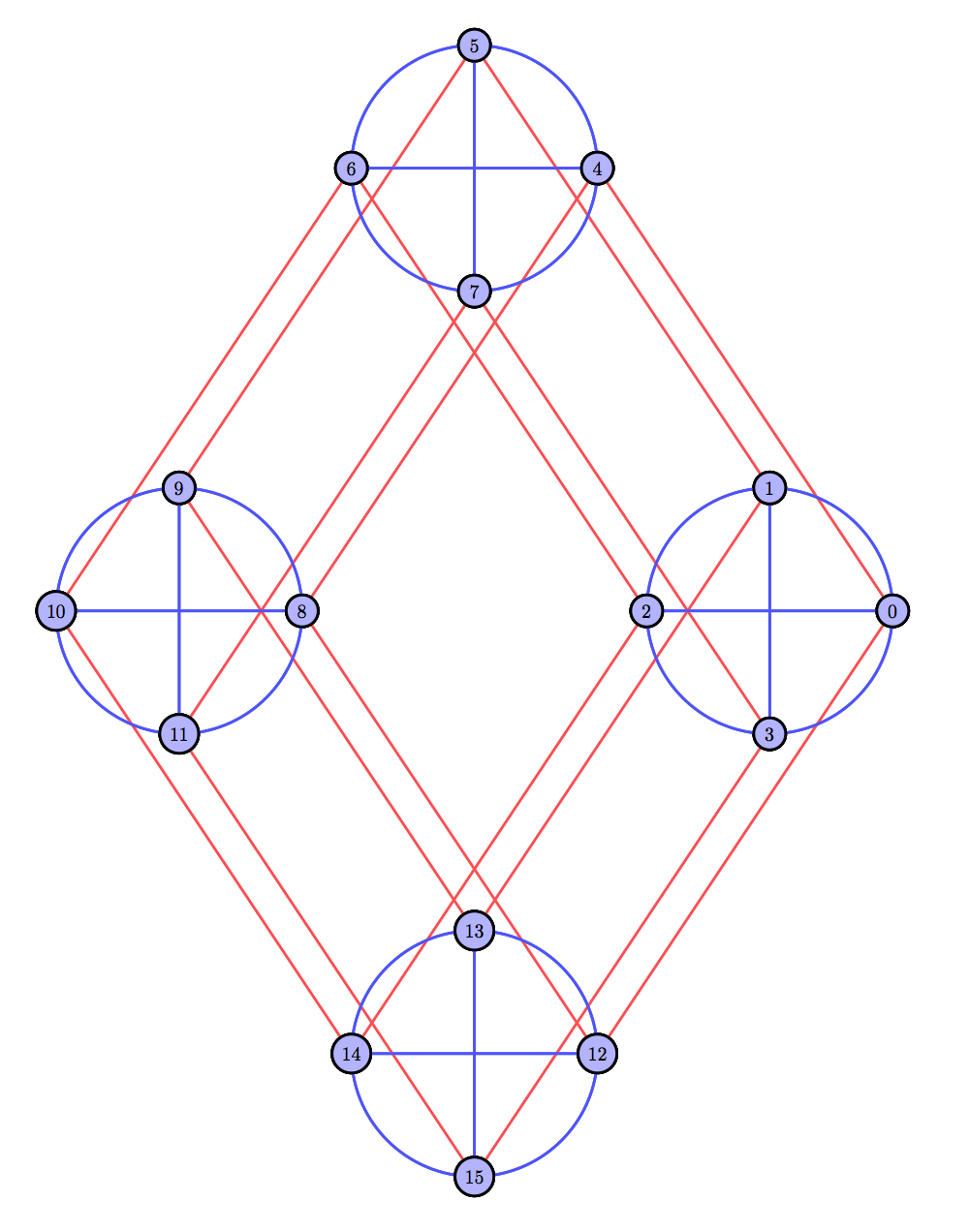}}
 \hspace*{.03in}
 {\bf =}
 \hspace*{.03in}
  \raisebox{-0.5\height}{\includegraphics[height=0.7in]{circ1.png}}
  \hspace*{.02in}
  ${\bf \times}$
  \hspace*{.02in}
 \raisebox{-0.4\height}{\includegraphics[height=0.7in]{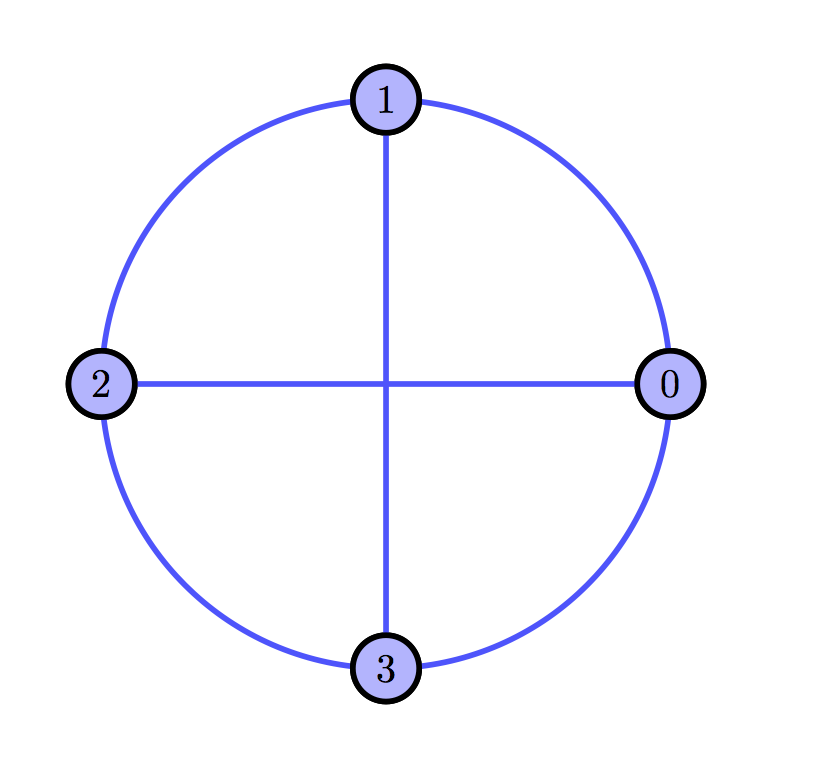}}
 \end{minipage}
\caption{Graph Downsampling and Coarsening of $G_{\diamond}$ in Fig $6$ on $G_2$ w.r.t. $s=1\in S_2$ with coarsened $\tilde{G}_2$.}
\label{fig:res}
\end{figure}
\\
Furthermore, we propose the \textit{separable} graph wavelet transform on $G_{\diamond}$ as an alternative construction, which is applied with respect to the individual graph factors. We denote by ${\bf W}_i$ the graph (e-)spline wavelet transform constructed in the vertex domain of circulant graph factor $G_i$ \[{\bf W}_i=\begin{bmatrix} {\bf \Psi}_{\downarrow 2}{\bf H}_{LP_{\vec{\alpha}}}\\{\bf \Phi}_{\downarrow 2}{\bf H}_{HP_{\vec{\alpha}}}\end{bmatrix},\] with downsampling matrices ${\bf \Psi}_{\downarrow 2}, {\bf \Phi}_{\downarrow 2}$ and ${\bf H}_{LP_{\vec{\alpha}}/HP_{\vec{\alpha}}}$ as defined in Sect. $3$. Then $({\bf W}_{1}\otimes {\bf W}_{2})$ represents the separable transform to process ${\bf x}$ with respect to $G_1$ and $G_2$, which entails the analysis of $N_1$ graph signal partitions $\{x((0:N_2-1)+(t-1)*N_2)\}_{t=1}^{N_1}$ on $G_2$, and subsequent $N_2$ partitions $\{w(t:N_2:N_1N_2-1)\}_{t=0}^{N_2-1}$ on $G_1$, with ${\bf w}=({\bf I}_{N_1}\otimes {\bf W}_2){\bf x}$. For partition $\bar{{\bf x}}_{i}\in\mathbb{R}^{N_i}$ on $G_i$, let $\bar{{\bf w}}_i={\bf P}_{N_i} {\bf W}_i \bar{{\bf x}}_i$ be the graph wavelet domain representation on the same graph, subject to the (node relabelling) permutation ${\bf P}_{N_i}$; here, the respective low-and high-pass values of $\bar{{\bf w}}_i$ can be subsequently assigned to suitable coarsened versions of $G_i$. Further, the recombination $\bar{{\bf w}}=\bar{{\bf w}}_1\otimes \bar{{\bf w}}_2$ gives rise to the graph signal $\bar{{\bf w}}$ on $G$.  
For a multiscale representation, the proposed scheme can be generalized to accommodate iterations on the low-pass branch, by defining the multilevel transform \[{\bf W}_i^{(j)}=\begin{bmatrix} {\bf W}_i^{j}&\\& {\bf I}_{N_i-\frac{N_i}{2^{j}}}\end{bmatrix}\dots {\bf W}_i^0\] and iterative permutation matrix \[{\bf P}_{N_i}^{(j)}={\bf P}_i^0\dots \begin{bmatrix} {\bf P}_i^{j}&\\& {\bf I}_{N_i-\frac{N_i}{2^{j}}}\end{bmatrix}\] at levels $j\leq J-1$,
giving \[{\bf w}=({\bf P}^{(J-1)}_{N_1}\otimes {\bf P}^{(J-1)}_{N_2})({\bf W}^{(J-1)}_{1}\otimes {\bf W}^{(J-1)}_{2}){\bf x}={\bf P}^{(J-1)}_{N_1 N_2}({\bf W}^{(J-1)}_{1}\otimes {\bf W}^{(J-1)}_{2}){\bf x},\] where $(\tilde{{\bf W}}_{1}\otimes \tilde{{\bf W}}_{2})=({\bf W}^{(j)}_{1}\otimes {\bf W}^{(j)}_{2})$ represents the introduced graph product transform at level $j$. The transform is invertible with inverse $(\tilde{{\bf W}}_1^{-1}\otimes \tilde{{\bf W}}^{-1}_2)$, depending on the invertibility of its circulant sub-wavelet transforms ${\bf W}^j_i$. \\
Traditionally, the application of a 2D discrete wavelet transform on an (image) matrix ${\bf X}$ can be expressed as ${\bf w}_{2D}= {\bf W}_1 {\bf X}{\bf W}_2^T$, whose row-vectorized form is given by $({\bf W}_1\otimes{\bf W}_2){\bf x}$. It becomes evident that this constitutes an analogy to one level of the proposed transform, where in the traditional domain we have ${\bf W}_2={\bf W}_1$, while in the graph domain the ${\bf W}_i$'s generally differ, as they are not defined on the same graph. This elucidates that our derived scheme can be regarded as the equivalent of operating on a graph signal (or vectorized image) with respect to confined direction (rows and columns), which are in this case dictated by the factors in the chosen graph decomposition. In contrast to a graph wavelet analysis of ${\bf x}$ on $G_{\diamond}$ via a suitable transform, we can therefore regard the analysis of partitions $\bar{{\bf x}}_i$ on $G_i$ and subsequent reassignment to the vertices of $G_{\diamond}$ as a two-dimensional extension of the former. 

\subsection{Smoothness and Sparsity on Product Graphs}
For the remainder of this discussion, following Def. \ref{def41} we are primarily interested in the analysis of multi-dimensional graph signals which admit the rank-1 decomposition ${\bf x}={\bf x}_1\otimes {\bf x}_2$ into smooth signal tensors ${\bf x}_i$ for a maximally sparse representation.\\
We begin by investigating how the smoothness of graph signal ${\bf x}$ with respect to $G_{\diamond}$ is related to the smoothness of the subgraph signals ${\bf x}_i$ with respect to $G_i$; here, we resort to the graph Laplacian quadratic form $S_2({\bf x})={\bf x}^T{\bf L}{\bf x}$ as a measure \cite{shu}. We denote the individual smoothness measures with $S_{\diamond}={\bf x}^T {\bf L}_{\diamond} {\bf x}$ for $G_{\diamond}$, and $S_i={\bf x}_i^T {\bf L}_i {\bf x}_i$ for factors $G_i$, and assuming degree regularity, simplify ${\bf D}_i=d_i{\bf I}_{N_i}$, which gives rise to the following relations:
\begin{itemize}
\item $S_{\otimes}=d_2 S_1 ||{\bf x}_2||_2^2+d_1 S_2 ||{\bf x}_1||_2^2-S_1 S_2$
\item $S_{\times}=S_1 ||{\bf x}_2||_2^2+S_2 ||{\bf x}_1||_2^2$
\item $S_{\boxtimes}=(1+d_2) S_1 ||{\bf x}_2||_2^2+(1+d_1) S_2 ||{\bf x}_1||_2^2-S_1 S_2$
\item $S_{\lbrack \enskip\rbrack}=||{\bf x}_1||_2^2 S_2+S_1 c_2^2+d_1||{\bf x}_1||_2^2({N_2} ||{\bf x}_2||_2^2-c_2^2)$, \end{itemize} with constant $c_2=\sum_{i=0}^{N_2-1} x_{2}(i)$. 
The total smoothness $S_{\diamond}$ is composed of the weighted sub-measures $S_i$, whose individual contribution is scaled by parameters pertaining to the corresponding subgraph signal energy and node degree of the opposing factor graph. One easily deduces that if for a chosen decomposition $G_{\diamond}=G_1\diamond G_2$ with factors $G_i$, the measures $S_i$ are small, i.e. the sub-signal tensors ${\bf x}_i$ are smooth with respect to factors $G_i$, then ${\bf x}$ is also relatively smooth on $G_{\diamond}$ with small $S_{\diamond}$, subject to a scaling. 
When ${\bf x}$ is a properly normalized eigenvector of ${\bf L}_{\diamond}$, the derived smoothness measures $S_{\diamond}$ constitute the corresponding eigenvalues, expressed in terms of the eigenvalues $S_i$ of factors ${\bf L}_i$. The above relations continue to hold for the symmetric normalized graph Laplacian matrices of non-regular graphs.\\
Furthermore, we consider the signal ${\bf L}_{\diamond}{\bf x}$ and analyze its sparsity $||{\bf L}_{\diamond}{\bf x}||_0$, following the interpretation of ${\bf L}_{\diamond}$ as a high-pass filter within the non-separable GWT of Sect. $4.2.1$:
\begin{itemize}
\item ${\bf L}_{\otimes}{\bf x}=({\bf L}_1{\bf x}_1)\otimes d_2 {\bf x}_2+d_1 {\bf x}_1\otimes ({\bf L}_2{\bf x}_2)-({\bf L}_1{\bf x}_1)\otimes ({\bf L}_2{\bf x}_2)$
\item ${\bf L}_{\times}{\bf x}=({\bf L}_1{\bf x}_1)\otimes {\bf x}_2+{\bf x}_1 \otimes({\bf L}_2{\bf x}_2)$
\item ${\bf L}_{\boxtimes}{\bf x}=({\bf L}_1{\bf x}_1)\otimes d_2 {\bf x}_2+d_1 {\bf x}_1\otimes ({\bf L}_2{\bf x}_2)-({\bf L}_1{\bf x}_1)\otimes ({\bf L}_2{\bf x}_2)+({\bf L}_1{\bf x}_1)\otimes {\bf x}_2+{\bf x}_1 \otimes({\bf L}_2{\bf x}_2)$
 \item ${\bf L}_{\lbrack \enskip \rbrack}{\bf x}={\bf x}_1\otimes ({\bf L}_2{\bf x}_2)+({\bf L}_1{\bf x}_1)\otimes c_2 {\bf 1}_{N_2}+d_1 {\bf x}_1\otimes ({N_2}{\bf x}_2-c_2 {\bf 1}_{N_2})$
\end{itemize}
It becomes evident that for constant ${\bf x}_i$ (and hence ${\bf x}$) such that ${\bf L}_i {\bf x}_i={\bf 0}_{N_i}$, we have ${\bf L}_{\diamond}{\bf x}=0$, preserving the nullspace. When ${\bf x}_i$ are linear polynomials and $G_i$ banded circulant graphs, ${\bf L}_i{\bf x}_i$ are sparse, which is not necessarily true for ${\bf L}_{\diamond}{\bf x}$ under any product operation.\\
Replacing ${\bf D}_i$ by diagonal e-degree matrix $\tilde{{\bf D}}_{i,\alpha_k}$ for exponential parameter $\alpha_k$, gives rise to equivalent relations for $\tilde{{\bf L}}_{i,\alpha_k}$ on the graph factors.
Hence, for periodic complex exponential graph signals ${\bf x}_i$ parameterised by $\alpha_k=\frac{2\pi k}{N_i},\enskip k\in\lbrack 0\enskip N_i-1\rbrack$ with resulting multi-dimensional complex exponential ${\bf x}$, we obtain ${\bf L}_{\alpha_{k1}\diamond\alpha_{k2}}{\bf x}=0$, with exception of the lexicographic product, for which this holds only if ${\bf x}_2$ is an all-constant vector, as previously evidenced by its eigenspace property.\\
Overall, this demonstrates that the vanishing moment properties of circulant (e-)graph Laplacians are to some extent preserved within the graph product, yet sparsity is reduced as a result of the newly arising interconnections between the factors. Due to the fact that the above relations cannot be generalized to powers of the graph Laplacian matrix, comparable property preservations are not extended to higher order (exponential) polynomial graph signals, which have a sparse representation with respect to ${\bf L}_i^k$. This suggests that by performing a separable signal (wavelet) analysis with respect to inherent circulant substructures, we generally gain a sparser representation as measured via the graph Laplacian and its powers; nevertheless, for signals that lie in the eigenspace of the given graph product, the annihilation property is also preserved at higher operator powers. \\
\\
In light of this, we compare the sparsity of representation attained via the proposed graph wavelet transforms, and discover that the separable approach, apart from preserving higher-order annihilation properties, can induce more sparsity. Here, it should be noted that we consider two levels of the 1-D graph wavelet transform (or alternatively, downsampling on both factors, with respect to ${\bf K}={\bf I}_{{N_1}}\otimes {\bf K}_2$ followed by ${\bf K}={\bf K}_{1}\otimes {\bf I}_{N_2/2}$), as comparable to one level of the 2-D transform, yet similarly as in the traditional domain, there is no direct equivalence between the two. \\
\\
{\bf Example 1:} Given graph signal ${\bf x}={\bf x}_1\otimes {\bf x}_2$ on $G=G_1\otimes G_2$, where $G_i$ are circulant and banded of bandwidth $M_i$ and ${\bf x}_i\in\mathbb{R}^{N_i}$ are linear polynomial, let ${\bf W}_i\in\mathbb{R}^{N_i\times N_i}$ and ${\bf W}_{\otimes}\in\mathbb{R}^{N_1 N_2\times N_1 N_2}$ represent first-order graph-spline wavelet transforms on factors $G_i$  and $G$ respectively. Here, we downsample w.r.t. $s=1\in S_i$ on each $G_i$ and reconnect nodes such that generating sets $S_i$ are preserved. Hence, separable representation ${\bf w}={\bf w}_1\otimes {\bf w}_2$ has $K=\frac{3}{4}N_1 N_2 -\frac{1}{2}(2M_1 M_2+M_1 N_2+M_2 N_1)$ zero entries, whereas non-separable ${\bf w}_{\otimes}={\bf W}_{\otimes}{\bf x}$ has a total of $K_{\otimes}=K_{\otimes,1}+K_{\otimes,2}$ zeros, with $K_{\otimes, 1}=\frac{1}{2} N_1 N_2 -(M_1 N_2+M_2 N_1-2 M_2 M_1)$ and $K_{\otimes,2}=\frac{1}{4} N_1 N_2-(\frac{3}{2}N_1 M_2+M_1 N_2-6 M_1 M_2)$ zeros at levels $1$ and $2$ respectively. We have $K>K_{\otimes}$ for $4M_i<N_i$, and $K>K_{\otimes}=K_{\otimes, 1}$ with $K_{\otimes, 2}=0$ for $N_i/4\leq M_i< N_i/2$, which implies that the separable approach induces a sparser representation at any bandwidth $M_i<N_i/2$. Note that the sparsity of ${\bf w}_{\diamond}$ is the same under any of the first three graph products.
\\
\\
{\bf Example 2:} For circulant lexicographic product graph $G=C_{N_1N_2,S}$ and decomposition $G=G_1\lbrack G_2\rbrack$, it can be deduced from the above relations, that we can gain sparsity by conducting the 2-D graph wavelet analysis of ${\bf x}_i$ on $G_i$ as opposed to the 1-D analysis of ${\bf x}$ on $G$, for any choice of compressible ${\bf x}_i$ as long as they do not lie in the eigenspace of the graph Laplacian. Here, the product-related relabelling ${\bf P}{\bf x}={\bf x}_2\otimes {\bf x}_1$, which renders a circulant matrix $\tilde{{\bf A}}_{\lbrack \enskip\rbrack}={\bf P}{\bf A}_{\lbrack\enskip \rbrack} {\bf P}^T$, corresponds to a simple stacking of columns instead of rows, thus preserving the tensor product, and associated smoothness properties with respect to the subgraphs.

\section{Illustrative Examples}
In an effort to further exemplify the proposed graph wavelet transforms from Sects. $3$ and $4$, we study their non-linear approximation (NLA) potential for (piecewise) smooth graph signals in two concrete cases of respectively an artificial and a data-driven graph setting.\\ For a given graph signal ${\bf x}\in\mathbb{R}^N$ with graph wavelet representation ${\bf w}\in\mathbb{R}^N$, the
corresponding non-linear approximation within the inverse graph wavelet basis with columns $\tilde{{\bf w}}_k$ is defined as $\tilde{{\bf x}}=\sum_{k\in I_K} w(k) \tilde{{\bf w}}_k$, where $I_K$ denotes the index-set of the $K$-largest magnitude coefficients of ${\bf w}$. \\
\\
Fig. \ref{fig:esplinedouble} illustrates the non-linear approximation performance for a sinusoidal graph signal ${\bf x}$ residing on the vertices of the circulant graph with generating set $S=\{1,2\}$ of cardinality $N=|V|=1024$, by comparing the \textit{HGESWT} (of Thm. \ref{thm32}) and complementary construction \textit{HCGESWT} (of Thm. \ref{thm34}), where the reconstruction error is measured as $SNR=10\log_{10}\frac{||{\bf x}||_2^2}{||\tilde{{\bf x}}-{\bf x}||_2^2}$. In particular, the transforms feature the same analysis high-pass filter, suitably parameterized by $\{\alpha_i\}_{i=1}^2$ to annihilate ${\bf x}$, and associated variable analysis low-pass filter, where the \textit{HCGESWT} is presented in two variations with either dual ($4.4$) or unilateral ($4.0$, on the analysis side) exponential vanishing moments. Further, the graph wavelet atoms (rows) are normalized to unit length and $5$ levels of decomposition considered, where graph coarsening reconnection is conducted by retaining the same generating set $S$. It becomes evident that perfect reconstruction can be attained at a relatively low number of retained wavelet coefficients as a result of the transform annihilation properties. The resulting basis functions of the analysis low-pass filters are further depicted in Fig. \ref{fig:esplinedouble}.
\\
\begin{figure}[htb]
\centering

	{\includegraphics[width=2.9in]{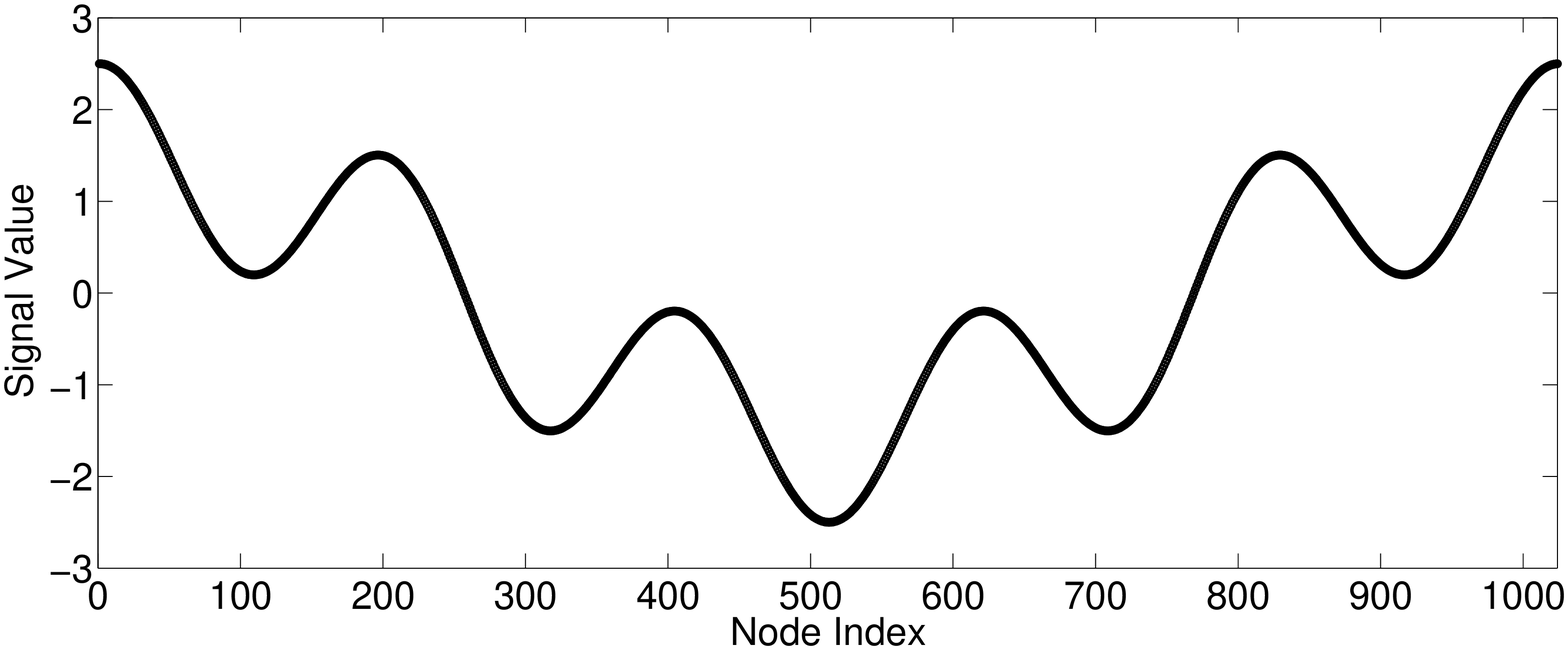}}
	{\includegraphics[width=2.9in]{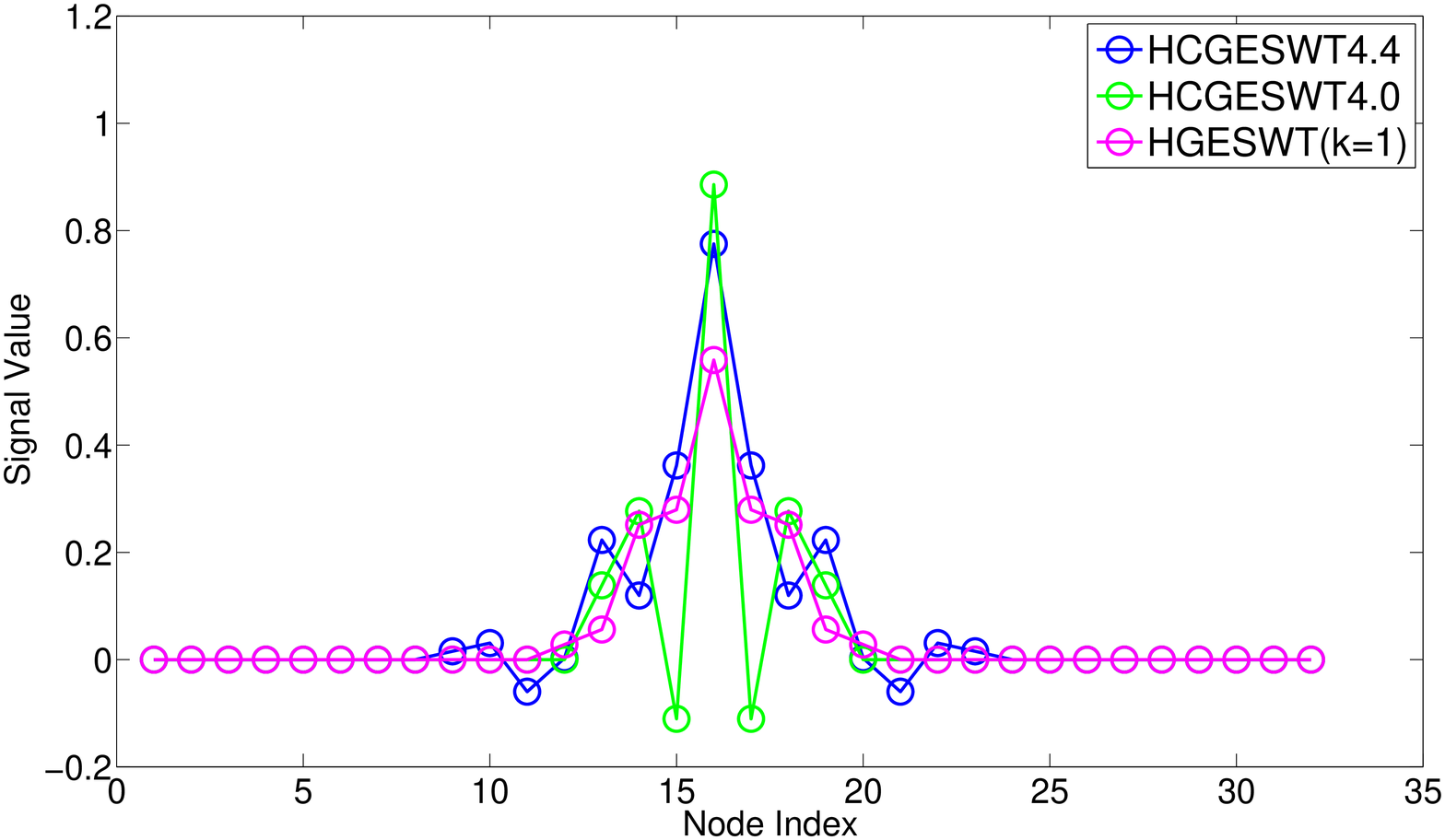}}%
	{\includegraphics[width=2.9in]{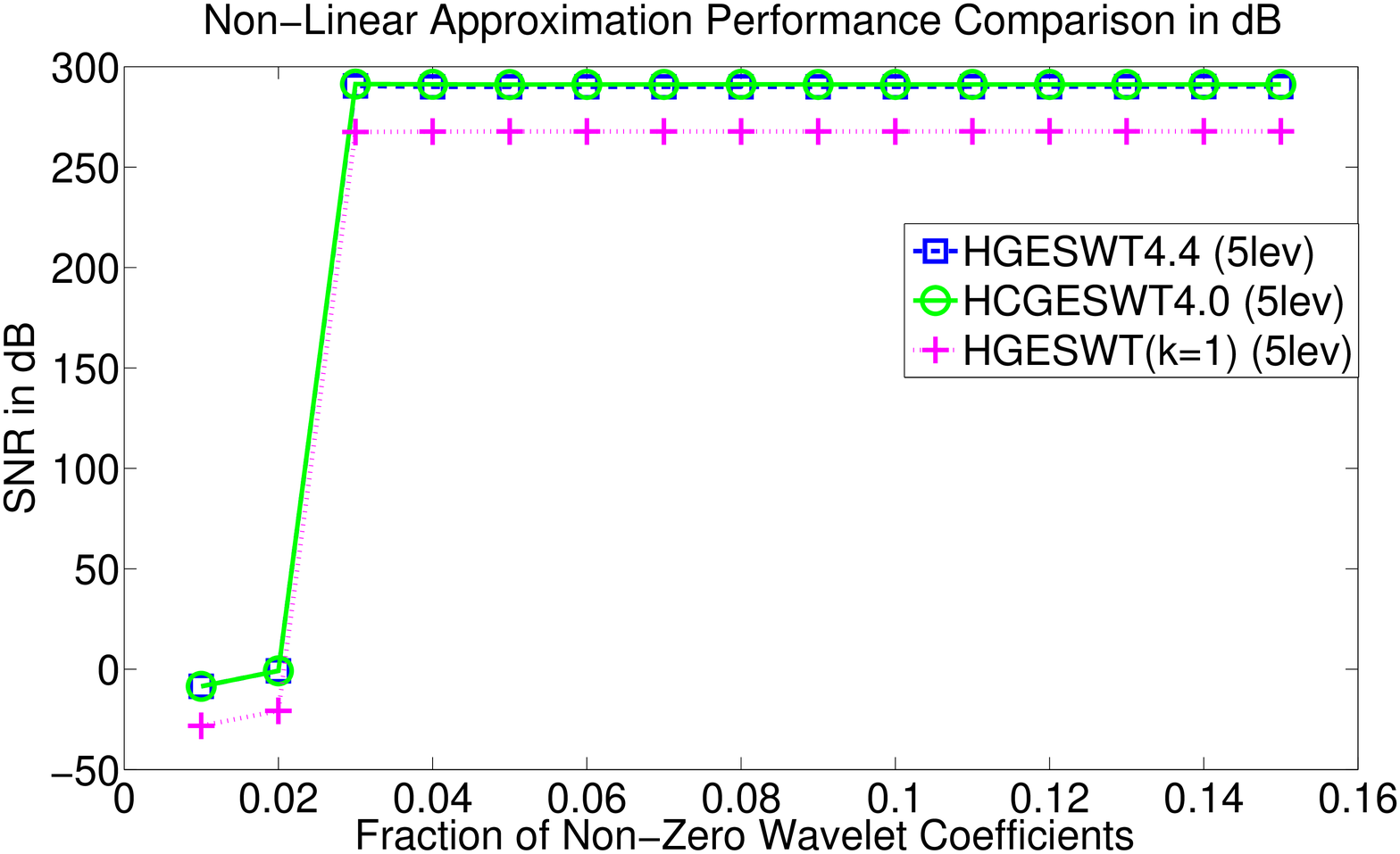}}\\
	
		\caption{Comparison of analysis graph low-pass filters (left) and NLA performance (right) for a sum of two sinusoidals with $(\alpha_1=\frac{2\pi 1}{N}$, $\alpha_2=\frac{2\pi 5}{N})$ (top). \label{fig:esplinedouble}}
\end{figure}\\
\\
Moreover, a data-driven graph example is presented to motivate the application of graph wavelets for sparse image approximation. Given an $N\times N$ image ${\bf I}$, let each pixel be represented by a node, such that connectivity between node pairs $(i,j)$ is established as the distance between their spatial location $p$ and intensity $I$ \begin{equation}\label{eq:bila}w_{i,j}=e^{-\frac{||p_i-p_j||_2^2}{\sigma_p^2}}e^{-\frac{|I(i)-I(j)|^2}{\sigma_I^2}},\enskip i,j\in\{0,...,N^2-1\}\end{equation}
to form the graph $G=(V,E)$ with adjacency matrix ${\bf W}$. In addition, let ${\bf x}\in\mathbb{R}^{N^2}$ denote the vector-stacked image graph signal with intensity value $x(i)=I(i)$ at node $i$. \\
Following a scheme for the analysis of images featuring distinct discontinuities or patterns, derived in our prior work (\cite{globalsip}, \cite{spie}), we proceed to firstly conduct a normalized graph cut on $G$, which segments the image graph into two (or more) regions of homogeneous intensity content by removing edges of minimum total weight \cite{cut}, and to secondly, construct suitable graph wavelets on the nearest circulant graph approximations, as discussed in Sect. $4$, of the resulting partitions $\{G_i\}_i$. Here, we opt for the first-order graph spline wavelet transform of Thm. \ref{thm31} ($k=1$), with downsampling conducted w.r.t $s=1$, and vary its data-localization, i.e. the underlying graph type, for a refined performance.\\
In particular, the GWTs are constructed on circulant approximations of the following graph variations $(i)$ \textit{{\bf nGWT}}: the original complete weighted graph, $(ii)$ \textit{{\bf sparseGWT(bil, RCM)}}: the graph is sparsified by Euclidean distance by discarding connections outside of the pixel grid $||p_i-p_j||_2\leq \sqrt{2}$ and subsequently relabelled by the RCM-algorithm (see also Sect. $4$), $(iii)$ \textit{{\bf sparseGWT(I, sort)}}: the graph is re-weighted as intensity-only, sparsified by a data-dependent intensity threshold $I_T$, and relabelled using signal sorting, and $(iv)$ \textit{{\bf GWT(S=(1), sort)}}: the graph is reduced to the smoothest simple cycle, i.e. the sorted sub-graph signal is projected onto the simple cycle graph. \\
\\
As previously discussed, the purpose of the RCM-algorithm is to provide a relabelling which minimizes the graph bandwidth prior to the approximation operation, and thus maintains compact support.
The underlying motivation for options $(iii)$ and $(iv)$ is that when the graph at hand is based solely on the intensity values of (image) signal partitions $\{{\bf x}_i\}_i$, one can effectively use a simple sorting operation on the latter as the relabelling which simultaneously minimizes the bandwidth of corresponding $G_i$ and total variation of ${\bf x}_i$. In contrast to the spatially localized GWT $(ii)$, the corresponding reordered signal is thus effectively smooth, as visualized in the example of Fig. $5$ in Sect. $4$, which results in a maximally sparse graph wavelet representation. Here, the signal sorting-based relabelling can be considered as an optimized form of the RCM-algorithm, which searches for a level structure (or multidimensional path) to traverse the graph such that the maximum value of the distances $|i-k|$ over all edges $(i,k)$ is minimized \cite{rcm}.\\
For a multiscale representation, downsampling is conducted with respect to the outmost cycle ($s=1$) and no reconnection applied for graph coarsening; if either subgraph is of odd dimension, nearest circulant approximation is alternatively employed to preserve circularity.\\
\\
Consider the example of a real image patch extracted from the $256\times 256$ `cameraman' in Fig. \ref{fig:cameraman1cut} \cite{spie}. We compare the proposed graph wavelet transforms, with normalized rows, to classical 2D wavelet transforms in form of the 2D Haar and linear spline (CDF 5/3) transforms at $5$ levels of decomposition; here, performance is measured as $PSNR=20\log_{10}\left(\frac{N}{||{\bf I}-\tilde{{\bf I}}||_F}\right)$ in dB, with a post-processing of outliers such that $\tilde{I}(i,j)\in\lbrack 0\enskip 1\rbrack$ for reconstructed image $\tilde{{\bf I}}$. 
\begin{figure}
\begin{subfigure}[h]{0.32\textwidth}
{\includegraphics[width=1.6in]{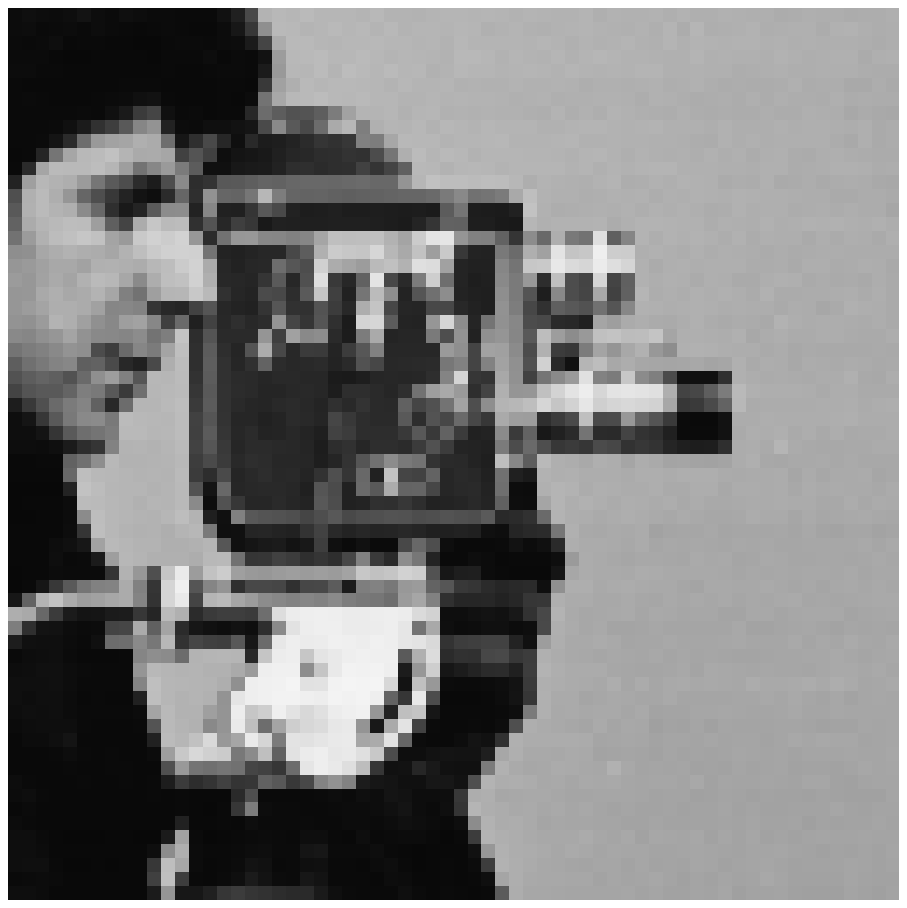}} 
\caption{}
\end{subfigure}
\begin{subfigure}[h]{0.32\textwidth}
 { \includegraphics[width=1.6in]{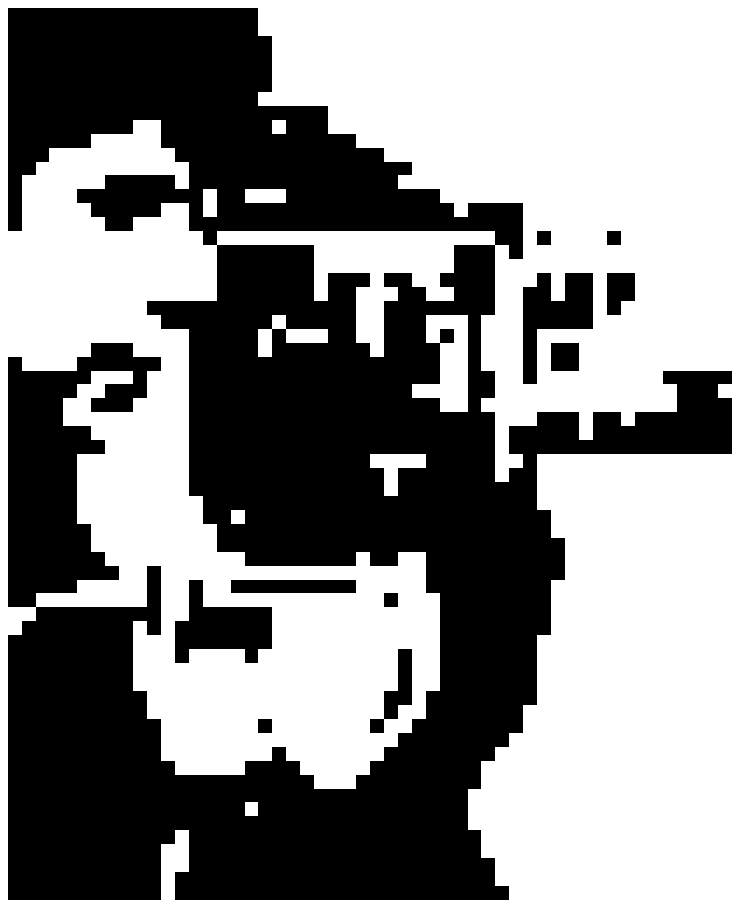}}
 \caption{}
 \end{subfigure}
 \begin{subfigure}[h]{0.32\textwidth}
 {\includegraphics[width=1.6in]{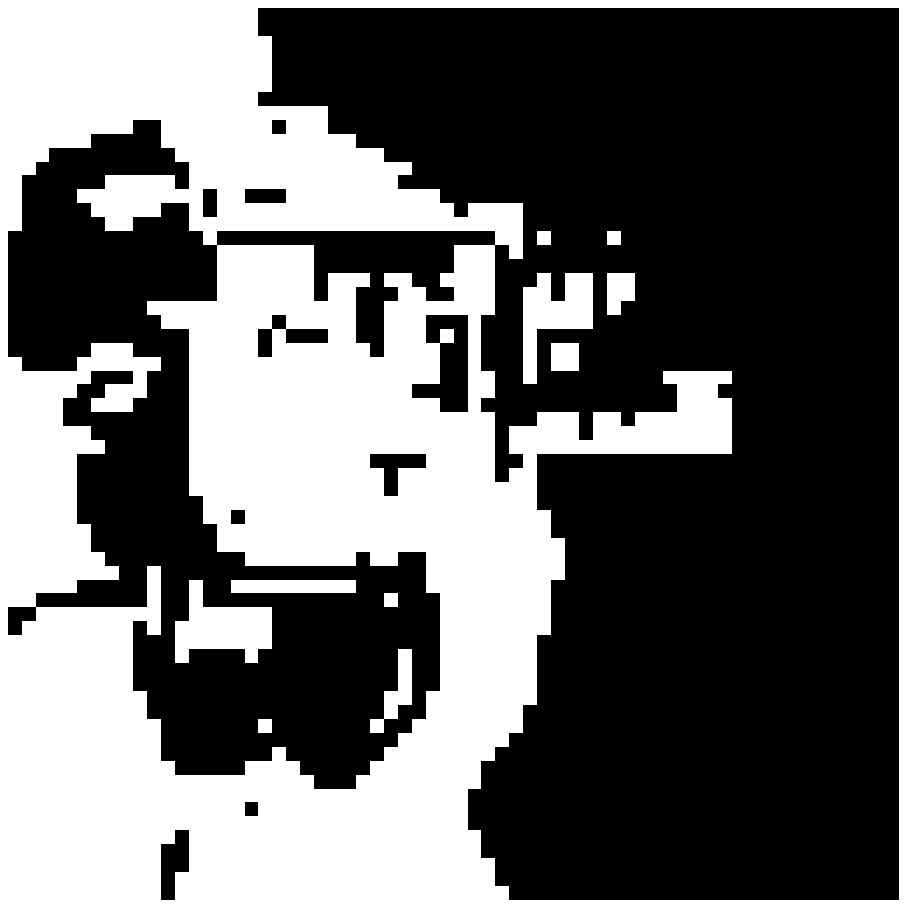}}
\caption{}
 \end{subfigure}
 \\
 \centering
 \begin{subfigure}[h]{0.5\textwidth}
  {\includegraphics[width=3in]{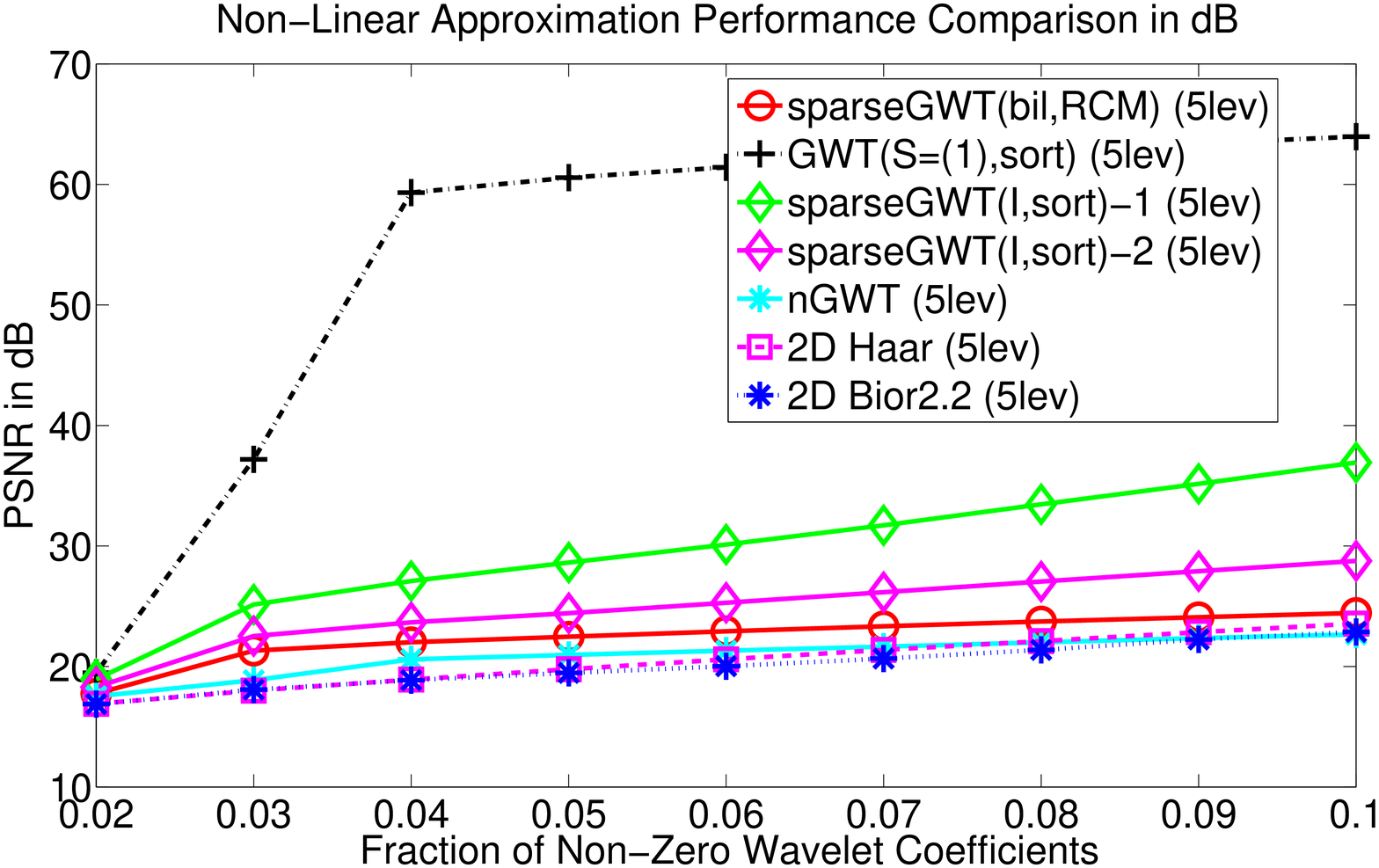}}
  \caption{}
  \end{subfigure}
  \caption{{$(a)$ Original $64\times 64$ image patch, $(b)$-$(c)$ Graph Cut Regions $\&$ $(d)$ Comparison of NLA performance.}\label{fig:cameraman1cut}}
\end{figure}
It becomes evident that the intensity-only based graph wavelet variations outperform the remaining, while all graph-based methods outperform the traditional tensor product wavelets. The effect of the former is further demonstrated through the use of two different intensity thresholds; in particular, the sparser graph with a smaller bandwidth leads to better results (see \textit{sparseGWT(I,sort)-1} in Fig. \ref{fig:cameraman1cut}). Best performance is achieved by a margin when the segmented subgraphs are reduced to their (sparsest) smoothest cycle of bandwidth $1$.

\section{Conclusion}
In this paper, we have introduced novel families of wavelets and associated filterbanks on circulant graphs with vanishing moment properties, which reveal (e-)spline-like functions on graphs, and promote sparse multiscale representations. Moreover, we have discussed generalizations to arbitrary graphs in the form of a multidimensional wavelet analysis scheme based on graph product decomposition, facilitating a sparsity-promoting generalization with the advantage of lower-dimensional processing. In our future work, we wish to further explore the sets of graph signals which can be annihilated with existing and/or evolved graph wavelets as well as refine its extensions and relevance for arbitrary graphs.

 \appendix
 \section{}
\subsection{}
\begin{proof}[Proof of Theorem 3.1]
It is self-evident that since the high-pass filter of Eq. $(4)$ is a power $k$ of the graph Laplacian matrix, whose associated polynomial representation has $2$ vanishing moments, the annihilation property is generalized to higher order of $2k$ vanishing moments; thus we proceed to demonstrate invertibility of the filterbank. The core of the proof follows a similar line of argumentation as the one provided in \cite{Ekambaram3} for $k=1$ with generalizations pertaining to the parameter $k$. For completeness we present the entire proof here. \\
Applying the binomial theorem, we observe that \[\frac{1}{2^k}\left({\bf I}_N\pm\frac{{\bf A}}{d}\right)^k=\frac{1}{2^k}\sum_{j=0}^k(\pm 1)^j{k\choose j}\left(\frac{{\bf A}}{d}\right)^j\] and so we need to show that the nullspace of \[\frac{1}{2^k}\left(\sum_{j\in\mathbb{Z}}{k\choose 2j}\left(\frac{{\bf A}}{d}\right)^{2j}+{\bf K}\sum_{j\in\mathbb{Z}}{k\choose 2j+1}\left(\frac{{\bf A}}{d}\right)^{2j+1}\right)\] is empty, where ${\bf K}$ is the diagonal matrix with entries $K(i,i)=1$ if node $i$ retains the low- and $K(i,i)=-1$ if it retains the high-pass component.
Assume the contrary and define the vector ${\bf z}={\bf V}{\bf r}$ to lie in the nullspace, given eigendecomposition $\frac{{\bf A}}{d}={\bf V}{\bf \Gamma}{\bf V}^H$, which yields the following simplifications:
\begin{equation}\frac{1}{2^k}\left(\sum_{j\in\mathbb{Z}}{k\choose 2j}\left(\frac{{\bf A}}{d}\right)^{2j}+{\bf K}\sum_{j\in\mathbb{Z}}{k\choose 2j+1}\left(\frac{{\bf A}}{d}\right)^{2j+1}\right){\bf V}{\bf r}={\bf 0}_N\end{equation}

\begin{equation}\Leftrightarrow||{\bf V}\sum_{j\in\mathbb{Z}}{k\choose 2j}{\bf \Gamma}^{2j}{\bf r}||_2^2=||{\bf K}{\bf V}\sum_{j\in\mathbb{Z}}{k\choose 2j+1}{\bf \Gamma}^{2j+1}{\bf r}||_2^2\end{equation}
where Eq. $(A.2)$ is the result of a rearrangement of terms in Eq. $(A.1)$ and subsequent application of the $l_2$-vector norm on both sides of the equality. After further simplification, we obtain
\begin{equation}\sum_{i=0}^{N-1} r(i)^2\left(\left(\sum_{j\in\mathbb{Z}}{k\choose 2j}\gamma_i^{2j}\right)^2-\left(\sum_{j\in\mathbb{Z}}{k\choose 2j+1}\gamma_i^{2j+1}\right)^2\right)\stackrel{(a)}{=}\sum_{i=0}^{N-1} r(i)^2(A_i^2-B_i^2)=0,\end{equation}
where in $(a)$, we let $A_i$ and $B_i$ represent the sum of even and odd terms in the binomial series respectively.
For the nullspace to be empty, we need to show that ${\bf r}={\bf 0}_N$, which follows if $(A_i^2-B_i^2)\neq 0$ and is strictly positive or negative $\forall i$. By utilizing the fact that for a general binomial series $(x+a)^n$, with terms $A_i$ and $B_i$, the following holds: $(x^2-a^2)^n=A_i^2-B_i^2$,
we obtain
\[\sum_{i=0}^{N-1} r(i)^2(A_i^2-B_i^2)=\sum_{i=0}^{N-1} r(i)^2(1-\gamma_i^2)^k=0.\]
The eigenvalues are given by $|\gamma_i|\leq 1$, via the Gershgorin circle theorem, where $\gamma_i>-1$ unless the graph is bipartite \cite{chung}; thus we have that $|r(i)|>0$ only if $|\gamma_{i}|=1$ and $r(i)=0$ otherwise. To examine these special cases, let the corresponding eigenvectors for $\gamma_1=-1$ and $\gamma_2=1$ be given by ${\bf \tilde{V }}$ and $\frac{r(0)}{\sqrt{N}}{\bf 1}_N$ respectively, such that ${\bf z}=\frac{r(0)}{\sqrt{N}}{\bf 1}_N+{\bf \tilde{V}} {\bf \tilde{r}}$, and substitute into Eq. $(A.1)$. Here, the multipicity of $\gamma_2=1$ is one, since the graph is connected \cite{chung}. We consider the case of a non-bipartite graph first:\\
\[\frac{r(0)}{\sqrt{N}}\left(\sum_{j\in\mathbb{Z}} {k\choose 2j}{\bf 1}_N+{\bf K}\sum_{j\in\mathbb{Z}} {k\choose 2j+1}{\bf 1}_N\right)={\bf 0}_N\]
Noting that $\sum_{j\in\mathbb{Z}}{k\choose 2j}=\sum_{j\in\mathbb{Z}}{k\choose 2j+1}$, we need at least one entry $K(i,i)=1$, such that $r(0)=0$. 
\\
In the bipartite case, due to spectral folding, if $\gamma$ is an eigenvalue of ${\bf A}$ with eigenvector $\begin{bmatrix} {\bf v}_{B}\\{\bf v}_{B^{\complement}}\end{bmatrix}$, so is $-\gamma$ with eigenvector $\begin{bmatrix} {\bf v}_{B}\\-{\bf v}_{B^{\complement}}\end{bmatrix}$, where $B$ is the set of the node indices in one bipartite set \cite{chung}. Then $\gamma_1=1$ and $\gamma_2=-1$ each have multiplicity one with respective eigenvectors ${\bf 1}_N$ and $\begin{bmatrix} {\bf 1}_{B}\\-{\bf 1}_{B^{\complement}}\end{bmatrix}$, where $|B|=|B^{\complement}|=N/2$, giving \begin{equation}
\begin{aligned}
&\frac{r(0)}{\sqrt{N}}\underbrace{\left(\sum_{j\in\mathbb{Z}}{k\choose 2j}{\bf I}_N+{\bf K}\sum_{j\in\mathbb{Z}}{k\choose 2j+1}{\bf I}_N\right)}_{T_1}{\bf 1}_N\\
&+r(1)\underbrace{\left( \sum_{j\in\mathbb{Z}}{k\choose 2j}{\bf I}_N -{\bf K}\sum_{j\in\mathbb{Z}}{k\choose 2j+1}{\bf I}_N\right)}_{T_2}\begin{bmatrix} {\bf 1}_{B}\\-{\bf 1}_{B^{\complement}}\end{bmatrix}={\bf 0}_N.\end{aligned}\end{equation}
Here we have used the property $\frac{{\bf A}}{d}{\bf v}=\gamma {\bf v}$, in $\left(\frac{{\bf A}}{d}\right)^j \begin{bmatrix} {\bf 1}_{B}\\-{\bf 1}_{B^{\complement}}\end{bmatrix}=(-1)^j \begin{bmatrix} {\bf 1}_{B}\\-{\bf 1}_{B^{\complement}}\end{bmatrix}$, leading to an alternating pattern on the RHS when $j$ is odd. \\
In particular, for any choice of downsampling pattern ${\bf K}$, the terms $T_1$ and $T_2$ in the first and second summands in Eq. $(A.4)$, will respectively have zero entries along the main diagonal, which lie in complementary index sets. Therefore, as long as at least one node retains the low-pass component $K(i,i)=1$, it follows that $r(0) = 0$ and $r(1) = 0$, which again implies ${\bf z} = {\bf 0}$, completing the proof. 
\end{proof}
\subsection{}
\begin{proof} [Proof of Theorem 3.2]
We can rewrite the simplified filters  
\begin{equation}{\bf H}_{LP_{\vec{\alpha}}}=\prod_{n=1}^{Tk} \left(\beta_n{\bf I}_N+\frac{{\bf A}}{d}\right)\end{equation}
\begin{equation}{\bf H}_{HP_{\vec{\alpha}}}=\prod_{n=1}^{Tk}\left(\beta_n{\bf I}_N-\frac{{\bf A}}{d}\right)\end{equation}
noting that the new indices incorporate multiplicities, as follows
\[{\bf H}_{LP_{\vec{\alpha}}}=\sum_{i=0}^{Tk} s_i \left(\frac{{\bf A}}{d}\right)^i\]
\[{\bf H}_{HP_{\vec{\alpha}}}=(-1)^{Tk}\sum_{i=0}^{Tk} (-1)^{i+Tk} s_i \left(\frac{{\bf A}}{d}\right)^i=\sum_{i=0}^{Tk} (-1)^{i} s_i \left(\frac{{\bf A}}{d}\right)^i\]
where the coefficients $s_i$ are the elementary symmetric polynomials $e_n(\beta_1,\dots,\beta_{Tk})$ in $\beta_n$:
\[s_0=e_{Tk}(\beta_1,\dots,\beta_{Tk})=\beta_1\beta_2\dots \beta_{Tk}\]
\[s_k=...\]
\[s_{Tk-1}=e_1(\beta_1,\dots,\beta_{Tk})=\beta_1+\beta_2+\dots+ \beta_{Tk}\]
\[s_{Tk}=e_0(\beta_1,\dots,\beta_{Tk})=1.\]
We need to prove that the filterbank  
\[\sum_{j\in\mathbb{Z}} s_{2j}\left(\frac{{\bf A}}{d}\right)^{2j}+{\bf K}\sum_{j\in\mathbb{Z}} s_{2j+1}\left(\frac{{\bf A}}{d}\right)^{2j+1}\]
with diagonal downsampling matrix ${\bf K}$ is invertible by showing that its nullspace is empty. Similarly, as in $A.1$, we assume the contrary and let ${\bf z}={\bf V}{\bf r}$ lie in its nullspace, where ${\bf V}{\bf \Gamma}^j {\bf V}^H=\left(\frac{{\bf A}}{d}\right)^j$, such that
\begin{equation}
\left(\sum_{j\in\mathbb{Z}} s_{2j}\left(\frac{{\bf A}}{d}\right)^{2j}+{\bf K}\sum_{j\in\mathbb{Z}} s_{2j+1}\left(\frac{{\bf A}}{d}\right)^{2j+1}\right){\bf V}{\bf r}={\bf 0}_N\end{equation}
\begin{equation}\Leftrightarrow ||{\bf V}\sum_{j\in\mathbb{Z}} s_{2j}{\bf \Gamma}^{2j}{\bf r}||_2^2=||{\bf K}{\bf V}\sum_{j\in\mathbb{Z}} s_{2j+1}{\bf \Gamma}^{2j+1}{\bf r}||_2^2\end{equation}
where Eq. $(A.8)$ results from rearranging Eq. $(A.7)$ and taking norms of both sides, and gives rise to
\[\sum_{i=0}^{N-1} r(i)^2 \left( \left(\sum_{j\in\mathbb{Z}} s_{2j} \gamma_i^{2j}\right)^2- \left(\sum_{j\in\mathbb{Z}} s_{2j+1} \gamma_i^{2j+1}\right)^2\right)=0,\enskip\text{and hence}\]
\begin{equation}\sum_{i=0}^{N-1} r(i)^2 \prod_{n=1}^{Tk}(\beta_n-\gamma_i)\prod_{n=1}^{Tk}(\beta_n+\gamma_i)=\sum_{i=0}^{N-1} r(i)^2 \prod_{n=1}^{T}(\beta_n^2-\gamma_i^2)^k=0.\end{equation}
Given parameters $\beta_n$ such that $|\beta_n|\leq 1,\enskip n=1,...,T$, and with eigenvalues satisfying $|\gamma_i|\leq 1, \enskip i=0,...,N-1$ by the Perron-Frobenius Theorem \cite{frob}, we assume $|\beta_n|\neq |\gamma_i|$. Thus, all summands in Eq. $(A.9)$ need to be of the same sign to guarantee ${\bf r}={\bf 0}_N$. As the function $f(\gamma_i)=\prod_{n=1}^T(\beta_n^2-\gamma_i^2)^k$, for spectrum ${\bf \gamma}=\{\gamma_i\}_{i=0}^{N-1}$, does not have exclusively positive or negative range for odd $k$, we require $k\in 2\mathbb{N}$. Furthermore, all terms remain of the same sign at any $k$ as long as parameters $\beta_n$ and $T$ are suitably chosen. This is a sufficient condition for guaranteeing invertibility at any downsampling pattern. \\
\\
If, for some $n$, we have $|\beta_n|=|\gamma_i|$ with $i\in\lbrack 0\enskip N-1\rbrack$, giving $|r(i)|\geq0$, we can show that for certain downsampling patterns, the transform continues to be invertible. In particular, this is the case when parameter $\alpha$ in $\beta_n=\frac{\tilde{d}_{\alpha}}{d}$ is such that $\alpha=\frac{2\pi k}{N}$ for some $k\in\lbrack 0 \quad N-1\rbrack$, i.e. ${\bf H}_{HP_{\alpha}}$ annihilates the $k$-th (eigen-)vector in the DFT matrix. For eigenvalues $\lambda_k=\sum_{j=1}^M 2 d_j \cos \left(\frac{2\pi j k}{N}\right)$ of the non-normalized symmetric and circulant adjacency matrix ${\bf A}$ with first row $\lbrack 0 \quad d_1\quad d_2 \dots d_2\quad d_1\rbrack$, we thus have $\tilde{d}_{\alpha}=\lambda_k$ for some $k\in \lbrack 0 \quad N-1\rbrack$. We proceed to show that the filterbank at hand is invertible for such $\alpha$ as long as downsampling is conducted with respect to $s=1\in S$, and more generally, when at least $m_i$ suitably chosen low-pass components are retained (for multiplicity $m_i$ of $\gamma_i$), and by extension, $\sum_{i=1}^P m_i$ components for $P$ distinct $\beta_n$ that satisfy $|\beta_n|=|\gamma_i|$.\\
\\
Assuming wlog $\beta_n=\gamma_i$, for $P$ distinct eigenvalues ($1\leq P< N$), each of multiplicity $m_i$ with corresponding eigenvector(s) $\{{\bf v}_{i,l}\}_{l=0}^{m_i-1}$, we consider, for the case of a non-bipartite graph, the following nullspace representation

\[{\bf z}=\sum_{n=1}^{P} \sum_{l=0}^{m_n-1} r_{n,l} {\bf v}_{n,l},\quad \forall n, m_n\geq 1,\]
where index $n$ signifies distinct eigenvalues (or $\beta_n$) and $r_{n,l}$ are scalar coefficients. With $\left(\frac{{\bf A}}{d}\right)^k {\bf v}_i=\gamma_i^k {\bf v}_i$, we obtain via substitution into Eq. $(A.7)$
\[\sum_{j\in\mathbb{Z}} s_{2j} \sum_{n=1}^{P} \sum_{l=0}^{m_n-1}  r_{n,l} \beta_n^{2j} {\bf v}_{n,l}+{\bf K}\sum_{j\in\mathbb{Z}} s_{2j+1}\sum_{n=1}^{P} \sum_{l=0}^{m_n-1}  r_{n,l} \beta_n^{2j+1} {\bf v}_{n,l}={\bf 0}_N.\]
If $K_{i,i}=1, \enskip i=0,...,N-1$, then
\[ \sum_{n=1}^{P} \sum_{l=0}^{m_n-1} \prod_{q=1}^{T} (\beta_q+\beta_n)^k r_{n,l} {\bf v}_{n,l}={\bf 0}_N\]
and if $\forall i, K_{i,i}=-1, \enskip i=0,...,N-1$, then
\[ \sum_{n=1}^{P} \sum_{l=0}^{m_n-1} \prod_{q=1}^{T} (\beta_q-\beta_n)^k r_{n,l} {\bf v}_{n,l}={\bf 0}_N\]
where in the case of the latter we observe that $\prod_{q=1}^T (\beta_q-\beta_n)^k$ is always zero since $\beta_q=\beta_n$ for $q=n$. Therefore, we need the number of low-pass components to be greater than or equal to the sum of multiplicities $m_n$ for all $P$ distinct eigenvalues  and, in addition, their (node) locations $D$ need to be suitably chosen to facilitate linearly independent partitions ${\bf v}_{n,l}(D)$ such that $r_{n,l}=0$. Note that for $\beta_n=-\gamma_i$ and $-\gamma_i\notin\gamma$, the opposite is the case, i.e. we need at least $\sum_{n=1}^P m_n$ suitably chosen high-pass components to facilitate linear independence of partitions $\{{\bf v}_{n,l}(D^{\complement})\}_{n,l}$. If both $\pm\gamma_i$ exist, similar reasoning as for the following bipartite case applies.\\
For bipartite graphs, let $m_i$ denote the multiplicity of eigenvalue $\gamma_i$ and $-\gamma_i$ respectively (due to symmetry of the spectrum \cite{chung}), and $\{{\bf v}_{i,l}\}_{l=0}^{m_i-1}$ and $\left\{\begin{bmatrix} {\bf v}_{i,l}^B\\ -{\bf v}_{i,l}^{B^{\complement}}\end{bmatrix}\right\}_{l=0}^{m_i-1}$ the corresponding eigenvectors, resulting in a nullspace representation of the form:\\
\[{\bf z}=\sum_{n=1}^P \sum_{l=0}^{m_n-1} \left(r_{n,l} {\bf v}_{n,l}+\tilde{r}_{n,l}\begin{bmatrix} {\bf v}_{n,l}^B\\ -{\bf v}_{n,l}^{B^{\complement}}\end{bmatrix}\right)\quad m_n\geq 1,\]
whose substitution into Eq. $(A.7)$ yields
\[\sum_{j\in\mathbb{Z}} s_{2j} \sum_{n=1}^{P} \sum_{l=0}^{m_n-1}  r_{n,l} \beta_n^{2j} {\bf v}_{n,l}+{\bf K}\sum_{j\in\mathbb{Z}} s_{2j+1}\sum_{n=1}^{P} \sum_{l=0}^{m_n-1}  r_{n,l} \beta_n^{2j+1} {\bf v}_{n,l}\]

\[+\sum_{j\in\mathbb{Z}} s_{2j} \sum_{n=1}^{P} \sum_{l=0}^{m_n-1}  \tilde{r}_{n,l} \beta_n^{2j} \begin{bmatrix} {\bf v}_{n,l}^B\\ -{\bf v}_{n,l}^{B^{\complement}}\end{bmatrix}-{\bf K}\sum_{j\in\mathbb{Z}} s_{2j+1}\sum_{n=1}^{P} \sum_{l=0}^{m_n-1}  \tilde{r}_{n,l} \beta_n^{2j+1} \begin{bmatrix} {\bf v}_{n,l}^B\\ -{\bf v}_{n,l}^{B^{\complement}}\end{bmatrix}={\bf 0}_N.\]
Thus, for $K_{i,i}=1, \enskip i=0,...,N-1$, we obtain
\[ \sum_{n=1}^{P} \sum_{l=0}^{m_n-1} \prod_{q=1}^T (\beta_q+\beta_n)^k r_{n,l} {\bf v}_{n,l}+\sum_{n=1}^{P} \sum_{l=0}^{m_n-1} \prod_{q=1}^T (\beta_q-\beta_n)^k \tilde{r}_{n,l}  \begin{bmatrix} {\bf v}_{n,l}^B\\ -{\bf v}_{n,l}^{B^{\complement}}\end{bmatrix}={\bf 0}_N\]
and for $K_{i,i}=-1, \enskip i=0,...,N-1$
\[ \sum_{n=1}^{P} \sum_{l=0}^{m_n-1} \prod_{q=1}^T (\beta_q-\beta_n)^k r_{n,l}  {\bf v}_{n,l}+\sum_{n=1}^{P} \sum_{l=0}^{m_n-1} \prod_{q=1}^T (\beta_q+\beta_n)^k \tilde{r}_{n,l} \begin{bmatrix} {\bf v}_{n,l}^B\\ -{\bf v}_{n,l}^{B^{\complement}}\end{bmatrix}={\bf 0}_N.\]
Hence, we require at least $m_n$ low-and at most $N-m_n$ high-pass components per $\beta_n$ at suitably chosen locations $D$ and $D^{\complement}$ such that the corresponding partitions of $\{{\bf v}_{n,l}(D)\}_{l=0}^{m_n-1}$ and $\left\{\begin{bmatrix} {\bf v}_{n,l}^B\\ -{\bf v}_{n,l}^{B^{\complement}}\end{bmatrix}({D^{\complement}})\right\}_{l=0}^{m_n-1}$ are linearly independent, leading to $r_{n,l}=0$ and $\tilde{r}_{n,l}=0$. In general, we need the number of retained low-pass components to be $\sum_{n=1}^P m_n\leq N/2 $ and $D$ such that the above partitions form linearly independent sets for $n=1,...,P$. \\
\\
Consider the relevant case, when we downsample by $2$ w.r.t. $s=1$ such that $D=(0:2:N-1)$ (corresponding to $D=B$ in the bipartite case) and $|D|=|D^{\complement}|=N/2$. When eigenbasis ${\bf V}$ is represented as the $N\times N$-DFT matrix, the relation ${\bf V}_{D,0:N-1}=\lbrack \tilde{{\bf V}} \tilde{{\bf V}} \rbrack$ holds, where $\tilde{{\bf V}}$ denotes the $N/2\times N/2$ DFT-matrix. Thus, the transform is invertible for $\beta_n=\gamma_i$, if corresponding eigenvalue(s) $\gamma_i$ with multiplicity $m_i$ are suitably located in the DFT-ordered spectrum $\gamma_{DFT}=\{\gamma_i\}_{i=0}^{N-1}$ such that the associated eigenvectors $\{{\bf v}_{i,l}\}_{l=0}^{m_i-1}$ remain linearly independent after downsampling, i.e. their pairwise column positions $(j,j')$ in ${\bf V}_{D,0:N-1}$ are not of the form $(j, N/2+j)$. Equivalently, the above can be extended for the bipartite case when $-\gamma_i$ exists, since we also have ${\bf V}_{D^{\complement},0:N-1}=\lbrack \tilde{{\bf V}} \tilde{{\bf V}} \rbrack$ up to a normalization constant per column.
\\
A special case occurs, when $\beta=\tilde{d}=0$; we need to show that ${\bf K}\frac{{\bf A}}{d}{\bf z}={\bf 0}_N$ as long as $\sum_{i=0}^{N-1}r(i)^2 \gamma_i^2=0$. The latter yields $r(i)=0$, except when $\gamma_i=0$. Since, however, the eigenvector(s) for $\gamma_i=0$ lie in the nullspace of $\frac{{\bf A}}{d}$, we have $r(0)\neq 0$, thus the filterbank is not invertible for any downsampling pattern, including the case when downsampling is conducted with respect to $s=1$. When $\frac{{\bf A}}{d}$ on the other hand is invertible with $\gamma_i\neq 0$, so is the filterbank. \end{proof}

\subsection{}
\begin{proof}[Proof of Corollary 3.3]
Consider the simplest case with $k=1,\enskip T=1$ and one parameter $\beta$: we need to show that the nullspace ${\bf z}$ of low-pass filter ${\bf H}_{LP_{\alpha}}$ in
 \begin{equation}\frac{1}{2}\left(\beta{\bf I}_N+\frac{{\bf A}}{d}\right){\bf z}={\bf 0}_N\end{equation}
is empty by contradiction. In a similar fashion as in previous proofs, we let ${\bf z}={\bf V}{\bf r}$, with $\frac{{\bf A}}{d}={\bf V}{\bf \Gamma}{\bf V}^H$ and obtain

\[||\beta {\bf V}{\bf r} + {\bf V}{\bf \Gamma}{\bf r}||_2^2=0\]
\[\Leftrightarrow\sum_{i=0}^{N-1} r(i)^2(\beta^2+2\beta\gamma_i+\gamma_i^2)=\sum_{i=0}^{N-1} r(i)^2(\beta+\gamma_i)^2=0.\]
Hence, it follows from inspection that the low-pass filter is invertible unless $\beta=-\gamma_i$. If $-\gamma_i\in\gamma$, similar reasoning as for the bipartite case applies. 
In the case of a bipartite graph, where $|\beta|=|\gamma_i|$ and $\gamma_i,-\gamma_i$ of respective multiplicity $m_i$ exist, with eigenvectors \[{\bf z}=\sum_{l=0}^{m_i-1} \left(r(l) {\bf v}_l+\tilde{r}(l)\begin{bmatrix} {\bf v}_l^B\\ -{\bf v}_l^{B^{\complement}}\end{bmatrix}\right),\] we observe after substitution into Eq. $(A.10)$ that
\[\sum_{l=0}^{m_i-1} r(l)(\beta{\bf v}_l +\gamma_i{\bf v}_l)+\sum_{l=0}^{m_i-1} \tilde{r}(l)\left(\beta \begin{bmatrix} {\bf v}_l^B\\ -{\bf v}_l^{B^{\complement}}\end{bmatrix} -\gamma_i\begin{bmatrix} {\bf v}_l^B\\ -{\bf v}_l^{B^{\complement}}\end{bmatrix}\right)={\bf 0}_N.\]
Due to spectral folding, one eigenvector-set always cancels out for $\beta=\pm \gamma_i$, so that we cannot guarantee zero coefficients, and hence invertibility. By extension, ${\bf H}_{LP_{\alpha}}^k$ is invertible, while ${\bf H}_{LP_{\vec{\alpha}}}^k$ requires invertibility of each individual factor ${\bf H}_{LP_{{\alpha_n}}}$ for parameters $\beta_n$, under the above.\\
\end{proof}



\bibliographystyle{elsarticle-num} 
\bibliography{sample2}

\begin{thebibliography}{10}
\expandafter\ifx\csname url\endcsname\relax
  \def\url#1{\texttt{#1}}\fi
\expandafter\ifx\csname urlprefix\endcsname\relax\def\urlprefix{URL }\fi
\expandafter\ifx\csname href\endcsname\relax
  \def\href#1#2{#2} \def\path#1{#1}\fi

\bibitem{shu}
D.~I. Shuman, S.~K. Narang, P.~Frossard, A.~Ortega, P.~Vandergheynst, The
  emerging field of signal processing on graphs: {Extending} high-dimensional
  data analysis to networks and other irregular domains, IEEE Signal Process.
  Mag. 30~(3) (2013) 83--98.

\bibitem{chung}
F.~R.~K. Chung, Spectral Graph Theory, American Mathematical Society, 1997.

\bibitem{moura}
A.~Sandryhaila, J.~M.~F. Moura, Discrete signal processing on graphs: Frequency
  analysis, {IEEE} Transactions on Signal Processing 62~(12) (2014) 3042--3054.

\bibitem{Coifman}
R.~Coifman, M.~Maggioni, Diffusion wavelets, Applied and Computational Harmonic
  Analysis 21~(1) (2006) 53--94.

\bibitem{ortega2}
S.~K. Narang, A.~Ortega, Perfect reconstruction two-channel wavelet filter
  banks for graph structured data, Signal Processing, IEEE Transactions on
  60~(6) (2012) 2786--2799.

\bibitem{ortega3}
S.~K. Narang, A.~Ortega, Compact support biorthogonal wavelet filterbanks for
  arbitrary undirected graphs, Signal Processing, IEEE Transactions on 61~(19)
  (2013) 4673--4685.
\newblock \href {http://dx.doi.org/10.1109/TSP.2013.2273197}
  {\path{doi:10.1109/TSP.2013.2273197}}.

\bibitem{spectral}
D.~K. Hammond, P.~Vandergheynst, R.~Gribonval, Wavelets on graphs via spectral
  graph theory, Applied and Computational Harmonic Analysis 30~(2) (2011) 129
  -- 150.

\bibitem{ricaud}
B.~Ricaud, D.~I. Shuman, P.~Vandergheynst, On the sparsity of wavelet
  coefficients for signals on graphs, in: SPIE Optical Engineering+
  Applications, International Society for Optics and Photonics, 2013, pp.
  88581L--88581L.

\bibitem{tight}
B.~Dong, Sparse representation on graphs by tight wavelet frames and
  applications, Preprint 2015.

\bibitem{mult}
D.~I. Shuman, M.~J. Faraji, P.~Vandergheynst, A multiscale pyramid transform
  for graph signals, arXiv preprint arXiv:1308.4942.

\bibitem{ekambaram1}
V.~N. Ekambaram, G.~C. Fanti, B.~Ayazifar, K.~Ramchandran, Circulant structures
  and graph signal processing, in: {IEEE} International Conference on Image
  Processing, {ICIP} 2013, 2013, pp. 834--838.
\newblock \href {http://dx.doi.org/10.1109/ICIP.2013.6738172}
  {\path{doi:10.1109/ICIP.2013.6738172}}.

\bibitem{ekambaram2}
V.~N. Ekambaram, G.~Fanti, B.~Ayazifar, K.~Ramchandran, Critically-sampled
  perfect-reconstruction spline-wavelet filterbanks for graph signals, in:
  {IEEE} Global Conference on Signal and Information Processing, GlobalSIP
  2013, 2013, pp. 475--478.
\newblock \href {http://dx.doi.org/10.1109/GlobalSIP.2013.6736918}
  {\path{doi:10.1109/GlobalSIP.2013.6736918}}.

\bibitem{spline}
M.~Unser, Ten good reasons for using spline wavelets, in: Proc. SPIE vol. 3169,
  Wavelet Applications in Signal and Image Processing V, 1997, pp. 422--431.

\bibitem{espline}
M.~Unser, T.~Blu, Cardinal exponential splines: Part i---theory and filtering
  algorithms, IEEE Trans. Signal Process 53 (2005) 1425--1438.

\bibitem{splines}
M.~Unser, Splines: A perfect fit for signal and image processing, Signal
  Processing Magazine, IEEE 16~(6) (1999) 22--38.

\bibitem{spie}
M.~S. Kotzagiannidis, P.~L. Dragotti, Higher-order graph wavelets and sparsity
  on circulant graphs, in: SPIE Optical Engineering+ Applications, Vol. 9597,
  International Society for Optics and Photonics, 2015, pp. 95971E--95971E--9.
\newblock \href {http://dx.doi.org/10.1117/12.2192003}
  {\path{doi:10.1117/12.2192003}}.

\bibitem{icassp}
M.~S. Kotzagiannidis, P.~L. Dragotti, The graph fri-framework- spline wavelet
  theory and sampling on circulant graphs, in: IEEE International Conference on
  Acoustics, Speech and Signal Processing (ICASSP 2016), 2016, pp. 6375--6379.

\bibitem{dict}
D.~Thanou, D.~I. Shuman, P.~Frossard, Parametric dictionary learning for graph
  signals, in: {IEEE} Global Conference on Signal and Information Processing,
  GlobalSIP 2013, Austin, TX, USA, December 3-5, 2013, 2013, pp. 487--490.
\newblock \href {http://dx.doi.org/10.1109/GlobalSIP.2013.6736921}
  {\path{doi:10.1109/GlobalSIP.2013.6736921}}.

\bibitem{acha2}
M.~S. Kotzagiannidis, P.~L. Dragotti, Sampling and {R}econstruction of {S}parse
  {S}ignals on {C}irculant {G}raphs - {A}n {I}ntroduction to {G}raph-{FRI},
  Appl. Comput. Harmon. Anal. (2017),
  http://dx.doi.org/10.1016/j.acha.2017.10.003, in press, available on arXiv:
  arXiv:1606.08085.

\bibitem{circul}
D.~Geller, I.~Kra, S.~Popescu, S.~Simanca, On circulant matrices, Preprint,
  Stony Brook University.

\bibitem{Ekambaram3}
V.~Ekambaram, Graph structured data viewed through a fourier lens, Ph.D.
  thesis, EECS Department, University of California, Berkeley (Dec 2013).

\bibitem{Kron}
F.~Dorfler, F.~Bullo, Kron reduction of graphs with applications to electrical
  networks, Circuits and Systems I: Regular Papers, IEEE Transactions on 60~(1)
  (2013) 150--163.

\bibitem{vanla}
N.~Sharon, Y.~Shkolnisky, A class of laplacian multiwavelets bases for
  high-dimensional data, Applied and Computational Harmonic Analysis 38~(3)
  (2015) 420 -- 451.
\newblock \href
  {http://dx.doi.org/http://dx.doi.org/10.1016/j.acha.2014.07.002}
  {\path{doi:http://dx.doi.org/10.1016/j.acha.2014.07.002}}.

\bibitem{strang}
G.~Strang, G.~Fix, A fourier analysis of the finite element variational method,
  in: Constructive Aspects of Functional Analysis, Rome, Italy: Edizioni
  Cremonese, 1973, pp. 795--840.

\bibitem{vet}
P.~L. Dragotti, M.~Vetterli, T.~Blu, Sampling moments and reconstructing
  signals of finite rate of innovation: Shannon meets strang-fix, IEEE Trans.
  Signal Process. 55~(5) (2007) 1741--1757.
\newblock \href {http://dx.doi.org/10.1109/TSP.2006.890907}
  {\path{doi:10.1109/TSP.2006.890907}}.

\bibitem{signless2}
D.~Cvetkovi{\'c}, P.~Rowlinson, S.~K. Simi{\'c},
  \href{http://www.sciencedirect.com/science/article/pii/S0024379507000316}{{Signless
  Laplacians of finite graphs}}, Linear Algebra and its Applications 423~(1)
  (2007) 155 -- 171.
\newblock \href {http://dx.doi.org/http://dx.doi.org/10.1016/j.laa.2007.01.009}
  {\path{doi:http://dx.doi.org/10.1016/j.laa.2007.01.009}}.
\newline\urlprefix\url{http://www.sciencedirect.com/science/article/pii/S0024379507000316}

\bibitem{signless1}
D.~Cvetkovi{\'c}, P.~Rowlinson, S.~Simi{\'c}, {Eigenvalue Bounds for the
  Signless Laplacian}, Publications de l'Institut Math{\'e}matique 81(95)~(101)
  (2007) 11--27.

\bibitem{esplinewav}
C.~Vonesch, T.~Blu, M.~Unser, Generalized daubechies wavelet families, Signal
  Processing, IEEE Transactions on 55~(9) (2007) 4415--4429.
\newblock \href {http://dx.doi.org/10.1109/TSP.2007.896255}
  {\path{doi:10.1109/TSP.2007.896255}}.

\bibitem{biorref}
M.~Vetterli, J.~Kovacevic, V.~K. Goyal,
  \href{http://www.fourierandwavelets.org/}{{Fourier} and Wavelet Signal
  Processing}, 2014.
\newline\urlprefix\url{http://www.fourierandwavelets.org/}

\bibitem{strangbook}
G.~Strang, T.~Q. Nguyen, Wavelets and {F}ilter banks., Wellesley-Cambridge
  Press, 1997.

\bibitem{esplines2}
M.~Unser, Cardinal exponential splines: part ii - think analog, act digital,
  IEEE Transactions on Signal Processing 53~(4) (2005) 1439--1449.
\newblock \href {http://dx.doi.org/10.1109/TSP.2005.843699}
  {\path{doi:10.1109/TSP.2005.843699}}.

\bibitem{varsplines}
M.~Unser, T.~Blu, Self-similarity: {P}art {I}---{S}plines and operators, {IEEE}
  Transactions on Signal Processing 55~(4) (2007) 1352--1363.

\bibitem{var2}
G.~Micula, \href{http://eudml.org/doc/127983}{A variational approach to spline
  functions theory.}, Rendiconti del Seminario Matematico 61~(3) (2003)
  209--227.
\newline\urlprefix\url{http://eudml.org/doc/127983}

\bibitem{Pesenson}
I.~Pesenson, Variational splines and paley--wiener spaces on combinatorial
  graphs, Constructive Approximation 29~(1) (2008) 1--21.
\newblock \href {http://dx.doi.org/10.1007/s00365-007-9004-9}
  {\path{doi:10.1007/s00365-007-9004-9}}.

\bibitem{green1}
F.~Chung, S.-T. Yau,
  \href{http://www.sciencedirect.com/science/article/pii/S0097316500930942}{Discrete
  green's functions}, Journal of Combinatorial Theory, Series A 91~(1--2)
  (2000) 191 -- 214.
\newblock \href {http://dx.doi.org/http://dx.doi.org/10.1006/jcta.2000.3094}
  {\path{doi:http://dx.doi.org/10.1006/jcta.2000.3094}}.
\newline\urlprefix\url{http://www.sciencedirect.com/science/article/pii/S0097316500930942}

\bibitem{green2}
R.~Ellis, Discrete green's functions for products of regular graphs, arXiv
  preprint math/0309080.

\bibitem{frob}
A.~Berman, R.~J. Plemmons,
  \href{http://opac.inria.fr/record=b1085532}{Nonnegative matrices in the
  mathematical sciences}, Computer science and applied mathematics, Academic
  Press, New York, 1979, includes index.
\newline\urlprefix\url{http://opac.inria.fr/record=b1085532}

\bibitem{chungdir}
F.~Chung, Laplacians and the cheeger inequality for directed graphs, Annals of
  Combinatorics 9 (2005) 1--19.

\bibitem{globalsip}
M.~S. Kotzagiannidis, P.~L. Dragotti, Sparse graph signal reconstruction and
  image processing on circulant graphs, in: {IEEE} GlobalSIP, 2014, pp.
  923--927.
\newblock \href {http://dx.doi.org/10.1109/GlobalSIP.2014.7032255}
  {\path{doi:10.1109/GlobalSIP.2014.7032255}}.

\bibitem{rcm}
E.~Cuthill, J.~McKee, Reducing the bandwidth of sparse symmetric matrices, in:
  Proceedings of the 1969 24th National Conference, ACM '69, ACM, 1969, pp.
  157--172.
\newblock \href {http://dx.doi.org/10.1145/800195.805928}
  {\path{doi:10.1145/800195.805928}}.

\bibitem{cut}
J.~Shi, J.~Malik, Normalized cuts and image segmentation, IEEE Transactions on
  Pattern Analysis and Machine Intelligence 22 (1997) 888--905.

\bibitem{toolbox}
N.~{Perraudin}, J.~{Paratte}, D.~{Shuman}, V.~{Kalofolias}, P.~{Vandergheynst},
  D.~K. {Hammond}, {GSPBOX: A toolbox for signal processing on graphs}, ArXiv
  e-prints\href {http://arxiv.org/abs/1408.5781} {\path{arXiv:1408.5781}}.

\bibitem{handbook}
R.~Hammack, W.~Imrich, S.~Klavzar, Handbook of Product Graphs, Second Edition,
  2nd Edition, CRC Press, Inc., Boca Raton, FL, USA, 2011.

\bibitem{kronapprox}
J.~Kamm, J.~G. Nagy,
  \href{http://www.sciencedirect.com/science/article/pii/S0024379598100241}{Kronecker
  product and \{SVD\} approximations in image restoration}, Linear Algebra and
  its Applications 284~(1--3) (1998) 177 -- 192, international Linear Algebra
  Society (ILAS) Symposium on Fast Algorithms for Control, Signals and Image
  Processing.
\newblock \href
  {http://dx.doi.org/http://dx.doi.org/10.1016/S0024-3795(98)10024-1}
  {\path{doi:http://dx.doi.org/10.1016/S0024-3795(98)10024-1}}.
\newline\urlprefix\url{http://www.sciencedirect.com/science/article/pii/S0024379598100241}

\bibitem{kronnetwork}
J.~Leskovec, D.~Chakrabarti, J.~Kleinberg, C.~Faloutsos, Z.~Ghahramani,
  \href{http://dl.acm.org/citation.cfm?id=1756006.1756039}{Kronecker graphs: An
  approach to modeling networks}, J. Mach. Learn. Res. 11 (2010) 985--1042.
\newline\urlprefix\url{http://dl.acm.org/citation.cfm?id=1756006.1756039}

\bibitem{pits}
N.~P. Pitsianis, The kronecker product in approximation and fast transform
  generation, Ph.D. thesis, Ithaca, NY, USA, uMI Order No. GAX97-16143 (1997).

\bibitem{bigdata}
A.~Sandryhaila, J.~M.~F. Moura, Big data analysis with signal processing on
  graphs: Representation and processing of massive data sets with irregular
  structure, {IEEE} Signal Process. Mag. 31~(5) (2014) 80--90.

\bibitem{loankron}
C.~V. Loan, N.~Pitsianis, Approximation with kronecker products, in: Linear
  Algebra for Large Scale and Real Time Applications, Kluwer Publications,
  1993, pp. 293--314.

\bibitem{lex}
R.~Hoshino, Independence polynomials of circulant graphs, Ph.D. thesis,
  Dalhousie University (2008).

\bibitem{tensor}
J.~C. George, R.~S. Sanders, When is a tensor product of circulant graphs
  circulant?, In eprint arXiv:math/9907119 (07/1999).

\bibitem{laplaceproduct1}
H.~Sayama, Estimation of laplacian spectra of direct and strong product graphs,
  CoRR abs/1507.03030.

\bibitem{laplaceproduct2}
S.~Barik, R.~Bapat, S.~Pati, On the laplacian spectra of product graphs, Appl.
  Anal. Discrete Math. 58 (2015) 9--39.

\bibitem{brouwer12}
A.~E. Brouwer, W.~H. Haemers, Spectra of Graphs, New York, NY, 2012.
\newblock \href {http://dx.doi.org/10.1007/978-1-4614-1939-6}
  {\path{doi:10.1007/978-1-4614-1939-6}}.

\bibitem{prod2}
W.~Imrich, S.~Klav{\v{z}}ar, Product graphs, structure and recognition,
  Vol.~56, Wiley-Interscience, 2000.

\end{thebibliography}


\end{document}